%% file: main.tex
\RequirePackage[l2tabu,orthodox]{nag}
\documentclass
[letterpaper,11pt,]
{article}

\usepackage{etex}
\usepackage{verbatim}
\usepackage{xspace,enumerate}
\usepackage[dvipsnames]{xcolor}
\usepackage[T1]{fontenc}
\usepackage[full]{textcomp}
\usepackage[american]{babel}
\usepackage{mathtools}
\usepackage{amsthm}
\usepackage[
letterpaper,
top=1in,
bottom=1in,
left=1in,
right=1in]{geometry}
\usepackage{newpxtext} %
\usepackage{textcomp} %
\usepackage[varg,bigdelims]{newpxmath}
\usepackage[scr=rsfso]{mathalfa}%
\usepackage{bm} %
\let\mathbb\varmathbb
\usepackage{microtype}
\usepackage[pagebackref,colorlinks=true,urlcolor=blue,linkcolor=blue,citecolor=OliveGreen]{hyperref}
\usepackage[capitalise,nameinlink]{cleveref}
\crefname{lemma}{Lemma}{Lemmas}
\crefname{fact}{Fact}{Facts}
\crefname{theorem}{Theorem}{Theorems}
\crefname{corollary}{Corollary}{Corollaries}
\crefname{claim}{Claim}{Claims}
\crefname{example}{Example}{Examples}
\crefname{algorithm}{Algorithm}{Algorithms}
\crefname{problem}{Problem}{Problems}
\crefname{definition}{Definition}{Definitions}
\crefname{exercise}{Exercise}{Exercises}
\usepackage{amsthm}

\newtheorem{theorem}{Theorem}[section]
\newtheorem*{theorem*}{Theorem}
\newtheorem{lemma}[theorem]{Lemma}
\newtheorem*{lemma*}{Lemma}
\newtheorem{fact}[theorem]{Fact}
\newtheorem*{fact*}{Fact}

\newtheorem*{proposition*}{Proposition}
\newtheorem{corollary}[theorem]{Corollary}
\newtheorem*{corollary*}{Corollary}

\newtheorem*{hypothesis*}{Hypothesis}

\newtheorem*{conjecture*}{Conjecture}
\theoremstyle{definition}
\newtheorem{definition}[theorem]{Definition}
\newtheorem*{definition*}{Definition}

\newtheorem*{construction*}{Construction}

\newtheorem*{example*}{Example}

\newtheorem*{question*}{Question}
\newtheorem{algorithm}[theorem]{Algorithm}
\newtheorem*{algorithm*}{Algorithm}

\newtheorem*{assumption*}{Assumption}

\newtheorem*{problem*}{Problem}

\newtheorem*{openquestion*}{Open Question}
\theoremstyle{remark}

\newtheorem*{claim*}{Claim}
\newtheorem{remark}[theorem]{Remark}
\newtheorem*{remark*}{Remark}

\newtheorem*{observation*}{Observation}
\usepackage{paralist}
\let\originalleft\left
\let\originalright\right
\renewcommand{\left}{\mathopen{}\mathclose\bgroup\originalleft}
\renewcommand{\right}{\aftergroup\egroup\originalright}
\usepackage{turnstile}
\usepackage{mdframed}
\usepackage{tikz}
\usetikzlibrary{positioning}
\usepackage{caption}
\DeclareCaptionType{Algorithm}
\usepackage{newfloat}
\usepackage{array}
\usepackage{subfig}
\usepackage{bbm}
\usepackage{xparse}
\usepackage{amsthm} %
\makeatletter
\let\latexparagraph\paragraph
\RenewDocumentCommand{\paragraph}{som}{%
  \IfBooleanTF{#1}
    {\latexparagraph*{#3}}
    {\IfNoValueTF{#2}
       {\latexparagraph{\maybe@addperiod{#3}}}
       {\latexparagraph[#2]{\maybe@addperiod{#3}}}%
  }%
}
\newcommand{\maybe@addperiod}[1]{%
  #1\@addpunct{.}%
}
\makeatother

\newcommand{\Authornote}[2]{}
\newcommand{\Authornotecolored}[3]{}
\newcommand{\Authorcomment}[2]{}
\newcommand{\Authorfnote}[2]{}

\usepackage{boxedminipage}
\newcommand{\paren}[1]{(#1)}
\newcommand{\Paren}[1]{\left(#1\right)}
\newcommand{\bigparen}[1]{\big(#1\big)}
\newcommand{\Bigparen}[1]{\Big(#1\Big)}
\newcommand{\brac}[1]{[#1]}
\newcommand{\Brac}[1]{\left[#1\right]}

\newcommand{\abs}[1]{\lvert#1\rvert}
\newcommand{\Abs}[1]{\left\lvert#1\right\rvert}

\newcommand{\card}[1]{\lvert#1\rvert}

\newcommand{\set}[1]{\{#1\}}
\newcommand{\Set}[1]{\left\{#1\right\}}
\newcommand{\bigset}[1]{\big\{#1\big\}}

\newcommand{\norm}[1]{\lVert#1\rVert}
\newcommand{\Norm}[1]{\left\lVert#1\right\rVert}

\newcommand{\Normt}[1]{\Norm{#1}_2}

\newcommand{\Snormt}[1]{\Norm{#1}^2_2}

\newcommand{\normo}[1]{\norm{#1}_1}
\newcommand{\Normo}[1]{\Norm{#1}_1}

\newcommand{\normi}[1]{\norm{#1}_\infty}

\newcommand{\iprod}[1]{\langle#1\rangle}
\newcommand{\Iprod}[1]{\left\langle#1\right\rangle}

\newcommand{\Esymb}{\mathbb{E}}
\newcommand{\Psymb}{\mathbb{P}}
\newcommand{\Vsymb}{\mathbb{V}}
\DeclareMathOperator*{\E}{\Esymb}
\DeclareMathOperator*{\Var}{\Vsymb}
\DeclareMathOperator*{\ProbOp}{\Psymb}
\renewcommand{\Pr}{\ProbOp}

\renewcommand{\ij}{{ij}}

\newcommand{\bits}{\{0,1\}}

\newcommand{\seteq}{\mathrel{\mathop:}=}
\newcommand{\from}{\colon}

\newcommand{\Mid}{\nonscript\;\middle\vert\nonscript\;}

\newcommand{\mper}{\,.}

\newcommand\bdot\bullet

\DeclareMathOperator{\Ind}{\mathbf 1}

\DeclareMathOperator{\Tr}{Tr}

\DeclareMathOperator{\poly}{poly}

\DeclareMathOperator{\polylog}{polylog}
\DeclareMathOperator{\supp}{supp}

\DeclareMathOperator{\rank}{rank}

\newcommand{\Hoelder}{H\"{o}lder\xspace}
\newcommand{\Holder}{\Hoelder}

\newcommand{\Z}{\mathbb Z}
\newcommand{\N}{\mathbb N}
\newcommand{\R}{\mathbb R}

\newcommand{\cA}{\mathcal A}
\newcommand{\cB}{\mathcal B}

\newcommand{\cG}{\mathcal G}

\newcommand{\cP}{\mathcal P}
\newcommand{\cQ}{\mathcal Q}

\newcommand{\cX}{\mathcal X}
\newcommand{\cY}{\mathcal Y}
\newcommand{\cZ}{\mathcal Z}

\newcommand{\bbS}{\mathbb S}

\renewcommand{\leq}{\leqslant}
\renewcommand{\le}{\leqslant}
\renewcommand{\geq}{\geqslant}
\renewcommand{\ge}{\geqslant}
\let\epsilon=\varepsilon
\numberwithin{equation}{section}
\newcommand\MYcurrentlabel{xxx}
\newcommand{\MYstore}[2]{%
  \global\expandafter \def \csname MYMEMORY #1 \endcsname{#2}%
}
\newcommand{\MYload}[1]{%
  \csname MYMEMORY #1 \endcsname%
}
\newcommand{\MYnewlabel}[1]{%
  \renewcommand\MYcurrentlabel{#1}%
  \MYoldlabel{#1}%
}
\newcommand{\MYdummylabel}[1]{}
\newcommand{\torestate}[1]{%
  \let\MYoldlabel\label%
  \let\label\MYnewlabel%
  #1%
  \MYstore{\MYcurrentlabel}{#1}%
  \let\label\MYoldlabel%
}
\newcommand{\restatetheorem}[1]{%
  \let\MYoldlabel\label
  \let\label\MYdummylabel
  \begin{theorem*}[Restatement of \cref{#1}]
    \MYload{#1}
  \end{theorem*}
  \let\label\MYoldlabel
}
\newcommand{\restatelemma}[1]{%
  \let\MYoldlabel\label
  \let\label\MYdummylabel
  \begin{lemma*}[Restatement of \cref{#1}]
    \MYload{#1}
  \end{lemma*}
  \let\label\MYoldlabel
}
\newcommand{\restateprop}[1]{%
  \let\MYoldlabel\label
  \let\label\MYdummylabel
  \begin{proposition*}[Restatement of \cref{#1}]
    \MYload{#1}
  \end{proposition*}
  \let\label\MYoldlabel
}
\newcommand{\restatefact}[1]{%
  \let\MYoldlabel\label
  \let\label\MYdummylabel
  \begin{fact*}[Restatement of \cref{#1}]
    \MYload{#1}
  \end{fact*}
  \let\label\MYoldlabel
}
\newcommand{\restate}[1]{%
  \let\MYoldlabel\label
  \let\label\MYdummylabel
  \MYload{#1}
  \let\label\MYoldlabel
}

\newcommand{\e}{\epsilon}
\newcommand{\eps}{\epsilon}

\allowdisplaybreaks
\newcommand*{\Id}{\mathrm{Id}}

\newcommand*{\normf}[1]{\norm{#1}_{\mathrm{F}}}
\newcommand*{\Normf}[1]{\Norm{#1}_{\mathrm{F}}}

\newenvironment{algorithmbox}{\begin{mdframed}[nobreak=true]\begin{algorithm}}{\end{algorithm}\end{mdframed}}

\newcommand{\normm}[1]{\norm{#1}_\textnormal{max}}

\renewcommand{\normo}[1]{\norm{#1}_{\textnormal{sum}}}
\renewcommand{\Normo}[1]{\Norm{#1}_{\textnormal{sum}}}

\providecommand{\normsum}[1]{\norm{#1}_{\textnormal{sum}}}
\providecommand{\Normsum}[1]{\Norm{#1}_{\textnormal{sum}}}

\renewcommand{\normi}[1]{\norm{#1}_{\max}}

\providecommand{\normmax}[1]{\norm{#1}_{\max}}

\providecommand{\Yin}{Y_{\mathrm{in}}}
\providecommand{\bbG}{\mathbb G}
\newcommand{\Score}{\text{Score} }

\newcommand{\dsdistance}{\delta_{\text{ds}}}

\providecommand{\normstar}[1]{\norm{#1}_{*}}
\providecommand{\Znull}{{Z_0}}
\providecommand{\Qnull}{Q_0}

\providecommand{\Bhat}{\hat B}
\providecommand{\Yhat}{\hat Y}
\providecommand{\deltatwosquared}{\delta_2^2}
\providecommand{\lbdeltatwo}{\underline{\delta_2}}
\providecommand{\Bnull}{{B_0}}
\providecommand{\Znulltrans}{Z_0^{\top}}

\providecommand{\Vnull}{V_0}
\providecommand{\Vnulltrans }{{V_0^{\top}}}
\providecommand{\SBM}{\text{SBM}}
\newcommand{\tE}{\mathop{\tilde{\mathbb{E}}}}

\newcommand{\dnode}{d_{\mathrm{node}}}
\newcommand{\fsos}{f_{\mathrm{sos}}}
\newcommand{\hatfsos}{\hat{f}_{\mathrm{sos}}}
\providecommand{\Lap}{\text{Lap}}

\newcommand*{\transpose}[1]{{#1}{}^{\mkern-1.5mu\mathsf{T}}}

\title{
  Private graphon estimation via sum-of-squares
}

\author{
    Hongjie Chen\thanks{ETH Z\"urich.}
  \and
    Jingqiu Ding\thanks{ETH Z\"urich.}
    \and
    Tommaso d'Orsi\thanks{BIDSA, Bocconi University.
  }
  \and
  Yiding Hua\thanks{ETH Z\"urich.
  }
    \and
    Chih-Hung Liu\thanks{National Taiwan University. Part of the work is done while the author was visiting ETH Z\"urich.
  }
    \and
    David Steurer\thanks{ETH Z\"urich.
  }
}

\begin{document}

\pagestyle{empty}

\maketitle
\thispagestyle{empty} %

\begin{abstract}

  \input{content/abstract}

\end{abstract}

\clearpage

\microtypesetup{protrusion=false}
\tableofcontents{}
\microtypesetup{protrusion=true}

\clearpage

\pagestyle{plain}
\setcounter{page}{1}

\input{content/introduction}

\input{content/results}

\input{content/techniques}

\input{content/preliminaries}

\input{content/mainalgorithm}

\input{content/sosidentifiability}

\input{content/doublystochastic}

\input{content/rounding}

\input{content/soslipschitz}

\input{content/Acknowledgement}

\phantomsection
\addcontentsline{toc}{section}{References}
\bibliographystyle{amsalpha}
\bibliography{bib/mathreview,bib/dblp,bib/custom,bib/scholar}

\appendix

\input{content/subsample}

\input{content/densityestimation}

\input{content/robustgraphonestimation}

\input{content/backgroundsos}

\input{content/DistanceMetric}

\input{content/lowerbound}

\input{content/theoremsfromBCS}

\input{content/concentration}

\end{document}

%% file: content/abstract.tex
We develop the first pure node-differentially-private algorithms for learning stochastic block models and for graphon estimation with polynomial running time for any constant number of blocks.
The statistical utility guarantees match those of the previous best information-theoretic (exponential-time) node-private mechanisms for these problems.
The algorithm is based on an exponential mechanism for a score function defined in terms of a sum-of-squares relaxation whose level depends on the number of blocks.

The key ingredients of our results are
(1) a characterization of the distance between the block graphons in terms of a quadratic optimization over the polytope of doubly stochastic matrices,
(2) a general sum-of-squares convergence result for polynomial optimization over arbitrary polytopes, and
(3) a general approach to perform Lipschitz extensions of score functions as part of the sum-of-squares algorithmic paradigm.

%% file: content/introduction.tex
\section{Introduction}

Differential privacy is an appealing and intensely studied theoretical framework for privacy preserving data analysis~\cite{dwork2006calibrating, dwork2014algorithmic}.
In the context of graph data, researchers investigate two natural versions of differential privacy:
For \textit{edge-differential privacy}, nodes in the input graph represent public information and only individual relationships, encoded by edges, represent sensitive information that the released data must not reveal.
For \textit{node-differential privacy}, individual nodes themselves (along with all their incident edges) represent sensitive information that we require the algorithm to protect.
(See \cref{definition:node-dp} for a precise definition of node-differential privacy.)

Under edge-differential privacy, it is feasible to release accurate approximations of both global properties of the graph (e.g. the degree distribution)~\cite{nissim2007smooth, hay2009accurate, blocki2012johnson, gupta2012iterative, karwa2014private, xiao2014differentially} and local properties of the nodes (e.g. whether two particular nodes belong to the same community or not){}\cite{mohamed2022differentially, chen2023private}.

On the other hand, node-differential privacy is significantly more stringent.
Here, we can only hope to release global statistics of the graph.
But even for the simplest such properties, private algorithms turn out to be challenging to design~\cite{blocki2013differentially,nodeDP,borgs2015private,MR3631012,borgs2018revealing,ullman2019efficiently,kalemaj2023node}.
The reason is that global statistics often are highly sensitive to modifications of even a single node.
For example, in a $n$-vertex graph, the number of connected components can change by an additive factor of \(\Omega(n)\) when we modify a single vertex.
In contrast, it can change by at most \(1\) when we add or remove a single edge.

\paragraph{Parameter estimation and privacy}
Given the above picture, results \cite{blocki2013differentially, borgs2015private, borgs2018revealing, ullman2019efficiently} in node differential privacy  adheres to the following blueprint:
\textit{(i)} they preserve privacy on all graphs,
\textit{(ii)} they consider a specific graph model,
\textit{(iii)} they argue how, from typical instances of that model, they  can accurately estimate its parameters.
The quality of the result then depends on the accuracy of the estimation, the significance of the models considered and the running time of the underlying algorithm.

In terms of (statistical) utility guarantees and privacy trade-offs, one of the best known node-private mechanisms is the mechanism for graphon estimation (closely related to parameter-learning for stochastic block model, see \cref{sec:graphon-preliminaries} for a definition) of \cite{borgs2015private, borgs2018revealing}.
This mechanism summarizes the global structure of the input graph as a small \(k\)-by-\(k\) matrix, whose entries represent connectivity parameters between different regions of the input graph.
A major drawback of this mechanism is its high computational cost: for every value of \(k\), it takes time exponential in the size of the input graph.
So far researchers have been able to turn this mechanism into a computationally efficient one only for edge-density estimation, roughly corresponding to the case \(k=1\) \cite{ullman2019efficiently}.
In this work, we show that for\textit{ every} value of \(k\), it is possible to achieve the utility and privacy guarantees of \cite{borgs2015private, borgs2018revealing} in time polynomial in the size of the input graph (but exponential in the output parameter \(k\)).

To describe the utility guarantees of our polynomial-time node-private algorithms, we first consider a slightly simplified version of the general stochastic block model with perfectly balanced block sizes.
Since our algorithms turn out to be robust against various kinds of errors (e.g. due to node corruptions or model misspecification), our utility guarantees for this simplified model extend to a wide range of models including those the mechanism of \cite{borgs2015private, borgs2018revealing} has previously been analyzed for.

\begin{definition}[Balanced stochastic block model]
	\label{definition:sbm-balanced}
	Let \(n,d,k \ge 2\) with \(k\) an integer, \(n\) a multiple of \(k\), and \(d \le n\).
	For a symmetric \(k\)-by-\(k\) matrix \(\Bnull\) whose entries are nonnegative and average to \(1\), we consider the following distribution over \(n\)-vertex graphs, called the \emph{\((\Bnull, d, n)\)-block model}\footnote{Since the parameter $k$ will always be clear from the context, we omit it in our notation.}:
	\begin{enumerate}
		\item Partition the vertex set \([n]\) uniformly at random into \(k\) parts of equal size.
		\item Then for every pair of parts \(a,b\in [k]\) and every pair of vertices \(i,j\) in part~\(a\) and part~\(b\), connect vertices \(i\) and \(j\) by an edge in \(\bm G\) independently at random with probability \(\tfrac d n \cdot \Bnull(a,b)\).
	\end{enumerate}
\end{definition}

Up to the balancedness of the partition, \cref{definition:sbm-balanced} is the canonical definition of stochastic block models \cite[Section 2.1]{sbmbook}.
Since the entries of \(\Bnull\) average to \(1\), every vertex has expected degree \(d\) in \(\bm G\).
Our main objective is to privately recover $\Bnull$ from the observed graph $\bm G\,.$

\paragraph{Distance measure for block-connectivity matrices}
For our objective, the first challenge comes from the observation that $\Bnull$ is \textit{not} identifiable.
Indeed, a simple permutation of $\Bnull$ may yield a very different matrix but lead to the same graph distribution. Following closely the existing literature \cite{wolfe2013nonparametric, borgs2015private, borgs2018revealing, MR3611494, mcmillan2018non, MR4319235}, we define a distance measure for \(k\)-by-\(k\) matrices \(B_1,B_2\in \R^{k\times k}\) that is invariant under permutations acting on the rows and columns of the matrices,
\begin{align}\label{eq:sbm-distance}
	\delta_2 (B_1,B_2) \seteq \min_{\substack{\phi_1,\phi_2 \from [0,1]\to [k]\\ \text{measure preserving}}} \Norm{B_1^{\phi_1} - B_2^{\phi_2}}_2\,.
\end{align}
Here \(B^\phi\from [0,1]\times [0,1]\to \R\) denotes the function \((x,y)\mapsto B(\phi(x),\phi(y))\), where \(B\in \R^{k\times k}\) and \(\phi\from [0,1]\to [k]\).
We say \(\phi\from [0,1]\to [k]\) is measure preserving if all preimages \(\phi^{-1}(1),\ldots,\phi^{-1}(k)\subseteq [0,1]\) have measure \(1/k\).

This definition now allows us to raise a well-formed question: Given a graph $\bm G$ sampled from the $(B_0,d,n)$-block model, can we find some $\hat{B}(\bm G)$, under node differential privacy, such that $\delta_2(\Bnull, \hat{B}(\bm G))$ is small? If so, can we do it efficiently?

\paragraph{State of the art and computational complexity}
Borgs et al. \cite{borgs2015private, borgs2018revealing} provided the first algorithms recovering the connectivity probability matrix $\Bnull$ with error rates matching, in many regimes, the optimal non-private procedures. 
Unfortunately, their results crucially rely on evaluating an NP-hard function and require time exponential in $n$ even for $k=2$ (see \cref{sec:techniques}). 
Thus it remained a fascinating open question that if these results could be captured by polynomial time algorithms.

In this work, we positively answer this question by introducing the \textit{first pure node differentially private polynomial time algorithm for stochastic block models estimation}. The error convergence of our algorithm is tight and match those of existing inefficient algorithms, up to an information-computation gap \cite{luo2023computational}.

\paragraph{Robustness and graphons}
As already mentioned,  Borgs et al. \cite{borgs2015private, borgs2018revealing} tackled more general models than \cref{definition:sbm-balanced}.
This is because their algorithms are robust to agnostic perturbations.
The introduction of an agnostic error allows one to consider the  non-parametric model of graphons, capturing more complex distributions such as random geometric graphs. 
Indeed many previous (non-private) results extend to these settings \cite{wolfe2013nonparametric, borgs2015private, borgs2018revealing, MR3611494, mcmillan2018non, MR4319235}.
Among (non-private) polynomial time algorithms,  \cite{pmlr-v80-xu18a} provided an algorithm for graphon estimation, albeit with worse error convergence.

Our efficient node differentially private algorithm also naturally applies to the more general settings of graphon estimation, thus capturing the models considered by Borgs et al. \cite{borgs2015private, borgs2018revealing}. Furthermore, our algorithm improves the state of the art \textit{even} in the non-private settings providing  the (conjecturally) optimal error convergence among efficient algorithms.

%% file: content/results.tex
\subsection{Results}

\paragraph{Stochastic block model}

We present the first node differentially private algorithm for learning stochastic block models with polynomial running time for any constant number of blocks.

\begin{theorem}[Learning SBMs with node differential privacy]
	\torestate{\label{thm:mainSBM}
  Consider a graph $\bm G$ sampled from the $(\Bnull, d, n)$-block model defined in \cref{definition:sbm-balanced}. 
  Suppose $\normi{\Bnull}\leq R\leq \sqrt{\frac{n\epsilon}{\polylog(n)}}$ for some $R\in \R_+$, and $\e^4 n^2\geq \polylog(n)$\footnote{We add these minor assumptions such that the error induced by estimating $d$ can be neglected. Particularly we avoid the error term $\frac{1}{d^2\e^2}$ which appears in \cite{borgs2018revealing}}.
  Then there exists an $n^{\poly(k)}$-time $\epsilon$-differentially node private algorithm which, given $\bm G$ and $R$, outputs $\hat{B}(\bm G)\in [0,R]^{k\times k}$ such that with probability at least $1-\frac{1}{d^{100}}-O\Paren{\frac{R}{n}}$, 
  \[
		\deltatwosquared\bigparen{\Bhat(\bm G), \Bnull}\leq  O\Paren{\frac{R k}{d}+\frac{R^2 k^2\log(n)}{n\epsilon}}\,.
  \]
}
\end{theorem}

Notice that the algorithm behind \cref{thm:mainSBM} is parametrized by $R$ and its utility guarantees hold whenever $\Bnull$ has entries bounded by $R$.
In contrast, $d$ is not assumed to be public and it is (approximately) learned by the algorithm from the input data.

In comparison, previous node differentially private algorithms \cite{borgs2015private, borgs2018revealing} run in $\exp(n)$ time even for $k=2$. 
Moreover, the $\exp(n)$ time algorithms by Borgs et al. \cite{borgs2015private, borgs2018revealing} require average degree $d$ at least logarithmic in $n$.
On the other hand, to the best of our knowledge, no previous polynomial time algorithm can match our guarantees even without the privacy requirement.
The non-private algorithm by Xu~\cite{pmlr-v80-xu18a} based on singular value decomposition can achieve error rate $O\Paren{\frac{k}{d}}\,,$ only when $d\geq \polylog(n)$.

Next, we interpret each error term in \cref{thm:mainSBM}(assuming $R=O(1)$). 
Recent results \cite{luo2023computational} provide strong evidence that obtaining better guarantees than $O(\frac{k}{d})$ is inherently hard for polynomial time algorithms. 
Indeed, this error term comes from the (conjectured) information computation gap for stochastic block models.
On the other hand, the error term $\frac{k^2\log n}{n\epsilon}$ due to the privacy requirements matches the guarantees of the previous exponential time algorithms \cite{borgs2015private, borgs2018revealing}. 
We provide a lower bound in \cref{thm:lower_bound} that shows the term $\frac{k^2\log n}{n\epsilon}$ is inherent up to a $\log n$ factor.

It is also interesting to compare our running time and utility guarantees to those one could expect to obtain from the popular subsample-and-aggregate framework~\cite{nissim2007smooth} for private statistical data analysis.
Following this framework, we would randomly split the vertex set into \(T\) parts and estimate \(B_0\) independently in each part using a non-private algorithm (e.g., the one we present \cref{sec:robustness}).
In this way, we obtain \(T\) estimates \(\hat B_1,\ldots,\hat B_T\) for \(B_0\) and the goal is to privately aggregate them to a single output \(\hat B\).
Even testing whether two of the initial estimates are meaningfully close boils down to an optimization problem similar to approximate graph isomorphism (on weighted \(k\)-vertex graphs).
Hence, one would expect algorithms for aggregation to take time exponential in \(k\) (similar to our algorithms).
At the same time, we can argue that the utility guarantees of such a mechanism necessarily is significantly worse compared to our utility guarantees.
Private aggregation must ensure that swapping any one of the \(T\) estimates changes the output distribution multiplicative by at most a factor of \(e^{\e}\), where \(\e\) is the privacy parameter.
Standard packing arguments for the space of possible choices of \(B_0\) shows that we must have \(T \gg k^2 /\e\).
However, since only a \(1/T\) fraction of the edge remains after the splitting operation, the best error bound we can hope for information-theoretically from the remaining set of edge is \(\Omega(T / d)\ge \Omega(\tfrac{k^2}{\e d})\).
Even for \(\e=0.1\), this error is worse than our bound by a factor of \(k\).
More, importantly this error bound would scale directly with the privacy parameter \(\e\), whereas for our algorithm privacy has \emph{no statistical cost} for a wide range of privacy parameters (namely, all privacy parameters \(\e\) satisfying \(\e \ge \tfrac {R k d \log n}{n}\)).

\paragraph{Graphon}

A graphon is a bounded and measurable function $W:[0,1]^2\to \R_{+}$ such that $W(x,y)=W(y,x)$, which is said to be normalized if $\int W=1$.
Given a normalized graphon~$W$ and an edge density parameter~$\rho$, a $(W,\rho,n)$-random graph on $n$ vertices is generated by first picking uniformly at random a value $\bm x_i$ in $[0,1]$ for each vertex $i\in[n]$, and then connecting every pair of vertices $i,j\in[n]$ independently with probability $\rho\cdot W(\bm x_i,\bm x_j)$.
The $\delta_2$ distance between two graphons $W,W'$ is defined to be
\[
  \delta_2(W, W')\seteq  \inf_{\substack{\phi \from [0,1]\to [0,1]\\ \text{measure preserving}}}\Normt{W^{\phi}-W'}\,,
\]
where $W^{\phi}(x,y) \seteq W(\phi(x),\phi(y))$.

Our result on graphon estimation under node differential privacy is the following.

\begin{theorem}[Graphon estimation with node differential privacy, formal statement in \cref{thm:formalmaintheorem}]
	\torestate{
		\label{thm:main_graphon}
		Let $\Lambda>0, \rho\in[0,1/\Lambda]$, and let ${W:[0,1]^2\to [0,\Lambda]}$ be a normalized graphon.  
    Suppose $\normi{\Bnull}\leq R\leq \sqrt{\frac{n\epsilon}{\polylog(n)}}$ for some $R\in \R_+$, and $\e^4 n^2\geq \polylog(n)$. 
    Then there exists an $n^{\poly(k)}$-time $\epsilon$-differentially node private algorithm that, given a \mbox{$(W,\rho,n)$-random} graph $\bm G$ and $R$, outputs a graphon $\hat{W}(\bm{G})$ such that,
		\begin{equation*}
			\E\delta_2^2(\hat{W}(\bm{G}),W) \leq O_R\Paren{
      \frac{k}{\rho n}+\frac{k^2\log(n)}{n\epsilon}+\sqrt{\frac{k}{n}}+\Paren{\epsilon_k^{(O)}(W)}^2} \,,
		\end{equation*}
		where $\epsilon_k(W)$ denotes the minimum $\delta_2$ distance between $W$ and $k$-block graphons.
	}
\end{theorem}

We prove \cref{thm:main_graphon} by observing that our SBM algorithm in \cref{thm:mainSBM} is robust to agnostic error.
Our graphon algorithm outputs a $k$-block graphon achieving the same guarantee as previous inefficient algorithms~\cite{borgs2015private,borgs2018revealing}, except the term $\frac{Rk}{\rho n}$ that is suggested to be inherent for efficient algorithms~\cite{luo2023computational}.
Notably, our algorithm runs in polynomial time for every fixed $k$, while all previous private algorithms~\cite{borgs2015private,borgs2018revealing} run in $\exp(n)$ time even for $k=2$.

To properly understand the guarantees of \cref{thm:main_graphon}, it is instructive to closely look into the error terms. 
There is an immediate correspondence between the first two and those in \cref{thm:mainSBM}. 
The error term $O(\sqrt{k/n})$ is induced from sampling error, which is to say, the distance between the ground truth graphon $W$ and the graphon associated with edge connection probability matrix $\Qnull$.
The approximation error $\eps_k^{(O)}(W)^2$ is induced from the agnostic error in approximating the graphon $W$ by a $k$-block graphon.
These terms are unavoidable and match those of the inefficient algorithms in ~\cite{borgs2018revealing}.
Finally, the conditions $R\leq \sqrt{\frac{n\epsilon}{\polylog(n)}}$ and $\e^4 n^2\geq \polylog(n)$ come from privately estimating the average degree $d$ of the graph\footnote{In \cite{ullman2019efficiently}, when $\e^4 n^2\geq \polylog(n)$, they show that privacy can be achieved for free in estimating the edge density of Erdos-Renyi graphs. We essentially borrow their algorithms and analysis. Thus we require the similar conditions}.

\paragraph{Lower bound on the sample complexity for private mechanisms}

The error bounds of our algorithm results contain a term of the form \(\tfrac {k^2 \log n}{\e n}\).
Thus, our bounds improve over the trivial estimator that outputs the zero matrix for every input only if the number of vertices is sufficiently large, \(n \gg {k^2}/{\e}\), compared to the number of blocks and the privacy parameter.

We prove the following information-theoretic lower bound that shows that \emph{no} private mechanism can significantly improve over the trivial estimator unless \(n \gg k^2 / \e\).

\begin{theorem}[Sample complexity lower bound for private estimation of stochastic block model]\label{thm:lower_bound}
  Suppose there is an $\eps$-differentially private algorithm such that for any symmetric matrix $\Bnull\in[0,4]^{k\times k}$ with entries averaging to 1, on input $\bm G$ sampled from the $(\Bnull,d,n)$-block model, outputs $\hat{B}(\mathbf G)\in \R^{k\times k}$ satisfying
  \[
    \Pr\Paren{\delta_2\Paren{\hat{B}(\mathbf G), \Bnull}\leq \frac{1}{20}} \geq \frac{2}{3} \,.
  \]
  Then, we must have \(n \geq \Omega\Paren{\frac{k^2}{\eps}}\).
\end{theorem}
We prove \cref{thm:lower_bound} in \cref{sec:lower_bound}.

\paragraph{Improvement in non-private setting} Our results improve the guarantees of existing polynomial time algorithms for the stochastic block model, surpassing them even in non-private settings. Previous algorithms only achieved an error rate of \(O\left(\frac{kR}{d}\right)\) when \(d > \text{polylog}(n)\). Our algorithm, however, attains an \(O\left(\frac{Rk}{d}\right)\) error rate without this assumption. Furthermore, we present a poly(\(n, k\)) time robust algorithm for estimating parameters in the balanced stochastic block model in \cref{sec:robustness}, under the assumption that the average degree \(d\) is known a priori.

\begin{theorem}
   Consider balanced stochastic block model (\cref{definition:sbm-balanced}).
    With high probability over $\bm G_0\sim \SBM(n,d,B_0)$, given average degree $d$ and any graph $\bm G$ obtained from $\bm G_0$ by arbitrarily corrupting $\eta\cdot n$ vertices, there is a $\poly(n,k)$-time algorithm which outputs a matrix $\hat{\bm B} \in [0,1]^{k\times k}$, and a community membership matrix $Z \in \{0,1\}^{n\times k}$ such that
  \begin{equation*}
      \|Z\hat{\bm B}Z^\top-\bm \Znull \Bnull \bm\Znull^{\top}\|_F^2\leq O_{R}\Paren{\frac{n^2}{d}\cdot k+\eta\cdot n^2} \,,
  \end{equation*}
  where $O_{R}$ hides $R$(which is the upper bound of entries in $\Bnull$).
\end{theorem}

%% file: content/techniques.tex
\section{Techniques}
\label{sec:techniques}

\providecommand{\znull}{z_0}

Let \(n,d,k,R \in \N\) with \(2\le d \le n\) and \(n\) a multiple of \(k\).
For a \(k\)-by-\(k\) matrix \(\Bnull\) with nonnegative entries bounded by \(R\) and averaging to \(1\),
we consider the following distribution \(\bm G\) over \(n\)-vertex graphs, called the \emph{\((\Bnull, d, n)\)-block model}:\footnote{
  In our formulation of the model, the vertex set is partitioned into parts of equal size.
  This formulation turns out to be more convenient in terms of estimation algorithms.
  The more common formulation assigns every vertex independently to one of the \(k\) parts at random.
  In \cref{sec:graphon}, we discuss how to translate guarantees between these two versions of the block model.
}
\begin{enumerate}
\item Partition the vertex set \([n]\) uniformly at random into \(k\) parts of equal size.

\item Then, for every pair of parts \(a,b\in [k]\) and every pair of vertices \(i,j\) in part~\(a\) and part~\(b\), respectively,
  connect vertices \(i\) and \(j\) by an edge in \(\bm G\) independently at random with probability \(\tfrac d n \cdot \Bnull(a,b)\).
\end{enumerate}
Since the entries of \(\Bnull\) average to \(1\), every vertex has expected degree \(d\) in \(\bm G\).

In other words, if we let \(\bm \Znull\in \set{0,1}^{n\times k}\) be the vertex-part incident matrix of the above random equipartition,
then \(\tfrac d n \cdot \bm \Znull \Bnull\transpose{ \bm \Znull}\) is the matrix of edge probabilities for \(\bm G\) conditioned on the partition.
Thus, letting \(\bm Y\) denote the adjacency matrix of \(\bm G\) scaled by \(\tfrac n d\), we have
\begin{equation}
  \label{eq:block-model}
  \E\Brac{\bm Y \Mid \bm \Znull} =\bm \Znull \Bnull \transpose{ \bm\Znull}\,.
\end{equation}
(Up to a permutation of the rows and columns, the matrix \(\bm \Znull \Bnull \transpose{ \bm\Znull}\) is a \(k\)-by-\(k\) matrix of \(\tfrac n k\)-by-\(\tfrac nk\) blocks, each containing a single entry of \(\Bnull\).)

Given the graph \(\bm G\), we aim to privately estimate the underlying block matrix \(\Bnull\) (in the sense of node-differential privacy).
For the sake of exposition, we assume at this point that the parameters \(n\), \(d\), \(k\), and \(R\) are known to the algorithm.
(Our final algorithms estimate the model parameter \(d\) from the input data and achieve essentially the same utility and privacy guarantees as for known \(d\).)

In terms of statistical utility and privacy guarantees, the best-known mechanisms are based on the following score function (and its Lipschitz extensions),
\begin{equation}
  \label{eq:nphardscore}
  s(B; Y) \seteq \max_{Z\in \cZ(n,k)} \iprod{ZB\transpose Z, Y} -\tfrac12 \normf{ZB\transpose Z}^2\,.
\end{equation}
Here, \(\cZ(n,k)\subseteq \set{0,1}^{n\times k}\) consists of all vertex-part incidence matrices for \(k\)-equipartitions of the vertex set \([n]\).

Following the well-known exponential-mechansim construction~\cite{mcsherry2007mechanism},
this score function defines a family of distributions \(p_{\lambda,Y}\) over \(k\)-by-\(k\) matrices \(B\) with nonnegative entries averaging to \(1\),
\begin{equation}
  \label{eq:exponential-distribution}
  p_{\lambda,Y}(B) \propto \exp(\lambda \cdot s(B; Y) / n)\,.
\end{equation}
(Here, we divide the exponent by \(n\) because the entries of \(Y\) are scaled linearly in \(n\).)
Using a straightforward discretization, we could sample from this distribution within the required accuracy using \(n^{\poly(k)}\) evaluations of the score function~\cite{mcsherry2007mechanism}, which would be polynomial\footnote{
  We expect that the exponential behavior in \(k\) is inherent for the number of score evaluations required to sample from the distribution \cref{eq:exponential-distribution}:
  Since the landscape of the score function is symmetric under the action of permutations on the rows and columns of \(B\), the distribution \cref{eq:exponential-distribution} is far from unimodal.
  In particular, the general framework of sampling log-concave or quasi-log-concave distributions does not directly apply to \cref{eq:exponential-distribution}
}
in the size of the input for every constant value of \(k\).

However, evaluating the score function for a single matrix \(B\) boils down to an NP-hard optimization problem.
Hence, even for small values of \(k\), say \(k=2\), it is not known how to sample from such distributions in polynomial time.
While good \emph{average-case} approximation algorithms exist (for the relevant statistical models) based on spectral techniques,
they have no direct consequences for our goals because, by definition, differential privacy must be maintained also in the \emph{worst case}.

Taking inspiration from the seminal works~\cite{Hopkins22SOSprivacy,hopkins2022robustness}, we investigate natural choices of score functions in order to achieve polynomial running time,
namely \emph{higher-order sum-of-squares (sos) relaxations} of the optimization problem underlying \cref{eq:nphardscore}.

The basic idea for our analysis is to simulate, to the extent possible, previous analyses of \cref{eq:nphardscore} within the sum-of-squares proof system.
While this strategy has been successfully applied to a wide range of problems (see~\cite{hopkins2020mean,barak2017quantum,raresClique} and reviews~\cite{barak2014sum,raghavendra2018high}), perhaps most relatedly for the design of robust algorithms to learn the parameters of arbitrary mixtures of \(k\) Gaussians~\cite{hopkins2018mixture,KSS18,robustGMM,buhai2023parallel}, several unique challenges arise for the clustering problem in \cref{eq:nphardscore}.
One of these challenges is that this optimization problem doesn't come with a low-dimensional algebraic structure like the Gaussian mixture model, where we can usefully simplify the problem, for example, by considering the (unknown) \(k\)-dimensional subspace spanned by the means and the (unknown) \(k^2\)-coordinates of all the means within this subspace.

\paragraph{Privacy}
For the privacy analysis of the resulting exponential mechanism, a key property of the score function \cref{eq:nphardscore} is its low sensitivity:
if \(Y,Y'\) are \(\tfrac n d\)-scaled adjacency matrices of two graphs that agree on all but vertex and that both have maximum vertex degree at most \(D\ge d\), then for every \(k\)-by-\(k\) matrix \(B\) with nonnegative entries bounded by \(R\),
\begin{equation}
  \Abs{s(B; Y)-s(B;Y')} \le \max_{Z\in \cZ(n,k)} \abs{\iprod{ZB\transpose Z, Y-Y'}} \lesssim D \cdot \tfrac n d \cdot R \,.
\end{equation}
Conveniently, any (reasonable) sum-of-squares relaxation of \cref{eq:nphardscore} directly inherits this kind of sensitivity bound.
Consequently, the corresponding exponential mechanisms~\cref{eq:exponential-distribution} are \(O(\lambda \cdot R\cdot D /d)\)-differentially private when restricted to input graphs with maximum degree at most \(D\)~\cite{mcsherry2007mechanism}.
Based on this mechanism, one can also achieve \(O(\lambda \cdot R)\)-differential privacy for all input graphs (without any restrictions on the maximum degree) with the same utility guarantees using the technique of Lipschitz extensions~\cite{borgs2015private}.
Later in this section, we discuss how to simulate this technique for exponential mechanisms based on sum-of-squares relaxations.

\paragraph{Utility (without sum-of-squares)}

In order to show utility guarantees of exponential mechanisms~\cite{mcsherry2007mechanism} for the current estimation problem,
we must show that with high probability over the draw of a random graph from a particular \(\Bnull\)-block model, the score of the intended solution \(\Bnull\) is sufficently larger than the scores of all solutions \(B\) far from \(\Bnull\) (in the appropriate metric for graphons).
Concretely, this gap between scores must be significantly larger than the dimension of the solution space, in our case \(k^2\) (up to constant factors).

To this end, let \(f(Z; B, Y)\seteq \iprod{Z B \transpose Z, Y}-\tfrac 12 \normf{ZB\transpose Z}^2\) denote the objective function of the optimization problem in \cref{eq:nphardscore}.
Then, for all \(Y\in \R^{n\times n}\), \(B,B_0\in \R^{k\times k}\), and \(Z,Z_0\in \cZ(n,k)\), the strong concavity of the function \(M\mapsto \iprod{M,Y}-\tfrac 12 \normf{M}^2\) yields the inequality,
\begin{equation}
  \label{eq:strongconcavity}
  f(Z;B,Y)-f(\Znull;\Bnull,Y)
  \le \iprod{U,\Yhat} - \tfrac12 \normf{U}^2
  \,,
\end{equation}
where \(U=ZB\transpose{Z}-\Znull \Bnull \transpose{\Znull}\) and \(\Yhat = Y - \Znull \Bnull \transpose{\Znull}\) is the gradient of the function \(M\mapsto \iprod{M,Y}-\tfrac 12 \normf{M}^2\) at the point \(\Znull \Bnull \transpose{\Znull}\).
(Actually the above inequality is an identity because there are no higher-order terms.)
Since the matrix \(U\) has rank at most \(2k\), the inner product with \(\Yhat\) satisfies\footnote{
  We remark that this inner product satisfies significantly stronger upper bounds\cite{gaooptimalgraphon,klopp2017oracle,borgs2015private}, smaller by a factor of about \(k/\log k\) for the case that \(n\) is sufficiently larger than both \(d\) and \(k\).
  The bound we discuss here corresponds to the best (inefficient) private mechanisms for graphon estimation in the literature\cite{borgs2018revealing}.
  While private mechanisms can match the improved bound, there is rigorous evidence that the bound we present here is best possible for polynomial-time algorithms, even if running times are allowed to be exponential in \(k\)~\cite{HopkinsSBM,luo2023computational}.
  This phenomenon is closely related to the conjectured information-computation gap for stochastic block models~\cite{sbmbook,luo2023computational}.
}
\begin{equation}
  \label{eq:matrixhoelder}
  \iprod{U,\Yhat} \le \normstar{U} \cdot \norm{\Yhat} \le \sqrt{2k\,}  \normf{U}\cdot \norm{\Yhat} \le 2k \norm{\Yhat}^2 + \tfrac 14\normf{U}^2\,.
\end{equation}
Here, \(\norm{\cdot}\) denotes the spectral norm (largest singular value) and \(\normstar{\cdot}\) its dual norm (sum of all singular values).
The first inequality holds by duality of the norms.
The second inequality follows from Cauchy-Schwarz applied to the vector of \(k\) (non-zero) singular values of \(U\).
The third step is the inequality of arithmetic and geometric mean.

Plugging the bound \cref{eq:matrixhoelder} into the inequality \cref{eq:strongconcavity}, we obtain
\begin{equation}
  \label{eq:spectralnorm-estimation-error}
  f(Z;B,Y)-f(\Znull;\Bnull,Y)
  \le 2k \norm{\Yhat}^2 - \tfrac 14\normf{U}^2\,.
\end{equation}
Since \(\tfrac 1 {n^2}\normf{U}^2\ge \deltatwosquared(B,B_0)\), the above inequality implies for the score function in \cref{eq:nphardscore},
\begin{equation}
  \label{eq:score-distinguishes}
  \deltatwosquared(B,\Bnull)
  \lesssim \tfrac 1 {n^2} \Paren{s(\Bnull; Y) - s(B; Y) + 2k \cdot \min_{\Znull\in \cZ(n,k)}\norm{Y-\Znull \Bnull \Znull}^2}
  \,.
\end{equation}
The corresponding exponential mechanism (with sufficiently high sampling accuracy) outputs with high probability a matrix \(\Bhat\) such that
\[
\tfrac \lambda n \cdot\Paren{ s(\Bnull; Y) - s(\Bhat; Y)} \lesssim k^2 \cdot \log n\,.
\]
If \(\bm Y\) is the \(\tfrac n d\)-scaled adjacency matrix of a random graph \(\bm G\) drawn from the \((\Bnull,d,n)\)-block model, then the matrix Bernstein inequality shows that with high-probability,
\[
  \norm{\bm Y - \bm \Znull \Bnull \transpose{\bm \Znull}} \lesssim \tfrac n d \cdot\sqrt{R d} \cdot \log n
\]
This bound together with the previous two inequalities shows that the output \(\Bhat\) of the exponential mechanism satisfies with high probability,
\[
  \deltatwosquared(\Bhat,\Bnull)
  \lesssim \tfrac 1 {n^2} \Paren{\tfrac n \lambda \cdot k^2 \cdot \log n + k \cdot \tfrac{ n^2 \cdot R \cdot \log n}{d}}
  = \tfrac 1 {\lambda} \cdot \tfrac {k^2 \log n}{n} + R \cdot \tfrac {k \log n} {d} \cdot
\]

\paragraph{Utility with sum-of-squares}

For our utility analysis of sos-based score functions,
we must show that these score functions sufficiently distinguish the matrix \(\Bnull\) from matrices \(B\) far from \(\Bnull\).
To this end, we aim to show that there are inequalities similar to \cref{eq:score-distinguishes} that have low-degree sum-of-squares proof.

Concretely, as the first step toward this goal, we show that the polynomial inequality \cref{eq:spectralnorm-estimation-error} has a degree-4 sum-of-squares proof in the variables \(Z\) (with slightly worse constant factors).
It is not clear that such a sum-of-squares proof should exist because the real-world proof proceeds in the eigenbasis of \(U\) --- an object that, in general, is not available to sum-of-squares proofs.
The sum-of-squares proof we establish is related to, but not implied by, inequalities that have previously been shown in the context of low-rank matrix estimation and matrix completion, e.g., the decomposability of the nuclear norm.
Indeed, the semidefinite programming relaxations used in these previous works correspond to degree-2 sum-of-squares relaxations, whereas the polynomials in inequality \cref{eq:spectralnorm-estimation-error} formally have degree-4 in the variables \(Z\) and thus falls outside the realm of degree-2 sum-of-squares relaxations.

The remaining step for our utility analysis of sos-based score functions is to show that the inequality \(\tfrac 1 {n^2}\normf{U}^2\ge \deltatwosquared(B,\Bnull)\) for \(U=ZB\transpose{Z}-\Znull \Bnull \transpose{\Znull}\) has a low-degree sum-of-squares proof in the variables \(Z\) up to a small multiplicative error.
Here, \(\tfrac 1 {n^2}\normf{U}^2\) is a degree-\(4\) polynomial in \(Z\) with coefficients depending on \(B\), \(\Bnull\), and \(\Znull\).
If we were to minimize this \(n\cdot k\)-variate polynomial over all \(Z\in \cZ(n,k)\), it is straightforward to show that the minimum is at least \(\deltatwosquared(B,\Bnull)\).
The challenge is to show that sum-of-squares can certify this lower bound up to a small (multiplicative) error.
Furthermore, we want to bound the error \emph{uniformly} over all choices of \(B\), \(\Bnull\), and \(\Znull\).
(In particular, there are no average-case properties to be exploited at this point.)

\renewcommand{\Ind}{\mathbb 1}

To this end, we show that it is possible to carry out a variable reduction from \(n\cdot k\) variables to just \(k^2\) variables.
Concretely, we associate to all choices of \(B\), \(\Bnull\), and \(\Znull\) a \(k^2\)-variate quadratic polynomial \(p(S; B,\Bnull,\Znull)\) in variables \(S\) and a linear function \(S(Z; \Znull)\) in variables \(Z\).
For every choice of \(\Znull\in\cZ(n,k)\), the linear function \(S(Z; \Znull)\)  maps \(\cZ(n,k)\) into the polytope \(P_k\) of \(k\)-by-\(k\) doubly-stochastic matrices such that the following polynomial identity in variables \(Z\) holds subject to the equipartition constraints \(\cA_1(Z)\seteq \set{Z\odot Z=Z, Z \Ind=\Ind, \transpose{Z} \Ind = \tfrac n k \Ind}\),
\[
\tfrac 1 {n^2}\normf{U}^2 = p\Bigparen{S(Z; \Znull); B,\Bnull,\Znull}\,.
\]
Consequently, in order to prove a lower bound for \(\tfrac 1 {n^2}\normf{U}^2\), it is enough to lower bound the polynomial \(p\).
We can expect this task to be easy for degree-\(O(k^2)\) sum-of-squares proofs because \(p\) has only \(k^2\) variables.
A-priori, however, the minimum of \(p\) over doubly-stochastic matrices \(S\) could be much smaller than the minimum of the polynomial \(\tfrac 1 {n^2}\normf{U}^2\) over matrices \(Z\in \Z(n,k)\).
Fortunately for us, the minimum of \(p\) over doubly-stochastic matrices \(S\) turns out to be an alternative characterization of the graphon distance \(\deltatwosquared(B,\Bnull)\).
(This characterization is related to the previously observed fact that graphon distances are Gromov-Wasserstein distances~\cite{xu2020gwdistance,Peyre15Gromov}.)
It remains to argue that level-\(O(k^2)\) sum-of-squares relaxations provide a multiplicative approximation to the exact minimum, as discussed in the following paragraph.

\paragraph{Approximation scheme for polynomial minimization over (near-)integral polytopes}

While degree bounds are available for polynomial optimization subject to general (Archimedean) polynomial systems (e.g.,~\cite{DBLP:journals/jc/NieS07}),
they do not suffice (directly) for our purposes because they only give additive\footnotemark{} guarantees and their degree bound is exponentially worse than our desired bound of \(k^{O(1)}\).
\footnotetext{
  Eventually our utility guarantees provide an additive error of about \(O(\tfrac k d)\).
  If we were ensure this small an additive error for the polynomial optimization step, our final running time would be exponential, or even doubly exponential, in the expected average degree \(d\) of the observed random graph.
}
By exploiting the convexity of the optimization domain in our setting (concretely, the convex polytope of doubly-stochastic matrices), we obtain degree bounds that are polynomial in the number of variables and the desired (additive) accuracy.
This degree bound builds on rounding techniques introduced for sum-of-squares relaxations of the best-separable-state problem~\cite{barak2017quantum}.
Finally, we exploit the integrality of the polytope and the fact that it is contained in the nonnegative orthant, in order to turn this additive guarantee into a multiplicative one for the problem of minimizing polynomials with nonnegative coefficients.

\paragraph{Lipschitz extension}

As mentioned before, exponential mechanisms based on the score function \cref{eq:nphardscore} (or its sum-of-squares relaxations) provide good guarantees of node-differential privacy only when restricted to input graphs with maximum degree \(O(d)\).
Lipschitz extensions of the linear function \(Y \mapsto \iprod{ZB\transpose Z, Y}\) allow us to extend these privacy guarantees to arbitrary input graphs~\cite{borgs2015private,borgs2018revealing,privateMedian}.
The idea is to project the input graph in a particular way into the set of graphs with maximum degree bounded by \(20R\cdot d\).\footnote{Here, the additional factor \(R\) in the degree bound reflects our assumption that the observed random graph stems from a matrix \(B_0\) with entries bounded by \(R\).}
Concretely, we replace the above linear function by the following piecewise-linear function (noting that for input graphs with maximum degree at most \(20Rd\), the scaled adjacency matrix \(Y\) has all row averages upper bounded by \(20R\)),
\begin{equation}
  \label{eq:piecewiselinearextension}
  Y\mapsto \max\Set{ \iprod{ZB\transpose Z,Y_{-}} \Mid 0 \le Y_{-} \le Y,~\tfrac 1 n Y_{-} \Ind \le 20 R\cdot \Ind,~Y_{-}=\transpose{Y_{-}}}\,.
\end{equation}
The above inequalities between vectors and matrices are understood entry-wise.
We use the notation \(Y_-\) in order to indicate that this matrix is an entry-wise lower bound for \(Y\).

Since \(ZB\transpose Z\) has only nonnegative entries, the above function agrees with the original linear function \(Y\mapsto \iprod{ZB\transpose Z, Y}\) for all symmetric matrices \(Y\) with all row averages upper bounded by \(20R\).
(In this case, the choice \(Y_- = Y\) achieves the maximum in \cref{eq:piecewiselinearextension}.)

At the same time, the function in \cref{eq:piecewiselinearextension} has small sensitivity.
In particular, for every pair of neighboring matrices \(Y,Y'\) (differing in at most one row and column), every \(Y_{-}\) feasible for \(Y\) can be made into a matrix \(Y_-'\) feasible for \(Y'\) by zeroing out one row and column.
Consequently, the function at \(Y'\) is at least as large as at \(Y\) up to an additive error of at most \(20n\cdot R\).
(Recall that \(R\) is an upper bound on the entries of \(ZB\transpose Z\).)
For symmetry reasons, it follows that the sensitivity of the function is at most \(20n\cdot R\).

In order to simulate this construction for sos-based score functions, we introduce \(Y_{-}\) as auxiliary variable for the sos proof system constrained in the same way as in \cref{eq:piecewiselinearextension}.

Concretely, we define the score function \(s_{\ell}(B; Y)\) as the largest number \(t\in \R\) such that the following polynomial system in variables  \(Z\) and \(Y_-\) is sos-consistent\footnote{
  Here, we say a polynomial is sos-consistent up to level \(\ell\ge 4\) if it has no sos-refutation within that level, which also means that there exists a level-\(\ell\) pseudo-distribution satisfying the polynomial system.\label{ft:sos-consistent}
  } up to level \(\ell\ge 4\),
\begin{equation}
  \label{eq:polynomial-system}
  \cA(Z,Y_-; B,Y, t) \seteq \cA_1(Z) \cup \cA_2(Y_-; Y) \cup \Set{\iprod{ZB\transpose Z, Y_-}-\tfrac12\normf{ZB\transpose Z}^2 \ge t}\,.
\end{equation}
Here, \(Z\) and  \(Y_-\)  are \(n\)-by-\(k\) and \(n\)-by-\(n\) matrices of variables, respectively.
The polynomial system \(\cA_1(Z)\) encodes that \(Z\) is the incidence matrix of a \(k\)-equipartition and \(\cA_2(Y_-;Y)\) contains the constraints for the Lipschitz extension in \cref{eq:piecewiselinearextension},
\begin{gather}
  \cA_1(Z) \seteq
  \Set{Z\odot Z=Z, Z \mathbf{1}=\Ind, \transpose{Z} \Ind = \tfrac n k \Ind}\,,
  \\
  \cA_2(Y_-;Y) \seteq
  \Set{0 \le Y_{-} \le Y,~\tfrac 1 n Y_{-} \Ind \le 20R\cdot \Ind,~Y_{-}=\transpose{Y_{-}}}
\end{gather}
In order to show that the score function \(s_\ell(B;Y)\) has sensitivity at most \(40n\cdot R^2\), we consider any two neighboring matrices \(Y,Y'\) differing in at most one row and column and any level-\(\ell\) pseudo-distribution \(\mu\) that witness the level-\(\ell\) sos-consistency of \cref{eq:polynomial-system} for \(t=s_\ell(B;Y)\).
Then, we construct from \(\mu\) and new level-\(\ell\) pseudo-distribution \(\mu'\) that witnesses the sos-consistency of the system \(\cA(Z,Y_-; B,Y', t-40n\cdot R^2)\).
Here, the idea is to simulate at the level of pseudo-distributions the process of zeroing out the row and column in \(Y_-\) where \(Y\) and \(Y'\) differ.
For this construction, it is important that it does not reduce the level of the pseudo-distribution.
Otherwise, we would be comparing scores \(s_\ell(B;Y)\) and \(s_{\ell'}(B;Y')\) for two different level parameters \(\ell\) and \(\ell'\), which is not enough to reason about the sensitivity of the function \(Y\mapsto s_{\ell}(B;Y)\) or the privacy guarantees of the resulting exponential mechanism.

Our previous utility analysis worked with the constraint \(\iprod{ZB\transpose Z, Y}-\tfrac12\normf{ZB\transpose Z}^2 \ge t\) and crucially exploited that \(Y\) was a constant for the proof system.
In this way, our sos proofs could make use of the existence of certain low-rank approximations for \(Y\) and spectral norm bounds for the error.
However, in the current polynomial system \(\cA(Z,Y_-; B, Y, t)\) we only have the   constraint \(\iprod{ZB\transpose Z, Y}-\tfrac12\normf{ZB\transpose Z}^2 \ge t\), where \(Y\) has been replaced by a variable \(Y_-\) for the proof system.
For this matrix of variables \(Y_-\), it is not clear if information about certain low-rank approximations or spectral norm bounds are available to us in the proof system (especially because we cannot add additional constraints for \(Y_-\) to the system without potentially sacrificing the sensitivity bound we require).

To resolve this issue, we expoit a monontonicity property of the constraints \(\cA(Z,Y_-; B, Y, t)\).
If the matrix \(Y\) satisfies our constraints on the row and column averages so that \(\cA_2(Y; Y)\) holds, then we show that we can derive our original constraint \(\iprod{ZB\transpose Z, Y}-\tfrac12\normf{ZB\transpose Z}^2 \ge t\) from the constraints \(\cA(Z,Y_-; B, Y, t)\) in the sos proof system (with small degree).
Hence, we can reuse the previous utility analysis whenever \(Y\) has upper bounded row and column averages.

\paragraph{Constant average degree}

The utility analysis outlined so far provides an upper bound on the error no better than \(\tfrac {k \log n}{d}\).
This error bound is meaningful only if \(d\gg k \log n\).
In particular, we cannot get a meaningful guarantee in this way for graphs with large constant degree.
This logarithmic factor in the error bound comes from the spectral norm bound for the centered adjacency matrix of our input graph.
(Indeed, due to the presence of vertices with degree at least \(\sqrt {\log n}\), this logarithmic factor is required for the spectral norm of the centered adjacency matrix.)
There are (at least) two approaches in the literature for avoiding this logarithmic factor in the final error bound.

The first approach is to use an analysis based on the cut-norm (as a substitute for the spectral norm) and Grothendieck's inequality (e.g.,~\cite{guedon14:_commun_groth}).
While this approach yields meaningful guarantees for constant average degree, the error bounds achieved with this approach in the literature are substantially worse than \(\tfrac k d\) as far as we are aware
(similar to the distinction between fast and slow error bounds of lasso for sparse linear regression).

The second approach to deal with constant average degree is to prune high-degree vertices (e.g.,~\cite{coja-oghlan-soda-05}) and to redo the analysis based on spectral norms for the remaining graph.
A recent iteration of this approach shows the following remarkable property:
after removing the rows and columns of vertices with average degree larger than \(20d R\), the centered adjacency matrix has spectral norm bounded by \(O(\sqrt {Rd})\) --- smaller by a factor \(\sqrt {\log n}\) than before pruning.

By virtue of being node-private, our algorithm turns out to be robust to pruning a small fraction of vertices.
That robustness allows us to carry out the utility analysis on the pruned graph (where the relevant spectral norm is nicely bounded) and conclude that utility also holds for the original graph before pruning.

%% file: content/preliminaries.tex
\section{Preliminaries}\label{section:preliminaries}

We use boldface to denote random variables, e.g., \(\bm X, \bm Y, \bm Z\).

We write \(f \lesssim g\) to denote the inequality \(f \le C \cdot g\) for some absolute constant \(C>0\).
We write \(O(f)\) and \(\Omega(f)\) to denote quantities \(f_-\) and \(f_+\) satisfying \(f_-\lesssim f\) and \(f \lesssim f_+\), respectively.

We denote functions of the variables \(Y_1, Y_2, \ldots, Y_t\), which are parameterized by \(X_1, X_2, \ldots, X_t\), using the notation \(f(X_1, X_2, \ldots, X_t; Y_1, Y_2, \ldots, Y_t)\).

The $L_p$-norm of a measurable function $f:[0,1]^2 \to \R$ is $\norm{f}_p \seteq \Paren{\int\int \abs{f(x,y)}^p dx dy}^{1/p}$.

For a matrix \(M\in \R^{n\times m}\), we denote
its  \((i,j)\)-th entry by \(M(i,j)\),
its \(i\)-th row by \(M(i,\cdot)\),
and its  $j$-th column by \(M(\cdot, j)\).
We use \(\norm{M}\) for the spectral norm of \(M\) and \(\normf{M}\) for the Frobenius norm of \(M\).
We denote by \(\normo{M}\) and \(\normm{M}\) the sum and the maximum of the absolute values of the entries in \(M\), respectively.
For two matrices \(M,N\in \R^{n\times m}\), we denote their inner product by \(\iprod{M,N} = \Tr M \transpose{N}=\sum_{i,j} M(i,j)N(i,j)\).

Given $n,k\in\N$ with $n$ a multiple of $k$, let $\cZ(n,k) \subseteq \bits^{n\times k}$ consist of all $n\times k$ binary matrices of which each row sums to 1 and each column sums to $n/k$.
That is, each $Z\in\cZ(n,k)$ encodes a $k$-equipartition of $[n]$.

\begin{definition}[Doubly stochastic matrix]
  \label{def:ds}
  A square nonnegative real matrix $M=(M_{ij})$ is \emph{doubly stochastic} if each of its rows and columns sums to $1$.
\end{definition}

\begin{definition}[Birkhoff polytope]
  For every $k\in\N$, the set of $k\times k$ doubly stochastic matrices forms a convex polytope known as the Birkhoff polytope $B_{k}$.
\end{definition}

The following theorem shows that the Birkhoff polytope has the set of permutation matrices as its corners.

\begin{theorem}[Birkhoff--von Neumann theorem]
  \label{thm:birkhoff}
  Every $k\times k$ doubly stochastic matrix is a convex combination of at most $2k^2$ permutation matrices.
\end{theorem}

\subsection{Differential privacy}

\paragraph{Node-adjacent graphs}
Two graphs are \textit{node-adjacent to} if they can be made isomorphic to each other by removing exactly one vertex from each graph.
In other words, we can turn one graph into the other one (up to isomorphism) by rewiring the edges of one vertex.

\begin{definition}[Node-differential privacy~\cite{KNRS13}]\label{definition:node-dp}
  A randomized algorithm $\mathcal{A}$ is $\e$-differentially (node) private if for all events in the output space $S$ and all node-adjacent graphs $G,G'$, we have
  \begin{equation*}
    \Pr\Brac{\mathcal{A}(G)\in S}\leq  \exp(\epsilon)\cdot \Pr\Brac{\mathcal{A}(G')\in S}
    \,.
  \end{equation*}
\end{definition}

Since all our algorithms on graphs have an output distribution that is invariant under permutations of the vertices of the input graph, we can assume node-adjacent input graphs to be identical up to rewiring a single vertex.

\paragraph{Exponential mechanism}

Our algorithm is based on the exponential mechanism by McSherry and Talwar~\cite{mcsherry2007mechanism}.

\begin{theorem}[Exponential Mechanism~\cite{mcsherry2007mechanism}]
  \label{thm:exp-mech}
  Let \(s\from \cX \times \cY\to \R\) be \(T\)-time computable.
  Suppose the functions \(s(\cdot; Y)\from \cX\to \R\) are \(L\)-Lipschitz.
  For all \(\eta>0\), let \(\cX_\eta\) denote an \(\eta\)-net of \(\cX\).
  Then, given \(Y\in \cY\) and \(\cX_\eta\), the exponential mechanism computes a randomized output \(\bm{ \hat X}\in \cX_\eta\) in time \(O(T\cdot \card{\cX_\eta})\) such that for all \(\delta>0\),
  \begin{displaymath}
    1-\delta \le \Pr\Set{\max_{X\in \cX}s(X,Y) \le s(\bm {\hat X})  + \eta\cdot L + \log \tfrac{\card{\cX_\eta}}{\delta}  }
  \end{displaymath}
  Furthermore, if \(\abs{s(X; Y)-s(X;Y')}\le \e\) for all \(X\in \cX\) and all pairs of adjacent inputs \(Y,Y'\in \cY\),
  then this mechanism is \(O(\e)\)-differentially private.
\end{theorem}

\subsection{Sum-of-squares hierarchy}

In this paper, we employ the sum-of-squares hierarchy \cite{barak2014sum,sos2016note,raghavendra2018high} for both algorithm design and analysis. 
As a broad category of semidefinite programming algorithms, sum-of-squares algorithms provide many optimal or state-of-the-art results in algorithmic statistics~\cite{hopkins2018mixture,KSS18,pmlr-v65-potechin17a,hopkins2020mean}.
We provide here a brief introduction to pseudo-distributions, sum-of-squares proofs, and sum-of-squares algorithms.
For more detailed background, please refer to \cref{section:backgroundsos}.

\paragraph{Pseudo-distribution} 

We can represent a finitely supported probability distribution over $\R^n$ by its probability mass function $\mu\from \R^n \to \R$ such that $\mu \geq 0$ and $\sum_{x\in\supp(\mu)} \mu(x) = 1$.
We define pseudo-distributions as generalizations of such probability mass distributions, by relaxing the constraint $\mu\ge 0$ and only requiring that $\mu$ passes certain low-degree non-negativity tests.

\begin{definition}[Pseudo-distribution]
  \label{def:pseudo-distribution}
  A \emph{level-$\ell$ pseudo-distribution} $\mu$ over $\R^n$ is a finitely supported function $\mu:\R^n \rightarrow \R$ such that $\sum_{x\in\supp(\mu)} \mu(x) = 1$ and $\sum_{x\in\supp(\mu)} \mu(x)f(x)^2 \geq 0$ for every polynomial $f$ of degree at most $\ell/2$.
\end{definition}

We can define the formal expectation of a pseudo-distribution in the same way as the expectation of a finitely supported probability distribution.

\begin{definition}[Pseudo-expectation]
  Given a pseudo-distribution $\mu$ over $\R^n$, we define the \emph{pseudo-expectation} of a function $f:\R^n\to\R$ by
  \begin{equation}
    \tE_\mu f \seteq \sum_{x\in\supp(\mu)} \mu(x) f(x) \,.
  \end{equation}
\end{definition}

The following definition formalizes what it means for a pseudo-distribution to satisfy a system of polynomial constraints.

\begin{definition}[Constrained pseudo-distributions]
  Let $\mu:\R^n\to\R$ be a level-$\ell$ pseudo-distribution over $\R^n$.
  Let $\cA = \{f_1\ge 0, \ldots, f_m\ge 0\}$ be a system of polynomial constraints.
  We say that \emph{$\mu$ satisfies $\cA$} at level $r$, denoted by $\mu \sdtstile{r}{} \cA$, if for every multiset $S\subseteq[m]$ and every sum-of-squares polynomial $h$ such that $\deg(h)+\sum_{i\in S}\max\set{\deg(f_i),r} \leq \ell$,
  \begin{equation}
    \label{eq:constrained-pseudo-distribution}
    \tE_{\mu} h \cdot \prod_{i\in S}f_i \ge 0 \,.
  \end{equation}
  We say $\mu$ satisfies $\cA$ and write $\mu \sdtstile{}{} \cA$ (without further specifying the degree) if $\mu \sdtstile{0}{} \cA$.
\end{definition}

We remark that if $\mu$ is an actual finitely supported probability distribution, then we have  $\mu\sdtstile{}{}\cA$ if and only if $\mu$ is supported on solutions to $\cA$.

\paragraph{Sum-of-squares proof} 

We introduce sum-of-squares proofs as the dual objects of pseudo-distributions, which can be used to reason about properties of pseudo-distributions.
We say a polynomial $p$ is a sum-of-squares polynomial if there exist polynomials $(q_i)$ such that $p = \sum_i q_i^2$.

\begin{definition}[Sum-of-squares proof]
  \label{def:sos-proof}
  A \emph{sum-of-squares} proof that a system of polynomial constraints $\cA = \{f_1\ge 0, \ldots, f_m\ge 0\}$ implies $q\ge0$ consists of sum-of-squares polynomials $(p_S)_{S\subseteq[m]}$ such that\footnote{Here we follow the convention that $\prod_{i\in S}f_i=1$ for $S=\emptyset$.}
  \[
    q = \sum_{\text{multiset } S\subseteq[m]} p_S \cdot \prod_{i\in S} f_i \,.
  \]
  If such a proof exists, we say that \(\cA\) \emph{(sos-)proves} \(q\ge 0\) within degree \(\ell\), denoted by $\mathcal{A}\sststile{\ell}{} q\geq 0$.
  In order to clarify the variables quantified by the proof, we often write \(\cA(x)\sststile{\ell}{x} q(x)\geq 0\).
  We say that the system \(\cA\) \emph{sos-refuted} within degree \(\ell\) if $\mathcal{A}\sststile{\ell}{} -1 \geq 0$.
  Otherwise, we say that the system is \emph{sos-consistent} up to degree \(\ell\), which also means that there exists a level-$\ell$ pseudo-distribution satisfying the system.
\end{definition}

The following lemma shows that sum-of-squares proofs allow us to deduce properties of pseudo-distributions that satisfy some constraints.
\begin{lemma}
  \label{lem:sos-soundness}
  Let $\mu$ be a pseudo-distribution, and let $\cA,\cB$ be systems of polynomial constraints.
  Suppose there exists a sum-of-squares proof $\cA \sststile{r'}{} \cB$.
  If $\mu \sdtstile{r}{} \cA$, then $\mu \sdtstile{r\cdot r' + r'}{} \cB$.
\end{lemma}

\paragraph{Sum-of-squares algorithm}

Given a system of polynomial constraints, the \emph{sum-of-squares algorithm} searches through the space of pseudo-distributions that satisfy this polynomial system, by solving semideﬁnite programming.

Since semidefinite programing can only be solved approximately, we can only find pseudo-distributions that approximately satisfy a given polynomial system.
We say that a level-$\ell$ pseudo-distribution \emph{approximately satisfies} a polynomial system, if the inequalities in \cref{eq:constrained-pseudo-distribution} are satisfied up to an additive error of $2^{-n^\ell}\cdot \norm{h}\cdot\prod_{i\in S}\norm{f_i}$, where $\norm{\cdot}$ denotes the Euclidean norm\footnote{The choice of norm is not important here because the factor $2^{-n^\ell}$ swamps the effects of choosing another norm.} of the coefficients of a polynomial in the monomial basis.

\begin{theorem}[Sum-of-squares algorithm]
  \label{theorem:SOS_algorithm}
  There exists an $(n+ m)^{O(\ell)} $-time algorithm that, given any explicitly bounded\footnote{A system of polynomial constraints is \emph{explicitly bounded} if it contains a constraint of the form $\|x\|^2 \leq M$.} and satisfiable system\footnote{Here we assume that the bit complexity of the constraints in $\cA$ is $(n+m)^{O(1)}$.} $\cA$ of $m$ polynomial constraints in $n$ variables, outputs a level-$\ell$ pseudo-distribution that satisfies $\cA$ approximately.
\end{theorem}

\begin{remark}[Approximation error and bit complexity]
  \label{remark:sos-numerical-issue}  
  For a pseudo-distribution that only approximately satisfies a polynomial system, we can still use sum-of-squares proofs to reason about it in the same way as \cref{lem:sos-soundness}.
  In order for approximation errors not to amplify throughout reasoning, we need to ensure that the bit complexity of the coefficients in the sum-of-squares proof are polynomially bounded.  
\end{remark}

%% file: content/mainalgorithm.tex
\section{Differentially private algorithm for graph estimation}
\label{section:mainalg}

In this section, we provide a technical overview of the analysis of our polynomial-time node-differentially-private algorithms for random graph estimation.
All the random graph models we consider have in common that with overwhelming probability, the adjacency matrix of the random graph is equal to a \(k\)-by-\(k\) block matrix up to various (small) approximation errors and statistical errors.
In order to factor out details specific to particular models, we first present differentially private algorithms whose utility guarantees are deterministic with respect to their input (but stochastic with respect to their internal randomness).
These algorithm receive as input an \(n\)-by-\(n\) matrix \(\Yin\) (eventually chosen as an appropriately scaled adjacency matrix of graph drawn from one of our random graph models) and achieve utility guarantees of the following form:
if the matrix \(\Yin\) admits certain approximations as a \(k\)-by-\(k\) block matrix \(\Znull \Bnull \transpose \Znull\) for \(\Znull\in\cZ(n,k)\) and \(\Bnull\in \R_{\ge 0}^{k \times k}\), then the randomized output \(\bm B\) is close in \(\delta_2\)-distance to \(\Bnull\) (with high probability over the internal randomness of the private algorithm).

Let \(Z\) and  \(Y\)  be \(n\)-by-\(k\) and \(n\)-by-\(n\) matrices of indeterminates, respectively.
Let \(R\ge 1\) be a scalar.
For an \(n\)-by-\(n\) matrix \(\Yin\) with nonnegative entries, we consider the following polynomial systems\footnote{
  Here, the constraint \(Z\ge 0\) in the system \(\cA_1(Z)\) is redundant in the sense that we can sos-derived it with degree \(2\) from the other constraints in the system.
  We include this constraint explicitly in the system to ensure that we can sos-derive \(Z \otimes Z\ge 0\) in degree \(2\)---a property we require for our sensitivity analysis.
}
in variables \(Y\) and \(Z\) and with coefficients depending on \(\Yin\),
\begin{gather}
  \cA_1(Z) \seteq
  \Set{Z\odot Z=Z,~Z \Ind=\Ind,~\transpose{Z} \Ind = \tfrac n k \Ind,~Z\ge 0}\,,
  \\
  \cA_2(Y;\Yin) \seteq
  \Set{0 \le Y \le \Yin,\tfrac 1 n Y \Ind \le 20 \cdot R \cdot \Ind,~Y=\transpose{Y}}
\end{gather}
Here, \(\odot\) denotes the entry-wise product of matrices (often called Hadamard product).
A matrix \(Z\) satisfies the constraints \(\cA_1(Z)\) if and only if \(Z\in \cZ(n,k)\) is the incidence matrix of a \(k\)-equipartition.
A symmetric matrix \(Y\) satisfies the constraints \(\cA(Y;\Yin)\) if and only if \(Y\) is entry-wise sandwiched between \(0\) and \(\Yin\) and \(Y\) has all row and column averages upper bounded by \(20R\).

For a \(k\)-by-\(k\) matrix \(B\) and a scalar \(t\in \R\), we consider the following combined polynomial system in variables \(Y\) and \(Z\) and with coefficients depending on \(B\), \(\Yin\), and \(t\),
\begin{equation}
  \label{eq:combined_poly_sys}
  \cA(Y,Z; B, \Yin, t) \seteq \cA_1(Z) \cup \cA_2(Y; \Yin) \cup \Set{f(Y,Z; B) \ge t}\,,
\end{equation}
where \(f(Y,Z;B)\) is the following polynomial in variables \(Y\) and \(Z\) with coefficients depending on \(B\),
\begin{equation}
  f(Y,Z; B) \seteq \iprod{Z B \transpose Z, Y} - \tfrac 12 \normf{ZB\transpose Z}^2\,.
\end{equation}
When we substitute for \(Y\) a concrete assignment \(Y_0\), we use \(f(Z;B,Y_0)\) to denote the resulting polynomial in \(Y\) with coefficients depending on \(B\) and \(Y_0\),
\begin{equation}
  f(Z; B,Y_0) \seteq \iprod{Z B \transpose Z, Y_0} - \tfrac 12 \normf{ZB\transpose Z}^2\,.
\end{equation}

In this section, we analyze exponential mechanisms based on score functions \(s_\ell(B;\Yin)\) defined as the largest value \(t\in\R\) such that the system \(\cA(Y,Z; B,\Yin, t)\) is sos-consistent up to level \(\ell\).

Toward analyzing these score functions, we first identify useful inequalities we can sos-prove in low degree from the system \(\cA(Y,Z; B, \Yin, t)\).
Suppose \(Y_0\) is a \(n\)-by-\(n\) matrix of the form \(Y_0=\Znull \Bnull \transpose \Znull\).
As discussed in \cref{sec:techniques}, a basic fact from low-rank matrix estimation is the following inequality,
\begin{equation}
  \label{eq:non-sos-lowrank}
  \normf{ZB\transpose{Z}-Y_0}^2
  \le 8k\cdot \norm{\Yin - Y_0}^2 + 4\Bigparen{f(\Znull; \Bnull, \Yin)-f(Z; B, \Yin)}\,,
\end{equation}
The above inequality shows that \(Y_0\) is identifiable from \(\Yin\) (in the sense of achieving small error in Frobenius norm) if the two matrices \(Y_0,\Yin\) are close in spectal norm.
In the non-private setting, we could aim to choose \(Z,B\) so as to maximze the correlation \(f(Z; B, \Yin)\) with \(\Yin\).
In that case, the last term on the right-hand side would be at most \(0\) and we would get an error bound for the Frobenius norm purely in terms of the spectral norm of \(\Yin-Y_0\) and the rank \(k\).
Moreover, the inequality shows that the error bound degrades gracefully even as the ``optimization gap'' \(f(\Znull; \Bnull, \Yin)-f(Z; B, \Yin)\) grows.

The following lemma shows that inequality \cref{eq:non-sos-lowrank} has a low-degree sos proof.
\begin{lemma}[Sum-of-squares proof of identifiability I]
  \label{lem:basic-sos-utility}
  For all symmetric \(n\)-by-\(n\) matrices \(\Yin, Y_0\) with \(\rank Y_0\le k\) and \(Y_0=\Znull \Bnull \transpose \Znull\) and for all symmetric \(k\)-by-\(k\) matrices \(B\),
  \[
    \cA_1(Z)
    \sststile{O(1)}{Z}
    \normf{Z B \transpose Z - Y_0}^2
    \le 48 k \cdot \norm{\Yin - Y_0}^2
    + 4\Bigparen{f(Z_0;B_0,\Yin) - f(Z; B, \Yin)}
    \,,
  \]
\end{lemma}

\begin{proof}
  For \(U=ZB\transpose Z - \Znull \Bnull \Znulltrans\) and \(\hat Y = \Yin-Y_0\), we have the following polynomial identity (which can be viewed as the Taylor expansion of \(Y\mapsto\iprod{Y,\Yin}-\tfrac 12 \normf{Y}^2\) around \(Y_0\) evaluated at \(ZB\transpose Z\)),
  \begin{equation}
    \label{eq:taylor-poly}
    f(Z; B, \Yin) - f(Z_0;B_0,\Yin) = \iprod{U, \hat Y} - \tfrac 12 \normf{U}^2\,.
  \end{equation}
  Let \(V_0\) be an orthogonal basis for the column span of \(Y_0\) so that \(\transpose {V_0}V_0=I_k\).
  Let \(V=\sqrt{\tfrac{k}{n}} Z\) so that \(\cA_1(Z)\sststile{4}{Z} \transpose V V =I_k\) by \cref{lem:sos-community-orthogonality}.
  (Here, we can view \(V\) as a formal orthogonal basis for the column span of \(Z B \transpose{Z}\).)
  It follows that \(\cA_1(Z)\sststile{4}{Z} (I_n-V\transpose V) U (I_n-\Vnull \Vnulltrans)=0\).
  (We can view this identity as an sos proof for the statement that \(U\) has rank at most \(2k\).)
  By \cref{cor:sos-spectral-hoelder} (sos version of the inequality \(\iprod{U,\hat Y}\le \sqrt{\rank U} \cdot \normf{U}\cdot \norm{\hat Y}\)),
  \begin{equation}
    \cA_1(Z)\sststile{O(1)}{Z}
    \iprod{U,\hat Y}\le \tfrac{1}{4} \normf{U}^2 + 12k \norm{\hat Y}^2
    \,.
  \end{equation}
  (When applying \cref{cor:sos-spectral-hoelder}, we choose \(M= U/\sqrt {6\,}\) and \(W=\sqrt {6\,} \hat Y\).)
  Combining this inequality with the polynomial identity \cref{eq:taylor-poly}, we obtain an sos proof of the desired inequality,
  \[
    \cA_1(Z)\sststile{O(1)}{Z}
    \normf{U}^2 \le 48 k \norm{\hat Y}^2 + 4 \cdot \bigparen{f(Z_0; B_0, \Yin) - f(Z;B,\Yin)}
    \,.
    \qedhere
  \]
\end{proof}

While \cref{lem:basic-sos-utility} already implies limited utility guarantees for some random graph models, we need the following more robust version to handle more challenging random graph models.
This version allows in addition to errors in spectral norm also errors in Frobenius norm and \(\ell_1\)-norm of the entries (denoted by \(\normsum{\cdot}\)).
To model these additional errors, we introduce intermediate matrices \(Y_1,Y_2\) and decompose the total error \(\Yin-Y_0\) into three parts \(\Yin-Y_2\), \(Y_2-Y_1\), and \(Y_1-Y_0\), which we expect to be small in \(\ell_1\)-norm, Frobenius norm, and spectral norm, respectively.

\begin{lemma}[Sum-of-squares proof of identifiability II]
  \torestate{\label{lem:abstractidentifiability}
    Let $Z_0 \in \cZ(n, k)$ and let $B, B_0 \in \R_{+}^{k \times k}$ be symmetric matrices with $\normm{B}, \normm{B_0} \leq R$.
    Then, for \(Y_0 = Z_0 B_0 \Znulltrans\) and for all \(n\)-by-\(n\) matrices \(Y_1,Y_2,\Yin\),
    \[
      \cA_1(Z), f(Z; B, \Yin) \ge t
      \sststile{O(1)}{Z}
      \normf{Z B \transpose Z - Y_0}^2 \lesssim k \cdot \norm{Y_1 - Y_0}^2 + \normf{Y_2-Y_1}^2 + R\normo{\Yin-Y_2}+\nu \,,
    \]
    where \(\nu = \max\bigset{0,f(Z_0; \Bnull, Y_2) - t}\).
  }
\end{lemma}

\begin{proof}
  Since \(\normm{B}\le R\), we have \(\cA_1(Z)\sststile{O(1)}{Z} 0\le ZB\transpose{Z} \le R \cdot \Ind_n \transpose{\Ind_n}\) (by \cref{lem:block-matrix-sos-bound}).
  Consequently, by \cref{lem:sos-version-of-ell1-elli-Hoelder} (sos version of \Holder inequality for \(\ell_\infty\) and \(\ell_1\))
  \[
    \cA_1(Z)\sststile{O(1)}{Z}
    f(Z; B, Y_2) - f(Z;B, \Yin)=\iprod{ZB\transpose{Z},Y_2-\Yin}
    \ge  -R\cdot \normo{\Yin-Y_2}\,.
  \]
  At the same time, we have the following polynomial identity,
  \[
    f(Z; B, Y_1)-f(Z_0; B_0, Y_1)
    = \iprod{ZB\transpose{Z}-\Znull\Bnull\Znulltrans,Y_1-Y_2} + f(Z; B,Y_2) - f(Z_0;B_0, Y_2)\,.
  \]
  Thus, by \cref{fact:simple-sos-version-of-cauchy-schwarz} (sos version of Cauchy--Schwarz),
  \[
    \begin{aligned}
      \cA_1(Z)\sststile{O(1)}{Z}
      f(Z; B, Y_1)-f(Z_0; B_0, Y_1)
      \ge
      & f(Z; B,Y_2) - f(Z_0;B_0, Y_2)\\
      & - \tfrac 1 {16}\normf{ZB\transpose{Z}-\Znull\Bnull\Znulltrans}^2
      -  4 \normf{Y_1-Y_2}^2 \,.
    \end{aligned}
  \]
  Together, these bounds imply,
  \[
    \begin{aligned}
      \cA_1(Z), f(Z;B,\Yin)\ge t\sststile{O(1)}{Z}\\
       f(Z_0; B_0, Y_1) - f(Z; B, Y_1)
      & \le && \tfrac 1 {16}\normf{ZB\transpose{Z}-Y_0}^2 +  4 \normf{Y_1-Y_2}^2 \\
      &&&  +  f(Z_0;B_0, Y_2) - f(Z; B,Y_2)\\
      & \le && \tfrac 1 {16}\normf{ZB\transpose{Z}-Y_0}^2 +  4 \normf{Y_1-Y_2}^2 + R \cdot\normo{\Yin-Y_2} \\
      &&& +  f(Z_0;B_0, Y_2) - f(Z; B,\Yin)\\
      & \le && \tfrac 1 {16}\normf{ZB\transpose{Z}-Y_0}^2 +  4 \normf{Y_1-Y_2}^2 + R \cdot\normo{\Yin-Y_2} \\
      &&& +  v
    \end{aligned}
  \]
  By \cref{lem:basic-sos-utility} (applied to matrices \(Y_0\) and \(Y_1\)),
  \[
    \begin{aligned}
      \cA_1(Z), f(Z;B,\Yin)\ge t\sststile{O(1)}{Z}\\
      \normf{Z B \transpose Z - Y_0}^2
      & \le && 48 k \cdot \norm{Y_1 - Y_0}^2
             + 4\Bigparen{f(Z_0;B_0,Y_1) - f(Z; B, Y_1)}\\
      & \le && 48 k \cdot \norm{Y_1 - Y_0}^2+  \tfrac 1 {4}\normf{ZB\transpose{Z}-Y_0}^2 \\
      &&& +  16 \normf{Y_1-Y_2}^2 + 4R \cdot\normo{\Yin-Y_2} +  4v
          \,,
    \end{aligned}
  \]
  as desired.
\end{proof}

For the sum-of-squares proof above, it is important that in the constraint \(f(Z; B, \Yin) \ge t\)  only \(Z\) is a variable for the proof system and that, in particular, \(\Yin\) is a constant for the proof system.
However, the polynomial system \(\cA(Y,Z;B,\Yin,t)\) that we consider for our algorithms only contains a constraint \(\Set{f(Y,Z; B) \ge t}\) where \(Y\) is a variable for the proof system.
The following lemma shows that the system \(\cA(Y, Z; B, \Yin,t)\) satisfies a monotonicty property that allows us to sos-derive \(f(Z; B, \Yin) \ge t\) from it.

\begin{lemma}\torestate{\label{lem:abstractextension}
  For all matrices \(\Yin \in \R^{n\times n}\) and every matrix \(B\) with nonnegative entries,
  \[
    \cA(Y, Z; B, \Yin, t) \sststile{4}{Y,Z}
    f(Z; B, \Yin) \ge t\,.
  \]
}
\end{lemma}
\begin{proof}
  Since
  \[
    f(Z;B,\Yin)-t = f(Y,Z;B)-t + \Iprod{ZB\transpose{Z},\Yin-Y}
  \]
  and we have the degree-4 constraint $f(Y,Z;B)-t \ge 0$, it suffices to show $\iprod{ZB\transpose{Z},\Yin-Y} \ge 0$.
  For all $i,j$, we have $\cA_1(Z) \sststile{2}{} (ZB\transpose{Z})(i,j)\ge0$ as
  $(ZB\transpose{Z})(i,j) = \sum_{a,b} B(a,b)Z(i,a)Z(j,b)$.
  Combining the degree-2 sos proof $ZB\transpose{Z}\ge0$ with the degree-1 constraint $\Yin\ge Y$ in $\cA_2(Y;\Yin)$, 
  \[
    \cA_1(Z) \cup \cA_2(Y;\Yin) \sststile{3}{} \Iprod{ZB\transpose{Z},\Yin-Y} \ge 0 \,.
  \]
  Therefore
  \[
    \cA(Y,Z;B,\Yin,t) \sststile{4}{} f(Z;B,\Yin) \ge t \,.
  \]
\end{proof}

The final part of our utility analysis is the following connection between the Frobenius norm distance between \(n\)-by-\(n\) matrices \(Z B \transpose Z\) and \(\Znull \Bnull \Znulltrans\) and the graphon \(\delta_2\) distance of the \(k\)-by-\(k\) matrices \(B\) and \(\Bnull\).

\begin{lemma}
  \torestate{
    \label{lem:abstractrounding}
    Let \(B,B_0\) be $k$-by-$k$ matrices with nonnegative entries and let \(t\in\R_{\ge 0}\) be a scalar.
    Suppose the following polynomial system in \(Z\) is sos-consistent up to level \(k^{O(1)}\),
    \[
      \cA_1(Z) \cup \cA_1(\Znull) \cup \Set { \tfrac 1 {n^2}\normf{ZB\transpose Z - Z_0 B_0 \transpose {Z_0}}^2 \le t}\,.
    \]
    Then, \(\deltatwosquared(B,B_0)\le 1.1 \cdot t \).
  }
\end{lemma}

\begin{proof}
  By \cref{lem:dsdistance}, the \(\delta_2\) distance has the following alternative characterization in terms of \(k\)-by-\(k\) doubly-stochastic matrices,
  \[
    \deltatwosquared(B,B_0) = \min_{S\in B_k} p(S; B,B_0)\,.
  \]
  Here, \(B_k\) denotes the set of all \(k\)-by-\(k\) doubly-stochastic matrices (also known as the Birkhoff polytope), and \(p(S;B,B_0)\) denotes the following quadratic polynomial in \(S\) with coefficients depending on \(B,B_0\),
  \[
    p(S; B,B_0) \seteq \frac{1}{k^2}\sum_{a,a',b,b'\in [k]}\Paren{B(a,b)-\Bnull(a',b')}^2\cdot S(a,a')\cdot S(b,b')\,.
  \]
  Relatedly, we can embed the set of community membership matrices \(\cZ(n,k)\) into the Birkhoff polytope \(B_k\), by the following quadratic map \(Z\mapsto S(Z)\),
  \[
    S(Z) \seteq \frac{k}{n} \transpose Z \Znull\,.
  \]
  Furthermore, this map has the property that \(p\bigparen{S(Z); B,B_0}=\tfrac 1{n^2}\normf{ZB\transpose{Z}-\Znull\Bnull\Znulltrans}^2\) for all \(Z,Z_0\in \cZ(n,k)\).

  We show in \cref{lem:sos-proof-properties-of-map-between-community-membership-and-doubly-stochastic} that these properties of the map \(S(Z)\) also have sos proofs,
  \[
    \cA_1(Z),\cA_1(\Znull) \sststile{O(1)}{Z,\Znull} \cA_{\mathrm{ds}}(S(Z)),p\bigparen{S(Z); B,B_0}=\tfrac 1{n^2}\normf{ZB\transpose{Z}-\Znull\Bnull\Znulltrans}^2\,,
  \]
  where \(\cA_{\mathrm{ds}}(S)\seteq \Set{ S \ge 0,~ S\Ind_k =\Ind_k,~\transpose S \Ind_k=\Ind_k }\) is the set of linear constraints describing \(B_k\).

  Thus, since the system \(\cA_1(Z)\cup \cA_1(\Znull)\cup \set{\tfrac 1{n^2}\normf{ZB\transpose{Z}-\Znull\Bnull\Znulltrans}^2\le t}\) is sos-consistent up to level \(k^{O(1)}\), then so is the system \(\cA_{\mathrm{ds}}(S)\cup \{p(S;B,B_0)\le t\}\) (with a small additive loss in the exact number of levels).

  In \cref{section:sos-optimization-polytope}, we show the sum-of-squares provides good multiplicative approximations for the problem of minimizing quadratic polynomials with nonnegative coefficients over integral polytopes.
  Concretely, \cref{corollary:sos-multiplicative-approx-polytope} shows that the sos-consistency of  the system \(\cA_{\mathrm{ds}}(S)\cup \{p(S;B,B_0)\le t\}\) up to level \(k^{O(1)}\) implies that there exists a doubly-stochastic matrix \(S^*\in B_k\) with \(p(S^*;B,B_0)\le 1.1t\).
  We conclude \(\deltatwosquared(B,B_0)\le 1.1t\) as desired.
\end{proof}

The following lemma allows us to bound the sensitivity of our score functions, which is the main ingredient for the privacy analysis of our algorithms.

\begin{lemma}[Sensitivity bound]\torestate{\label{lem:abstractsensitivity}
  For every two matrices \(\Yin,\Yin'\in\R^{n\times n}\) that differ in at most one row and one column, there exist linear polynomials \(L=(L_{ij})\) such that for every polynomial inequality \(p(Y,Z)\ge0\) in \(\cA(Y, Z; B, \Yin', t-\Delta)\),
  \[
    \cA(Y, Z; B, \Yin, t)
    \sststile{\deg(p)}{}
    p(L(Y),Z)
  \]
  where \(\Delta = 40nR\normm{B}\).
}
\end{lemma}
In particular, by \cref{lem:lip_sos-sos_consistency}, the property established in the above lemma implies that, whenever the system \(\cA(Y, Z; B, \Yin, t)\) is sos-consistent up to level \(\ell\), then so is the system \(\cA(Y, Z; B, \Yin', t-\Delta)\) (under the conditions of the lemma).
\begin{proof}
  Without loss of generality, we can assume $\Yin$ and $\Yin'$ differ on the first row and column.
  Consider the following linear functions
  \[
    L_{ij}(Y) = 
    \begin{cases}
      Y(i,j) & \text{if } i,j>1 \,, \\
      0 & \text{otherwise} \,.
    \end{cases}
  \]
  Sos proofs for inequalities in $\cA_1(Z) \cup \cA_2(Y;\Yin')$ are straightforward.
  For $f(L(Y),Z;B)\ge t-\Delta$, observe that
  \begin{align*}
    f(L(Y),Z;B) - (t-\Delta)
    &= \Paren{f(Y,Z;B) - t} + \Paren{\Delta - \Iprod{ZB\transpose{Z},Y-L(Y)}} \\
    &= \Paren{f(Y,Z;B) - t} + \Paren{\Delta - 2\Iprod{(ZB\transpose{Z})(1,\cdot),Y(1,\cdot)}} \,.
  \end{align*}
  For all $i,j$, we have $\cA_1(Z) \sststile{2}{} (ZB\transpose{Z})(i,j) \le \normm{B}$ as
  \begin{align*}
    (ZB\transpose{Z})(i,j) 
    &= \sum_{a,b} B(a,b)Z(i,a)Z(j,b) 
      \le \normm{B} \cdot \sum_{a,b} Z(i,a)Z(j,b) \\
    &= \normm{B} \cdot \Paren{\sum_a Z(i,a)} \Paren{\sum_b Z(j,b)}
      = \normm{B} \,.
  \end{align*}
  Then,
  \[
    \cA_1(Z) \cup \cA_2(Y;\Yin)
    \sststile{3}{}
    \Iprod{(ZB\transpose{Z})(1,\cdot),Y(1,\cdot)} 
    \le \normm{B}\cdot\sum_{j}Y(1,j) 
    \le 20nR\normm{B} \,.
  \]
  where the first inequality uses $Y\ge0$ and the second uses $Y\Ind \leq 20nR\cdot\Ind$ in $\cA_2(Y;\Yin)$.
  Since we have degree-4 constraint $f(Y,Z;B)-t \ge 0$, then
  \[
    \cA(Y, Z; B, \Yin, t)  
    \sststile{4}{}
    f(L(Y),Z;B)\ge t-\Delta \,.
  \]
\end{proof}

\begin{theorem}\torestate{\label{thm:abstract-main-theorem}
    For every \(R\ge 1\) and \(k\ge 2\), there exists an $\epsilon$-differentially private algorithm with the following utility guarantee:
    For every input matrix $\Yin\in\R_{\ge 0}^{n\times n}$ and all matrices \(Y_0,Y_1,Y_2\) satisfying \(\cA_2(Y_2; \Yin)\) and \(Y_0=\Znull \Bnull \transpose \Znull\) for some \(\Znull\in\cZ(n,k)\) and \(\Bnull\in\R_{\ge 0}^{k\times  k}\) with \(\normmax{\Bnull}\le R\), the randomized output \(\hat{\bm B}\in \R_+^{k\times k}\) of the algorithm satisfies with probability at least $1-n^{-\Omega(k^2)}$,
    \begin{equation*}
      \deltatwosquared(\hat{\bm B},B_0)
      \lesssim  \frac{1}{n^2}\cdot \Paren{k \cdot\norm{Y_0 - Y_1}^2 + \normf{Y_1-Y_2}^2+ R\cdot\normo{Y_2-\Yin} } + \frac{R^2 k^2\log(n)}{\epsilon n}\,.
    \end{equation*}
    Furthermore, the algorithm runs in time \(n^{\poly(k)}\).
}
\end{theorem}

\begin{proof}
  For any symmetric matrix $B\in [0,R]^{k\times k}$ , consider the score function \(\Score(B;\Yin)\) which is defined as
  \begin{equation*}
    \min \Set{t\in \R:\exists \text{ level-}k^{10}\text{ pseudo-distribution } \tilde\mu \text{ satisfying } \cA(Y, Z; B, \Yin, t)}
  \end{equation*}
  and the corresponding exponential mechanism \[\hat{\bm B}\sim \exp\Paren{\frac{\epsilon\cdot \Score(B;\Yin)}{\Delta}}\,,\] where $\Delta=40R^2 n$.

  By \cref{lem:abstractsensitivity}, the sensitivity of the score function is upper bounded by $40R^2 n$.
  This sensitivity bound together with \cref{thm:exp-mech} implies that the exponential mechanism is $\e$-differentially node private.
  Moreover, for each matrix \(B\), it takes \(n^{\poly(k)}\) time to evaluate the score function \(\Score(B;\Yin)\).
  Therefore by taking \(\eta=n^{-100}\) in \cref{thm:exp-mech} (exponential mechanism), our algorithm runs in time \(n^{\poly(k)}\).

  It remains to analyze the utility of this exponential mechanism.
  To this end let \(\hat{\bm {B}}\) denote the randomized output of the mechansim for input \(\Yin\).
  Since our score function is $L$-Lipschitz with $L\leq O(n^{10})$.
  By taking net size $\eta=n^{-100}$ and probability bound $\delta=n^{-\Omega(k^2)}$, in \cref{thm:exp-mech} (utility guarantees of exponential mechanism), the difference of the scores \(\bm t=\Score(\hat{\bm B};\Yin)\) and \(t_0=\Score(\Bnull;\Yin)\) is at most $O\Paren{\frac{n k^2\log(n)}{\epsilon}}$ with probability at least $1-n^{-\Omega(k^2)}$.

  Let $\ell\geq k^{10}$ and \(\mu\) be the level-\(\ell\) pseudo distribution that witnesses that the system \(\cA(Y,Z;  \bm B, \Yin, t)\) is sos-consistent up to level \(\ell\).
  Let \(Y_2\) be the matrix obtained from \(\Yin\) by removing the rows and columns averaging to more than \(20R\).
  By \cref{lem:abstractextension}, the pseudo-distribution \(\mu\) satisfies the constraint \(f(Z; \hat{\bm B}, \Yin)\ge \Score(\hat{\bm B};\Yin)=\bm t\).
  Therefore, by \cref{lem:abstractidentifiability}, the pseudo-distribution \(\mu\) also satisfies the constraint
  \begin{equation*}
    \normf{Z \hat{\bm B} \transpose Z - Y_0}^2 \lesssim {k \cdot \norm{Y_0 -Y_1}^2 + \normf{Y_1-Y_2}^2 + \normm{B}\normo{Y_2-\Yin}+(t_0-\bm t)}\,.
  \end{equation*}
  Thus, the pseudo-distribution \(\mu\) witnesses that the following polynomial system in variables \(Z\) is sos-consistent up to level $k^{10}$,
  \[
    \cA_1(Z), \normf{Z \hat{\bm B} \transpose Z - \Znull \Bnull \transpose \Znull}^2 \le O\Paren{k \cdot \norm{Y_0 - Y_1}^2 + \normf{Y_1-Y_2}^2 + \normm{B} \normo{Y_2-\Yin}+(t_0-\bm t)}\,.
  \]
  By \cref{lem:abstractrounding}, we can conclude the desired bound
  \begin{equation*}
    \deltatwosquared(\hat{\bm B},B_0)\leq \frac{1}{n^2}\cdot O\Paren{k \cdot\norm{Y_0 - Y_1}^2 + \normf{Y_1-Y_2}^2 + R\normo{Y_2-\Yin}+(t_0-\bm t)}\,.
  \end{equation*}
  Since as proved, with probability at least \(1-n^{-\Omega(k^2)}\), we have \(t_0-\bm t\leq O\Paren{\frac{n k^2\log(n)}{\epsilon}}\), it follows that
    \begin{equation*}
        \deltatwosquared(\hat{\bm B},B_0)\leq \frac{1}{n^2}\cdot O\Paren{k \cdot\norm{Y_0 - Y_1}^2 + \normf{Y_1-Y_2}^2 + R\cdot\normo{Y_2-\Yin}+\frac{R^2 nk^2\log(n)}{\epsilon}}\,.
    \end{equation*}
\end{proof}

\begin{algorithmbox}[Private estimation algorithm for low rank matrix estimation]
  \label{algo:privateMatrix}
  \mbox{}\\
  \textbf{Input:} matrix $\Yin$, rank $k$, privacy parameter $\epsilon>0$, and a bound $R\geq \normm{B}$.

  \noindent
  \textbf{Output:} .
  \begin{enumerate}[1.]
    \item For each $B\in [0,R]^{k\times k}$, let $\Score(B;\Yin)$ be defined as
    \begin{align*}
      \min \Set{t\in \R:\exists \text{ level-}k^{10}\text{ pseudo-distribution } \tilde\mu \text{ satisfying } \cA(Y, Z; B, \Yin, t)}\,.
    \end{align*}
    \item Output $\hat{\bm B}\in [0,R]^{k\times k}$, which is sampled from the distribution
      \begin{equation*}
        \hat{\bm B}\propto \exp\Paren{\frac{\epsilon}{\Delta} \Score(B;\Yin)}
      \end{equation*}
      where $\Delta=40nR^2$\,.
    \end{enumerate}
\end{algorithmbox}

\begin{remark}[Numerical issues]
  \label{remark:numerical-issues}
  Since we only have efficient algorithms for solving semidefinite programing up to a given precision, we can only efficiently search for pseudo-distributions that approximately satisfy a given polynomial system.
  Actually, in the first step of \cref{algo:privateMatrix}, we are binary searching for the minimum $t$ such that there exists a level-$k^{10}$ pseudo-distributions \emph{approximately} satisfying $\cA(Y, Z; B, \Yin, t)$.
  The analysis of our algorithm based on sos proofs still works due to our discussion in \cref{remark:sos-numerical-issue}.
\end{remark}

\subsection{Private estimation for balanced stochastic block models}
In this section, we prove our results for stochastic block models.
\restatetheorem{thm:mainSBM}

We describe our algorithm in \cref{algo:privateSBM}.
\begin{algorithmbox}[Private estimation algorithm for perfectly balanced stochastic block models]
    \label{algo:privateSBM}
    \mbox{}\\
    \textbf{Input:} Adjacency matrix $A$, number of blocks $k$, privacy parameter $\epsilon>0$, and a bound $R\geq 20 \normm{B_0}$.

    \noindent
    \textbf{Output:} A matrix $\hat{\bm B}\in [0,R]^{k\times k}$.
    \begin{enumerate}[1.]
    \item Run \cref{algo:SBMDensityEstimation}, and obtain private edge density estimator $\hat{\bm \rho}\in [0,1]$.
    \item Run \cref{algo:privateMatrix} with $\bm\Yin=\frac{\bm A}{\hat{\bm \rho}}$, and return the result $\hat{\bm B}$.
    \end{enumerate}
\end{algorithmbox}

\subsubsection{Guarantees of private edge density estimation}
We first prove the guarantees of the private edge density estimation algorithm in the setting of stochastic block model.
\begin{lemma}\label{lem:densityestimationSBM}
  Under the setting of \cref{thm:mainSBM}, with probability at least $1-O(\frac{R}{n})$, we have
  \begin{equation*}
    \Abs{\hat{\bm \rho}-\rho(\bm G)}^2\leq  \polylog(n)\cdot O\Paren{\frac{R^2\rho^2}{\e^2 n^2}+ \frac{1}{\e^4 n^4}}
  \end{equation*}
\end{lemma}
We describe the algorithm in \cref{algo:SBMDensityEstimation}
\begin{algorithmbox}[Private algorithm for estimation of target edge density]
  \label{algo:SBMDensityEstimation}
  \mbox{}\\
  \textbf{Input:} Adjacency matrix $A$, privacy parameter $\epsilon>0$, and a bound $R\geq 20 \Lambda =20 \normm{B_0}$.

  \noindent
  \textbf{Output:} $\e$-differentially private target edge density estimator $\hat{\bm \rho}$.
  \begin{enumerate}[1.]
  \item Let $\hat{\bm \rho}_{c}=\rho(G)+\Lap\Paren{\frac{10}{n\epsilon}}$, and $\hat{\bm \rho}_{u}=\hat{\bm \rho}_{c}+\frac{100\log(n)}{n\epsilon}$.
  \item Let $D=10R \hat{\bm \rho}_{u} n \log(n)$, run $\frac{\e}{2}$-differentially private density estimation \cref{algo:privateDensityEstimation} and return $\hat{\bm \rho}$.
  \end{enumerate}
\end{algorithmbox}
\begin{proof}[Proof of \cref{lem:densityestimationSBM}]
  By the promise of Laplacian mechanism, $\hat{\bm \rho}_c$ is $\frac{\epsilon}{2}$-differentially private. 
  Moreover by \cref{cor:privateDensitydegreeconcentrated}, for fixed $\hat{\bm \rho}_c$, $\hat{\bm \rho}$ is $\frac{\epsilon}{2}$-differentially private.
  Therefore by composition theorem, \cref{algo:SBMDensityEstimation} is $\e$-differentially private.

  By the \cref{lem:private_edge_density}, with probability at least $1-\frac{1}{n^{10}}$, we have $\Abs{\hat{\bm \rho}_c-\rho(\bm Q_0)}\leq \frac{100\log(n)}{n\epsilon}+\frac{\rho(\bm Q_0)}{10}$.
  By the \cref{lem:rhoQtorho}, with probability $1-O(\frac{R}{n})$, we have $\Abs{\rho(\bm Q_0)-\rho}\leq \frac{1}{10}\rho$. 
  Therefore with probability at least $1-O(\frac{R}{n})$, we have
  \begin{align*}
      \frac{1}{2}\rho\leq \hat{\bm \rho}_{c}+\frac{100\log(n)}{n\epsilon} \leq 2\rho+\frac{200\log(n)}{n\epsilon} \,. 
  \end{align*}
  As result, with probability at least $1-O(\frac{R}{n})$, the degree of graph $\bm G$ is bounded by $10R \hat{\bm \rho}_{u} n \log(n)$.

  Now applying \cref{cor:privateDensitydegreeconcentrated}, with probability at least $1-O(\frac{R}{n})$, we have
  \begin{equation*}
      \Abs{\hat{\bm \rho}-\rho(\bm G)}^2\leq O\Paren{\frac{\log^2(n) D^2}{\e^2 n^4}+\frac{\log^4(n)}{\e^4 n^4}}\leq \polylog(n)\cdot O\Paren{\frac{R^2\rho^2}{\e^2 n^2}+ \frac{1}{\e^4 n^4}}\,.
  \end{equation*}
\end{proof}

As result, we obtain the following corollary, which is easier to apply for obtaining our utility guarantees:
\begin{corollary}\label{corollary:target-density-estimation-sbm}
  Under the setting of \cref{thm:mainSBM}, let $\rho=\frac{d}{n}$, we can output a $\e$-differentially private estimator $\hat{\bm \rho}$ such that with probability at least $1-O\Paren{\frac{R}{n}}$, $\Abs{\hat{\bm \rho}-\rho}\leq \frac{\rho}{10}$ and
  \begin{equation*}
      \Abs{\hat{\bm \rho}-\rho}^2\leq \frac{\rho^2}{n\e}+\frac{\rho\polylog(n)}{n^2}\,.
  \end{equation*}
  Furthermore, the algorithm runs in  $\poly(n)$ time.
\end{corollary}
\begin{proof}
  By taking $t=\sqrt{\frac{10\rho(\bm Q_0)\log n}{n^2}}$ in \cref{lem:convergence-edge-density}, we have $\Pr\Brac{\Paren{\rho(\bm Q_0)-\rho(\bm G)}^2\geq t}\leq \frac{1}{n^{10}}$. 
  Combined with \cref{lem:densityestimationSBM}, with probability at least $1-O\Paren{\frac{\Lambda}{n}}$,
  \begin{equation*}
      \Abs{\hat{\bm \rho}-\rho}^2\leq  \polylog(n)\cdot O\Paren{\frac{R^2\rho^2}{\e^2 n^2}+ \frac{1}{\e^4 n^4}}+\frac{10\rho\log(n)}{n^2}\,.
  \end{equation*}
  Under the assumption that $R\leq \sqrt{\frac{n\epsilon}{\polylog(n)}}$ and $\e^4 n^2\geq \polylog(n)$, we have 
\begin{equation*}
  \Paren{\hat{\bm \rho}-\rho}^2\leq \frac{\rho^2}{n\e}+\frac{\rho\polylog(n)}{n^2}\,.
\end{equation*}
As  $\frac{\rho^2}{n\e}+\frac{\rho\polylog(n)}{n^2}\leq \frac{1}{100}\rho^2$, we have $\Abs{\hat{\bm \rho}-\rho}\leq \frac{\rho}{10}$.
\end{proof}

\subsubsection{Proof of error rate for balanced stochastic block models}

Now we finish the proof of \cref{thm:mainSBM}.
\begin{proof}
We first show that the \cref{algo:private} is \(\e\)-differentially private.
Indeed, by the promise of Laplacian mechanism, our estimator \(\hat{\bm \rho}\) is \(\frac{\e}{2}\)-differentially private.
For every \(\hat{\bm \rho}\in \R\), \(\hat{\bm B}\) is \(\frac{\e}{2}\)-differentially private.
By the basic composition theorem~\cite{dwork2014algorithmic}, \cref{algo:private} is $\e$ differentially private.

Let the average edge density $\rho=\frac{d}{n}$. 
By \cref{lem:densityestimationSBM}, with probability at least $1-O(\frac{R}{n})$ we have $\frac{1}{2}\rho\leq \hat{\bm \rho}\leq 2\rho$.

Let $\bm \Qnull=\bm Z_0B_0\bm Z_0^\top$ be the edge connection probability matrix.
Let $S\subseteq [n]$ be the set of vertices with degree larger than $20R d$.
We let $\bar{\bm A}$ be the adjacency matrix of the graph obtained by removing edges incident to vertices in $S$, and $\bar{\bm \Qnull}\in [0,1]^{n\times n}$ be the edge connection probability matrix restricted to $S$.

Since $\bm W=\bm A-\bm Q_0$ is a symmetric random matrix with
\begin{equation*}
    \bm W(i,j) =
      \begin{cases}
        1 - \bm Q_0({i,j}) & \text{w.p. } \bm Q_0({i,j})\\
        - \bm Q_0({i,j}) & \text{w.p. } 1- \bm Q_0({i,j})
      \end{cases}
\end{equation*}
Furthermore, we have $\bm Q_0(i,j)\leq R\cdot \rho$ by definition.
By \cref{theorem:pruned-spectral-norm}, with probability at least \(1-\frac{1}{n^2}\), we have \(\norm{\bar{\bm A}- \bar{\bm \Qnull}}\leq O\Paren{\sqrt{R\rho n}}\).
By \cref{lem:degree-counting}, with probability at least \(1-\exp(-R\rho n)\), we have \(\normo{\bar{\bm A}-\bm A}\leq R \rho n^2 \cdot\exp(-R \rho n)\).

We let $\bm \Yin=\frac{\bm  A}{\hat{\bm\rho}}$, $\bm Y_2=\frac{\bar{A}}{\hat{\bm \rho}}$, $\bm Y_1=\frac{\bar{\bm Q_0}}{\hat{\bm\rho}}$ and $\bm Y_0=\bm Z_0B_0\bm Z_0^\top$.
Furthermore, we let $\bm t_0 = \iprod{\bm Y_0,\bm \Yin}- \tfrac 12\normf{\bm Y_0}^2$ and $t_1 = \iprod{ZBZ^\top,\bm \Yin}- \tfrac 12\normf{ZBZ^\top}^2$.
Since $\cA_2(\bm Y_2,\bm \Yin)$ is satisfied, we can apply \cref{thm:abstract-main-theorem}, and get
\begin{equation*}
    \deltatwosquared(\bm B,B_0)\leq \frac{1}{n^2}\cdot O\Paren{k \cdot\norm{\bm Y_2 -\bm Y_1}^2 + \normf{\bm Y_1-\bm Y_0}^2 + R\normo{\bm Y_2-\bm \Yin}+\frac{R^2 nk^2\log(n)}{\epsilon}}\,.
\end{equation*}

Therefore, since $\norm{\bm Y_2-\bm Y_1}^2\leq \frac{R \rho n}{\hat{\bm \rho}^2}\leq O\Paren{\frac{R n}{\rho}}$, and
\begin{equation*}
    \normo{\bm Y_2-\bm \Yin}\leq \frac{1}{\hat{\bm \rho}}\normo{\bar{\bm A}-\bm A}\leq \frac{R \rho n^2}{\hat{\bm \rho}} \cdot\exp\Paren{-R \rho n}\leq 2R n^2 \exp(-R \rho n)\,.
\end{equation*}
Furthermore, by \cref{corollary:target-density-estimation-sbm}, with probability at least $1-O(\frac{\Lambda}{n})$,
\begin{equation*}
  \Paren{\hat{\bm \rho}-\rho}^2\leq \frac{\rho^2}{n\e}+\frac{\rho\polylog(n)}{n^2}\,.
\end{equation*}

As result, with probability at least $1-\frac{1}{(Rd)^{100}}-O\Paren{\frac{R}{n}}$, we have
\begin{align*}
    \normf{\bm Y_1-Y_0}^2 &\leq \Normf{\frac{\bar{\bm Q}_0}{\hat{\bm \rho}}-\frac{\bm Q_0}{\rho}}^2\\
    &\leq \frac{2}{\hat{\bm \rho}^2}\Normf{\bm Q_0-\bar{\bm Q}_0}^2+ 2\Normf{\bm Q_0}^2 \cdot \Paren{\frac{\hat{\bm \rho}-\rho}{\rho\hat{\bm \rho}}}^2\\
    &\leq \exp(-Rd)\cdot \frac{2R^2 d^2}{\hat{\bm \rho}^2}+\frac{R^2 n^2}{\rho^2}\cdot\Paren{\frac{\rho^2}{n\e}+\frac{\rho\polylog(n)}{n^2}}\\
    &\leq \exp(-Rd)\cdot 8R^2 n^2+R^2n^2\Paren{\frac{1}{n\e}+\frac{\polylog(n)}{\rho n^2}}
\end{align*}

In conclusion, with probability at least $1-\exp(-R\rho n)-O(\Lambda/n)$ we have
\begin{align*}
    \deltatwosquared(\bm B,B_0)&\lesssim \frac{1}{n^2}\cdot \Paren{k \cdot\norm{\bm Y_2 -\bm Y_1}^2 + \normf{\bm Y_1-\bm Y_0}^2 + R\normo{\bm Y_2-\bm \Yin}+\frac{R^2 nk^2\log(n)}{\epsilon}}\\
    &\lesssim \frac{R k}{d}+R^2 \exp(-Rd)+R^2\Paren{\frac{1}{n\e}+\frac{\polylog(n)}{d n}}+\frac{R^2 k^2\log(n)}{n\epsilon}\\
    & \lesssim \frac{R k}{d}+\frac{R^2 k^2\log(n)}{n\epsilon}\,.
\end{align*}
\end{proof}

\subsection{Private estimation for graphon}
\subsubsection{Preliminaries for graphon estimation}
\label{sec:graphon-preliminaries}
A graphon is a bounded and measurable function $W:[0,1]^2\to \R_{+}$ such that $W(x,y)=W(y,x)$, which is said to be normalized if $\int W=1$.
Given a target edge density $\rho$, and a normalized graphon $W\in [0,1]^2\to \R_+$, a $(W,\rho,n)$-random graph on $n$ vertices is sampled in the following way.
For each vertex $i\in [n]$ in the graph, a real value $\bm x_i$ is sampled uniformly at random from $[0,1]$.
Then each pair of vertices $i,j$ is connected independently with probability $\bm\Qnull(i,j) \seteq \rho\cdot W(\bm x_i,\bm x_j)$.
We call $\bm\Qnull$ the edge connection probability matrix.

For any given graphon $W$ and a measure-preserving mapping $\phi:[0,1]\to[0,1]$ (with regard to the Lebesgue measure), $W^\phi(x,y) \seteq W(\phi(x),\phi(y))$ define the same distribution on graphs as $W(x,y)$.
Thus we consider the following distance between graphons.

\begin{definition}[$\delta_2$ distance between graphons]
  \label{def:delta_2}
  The $\delta_2$ distance between two graphons $W_1,W_2$ is
  \[
    \delta_2(W_1, W_2) \seteq
    \inf_{\substack{\phi \from [0,1]\to [0,1]\\ \text{measure preserving}}} 
    \Normt{W_1^{\phi}-W_2}\,,
  \]
  where $W_1^{\phi}(x,y) \seteq W_1(\phi(x),\phi(y))$.
\end{definition}

\paragraph{Block graphons}
Given a symmetric and nonnegative $k\times k$ matix $B$, we can define a $k$-block graphon $W[B]$ as follows. 
Let $(I_1,I_2,\ldots,I_k)$ be the partition of $[0,1]$ into adjacent and disjoint intervals of lengths $1/k$. 
Let $W[B]:[0,1]^2\to\R_{+}$ to be the step function that equals $B(i,j)$ on $I_i\times I_j$ for every $i,j\in[k]$.

Given an arbitrary graphon $W$, when we approximate it by $k$-block graphons, the approximation error is given by
\begin{equation}
  \epsilon_k^{(O)}(W)\coloneqq \min_{B\in \R_{+}^{k\times k}} \Norm{W[B]-W}_2 \,.
\end{equation}
It is easy to see $\epsilon_k^{(O)}(W)\to 0$ as $k\to \infty$.

\subsubsection{Private algorithm for graphon estimation}\label{sec:graphon}
In this section, we prove our main result for graphon estimation.
The proof is similar to the setting of the stochastic block model, but now we need to tackle agnostic error induced by model misspecification.
\begin{theorem}[Main theorem for private graphon estimation]\torestate{\label{thm:formalmaintheorem}
    For any $\rho\in [0,1],R\in \R_+$ and arbitrary graphon $W:[0,1]\times [0,1]\to [0,R]$, suppose we are given $n$-vertex $(W,\rho,n)$ random graph $G$ and $R\in\R_+$ such that $R\geq 8 \Lambda$. 
    Suppose for privacy parameter $\e>0$, $\normm{W}\leq R\leq \sqrt{\frac{n\epsilon}{\polylog(n)}}$ for some $R\in \R_+$, and $\e^4 n^2\geq \polylog(n)$.
    Then for any positive integer $k\in \Z_+$, \cref{algo:private} outputs a symmetric matrix $\hat{\bm B}\in [0,R]^{k\times k}$ such that
    \begin{equation*}
        \E\delta_2^2(W[\hat{\bm B}],W)\lesssim \frac{R k}{\rho n}+\frac{R^2 k^2\log(n)}{n\epsilon}+R^2\sqrt{\frac{k}{n}}+\Paren{\epsilon_k^{(O)}(W)}^2\,,
    \end{equation*}
    where $\epsilon_k^{(O)}(W)$ represents the approximation error:
    \begin{equation*}
        \epsilon_k^{(O)}(W)\coloneqq \min_{B\in [0,R]^{k\times k}} \norm{W[B]-W}_2\,.
    \end{equation*}
    Furthermore \cref{algo:private} is \(\e\)-differentially private, and runs in $n^{\poly(k)}$ time.
}
\end{theorem}

The previous work \cite{borgs2018revealing} achieves error rate
\begin{equation*}
        \E\delta_2^2(W[B],W)\lesssim\frac{Rk^2}{n^2}+\frac{R\log(k)}{\rho n}+\frac{R^2k^2\log(n)}{n\epsilon}+R^2\sqrt{\frac{k}{n}}+\frac{R^2}{n^2\rho^2\e^2}+\Paren{\epsilon_k^{(O)}(W)}^2\,.
\end{equation*}
with an exponential time algorithm.
On the other hand, for the non-private error rate, \cite{luo2023computational} provides evidence that obtaining better guarantees than $O(\frac{k}{\rho n})$ is inherently hard for polynomial time algorithms, based on lower bound against low degree polynomial estimators.
Our result matches the privacy error terms in \cite{borgs2018revealing}, and matches the lower bound of non-private error rate provided in \cite{luo2023computational}.
The algorithm is essentially the same as stochastic block models, and is described in \cref{algo:private}.
\begin{algorithmbox}[Private estimation algorithm for graphon estimation]
  \label{algo:private}
  \mbox{}\\
  \textbf{Input:} Adjacency matrix $A$, number of blocks $k$, privacy parameter $\epsilon>0$, and a bound of underlying graphon $R\geq 20 \Lambda =20 \normi{W}$.

  \noindent
  \textbf{Output:} A matrix $\hat{\bm B}\in [0,R]^{k\times k}$, with associated graphon estimator $W[\hat{\bm B}]:[0,1]\times [0,1]\to \R_+$.
  \begin{enumerate}[1.]
  \item  Run \cref{algo:graphonDensityEstimation}, and obtain private edge density estimator $\hat{\bm \rho}\in [0,1]$.
  \item Run \cref{algo:privateMatrix} with $\Yin=\frac{A}{\hat{\bm \rho}}$, and return the result $\hat{\bm B}$.
  \end{enumerate}
\end{algorithmbox}

\subsubsection{Guarantees of private edge density estimation}
We first prove the guarantees of the private edge density estimation algorithm in the setting of graphon.
\begin{lemma}\label{lem:densityestimationgraphon}
  Under the setting of \cref{thm:mainSBM}, With probability at least $1-O(\frac{R}{n})$, we have $\Abs{\hat{\bm \rho}-\frac{d}{n}}\leq \frac{d}{10n}$.
  Furthermore, with probability at least $1-O(\frac{R}{n})$, we have
  \begin{equation*}
    \Abs{\hat{\bm \rho}-\rho(\bm G)}^2\leq  \polylog(n)\cdot O\Paren{\frac{R^2\rho^2}{\e^2 n^2}+ \frac{1}{\e^4 n^4}}
  \end{equation*}
\end{lemma}
The algorithm is the same as \cref{algo:SBMDensityEstimation}
\begin{algorithmbox}[Private algorithm for estimation of target edge density]
\label{algo:graphonDensityEstimation}
  \mbox{}\\
  \textbf{Input:} Adjacency matrix $A$, privacy parameter $\epsilon>0$, and a bound of underlying graphon $R\geq 20 \Lambda =20 \normi{W}$.

  \noindent
  \textbf{Output:} $\e$-differentially private target edge density estimator $\hat{\bm \rho}$.
  \begin{enumerate}[1.]
  \item Let $\hat{\bm \rho}_{c}=\rho(G)+\Lap\Paren{\frac{10}{n\epsilon}}$, and $\hat{\bm \rho}_{u}=\hat{\bm \rho}_{c}+\frac{100\log(n)}{n\epsilon}$.
  \item Let $D=10R \hat{\bm \rho}_{u} n \log(n)$, run $\e$-differentially private density estimation \cref{algo:privateDensityEstimation}.
  \end{enumerate}
\end{algorithmbox}
\begin{proof}[Proof of \cref{lem:densityestimationSBM}]
  By the promise of Laplacian mechanism, $\hat{\bm \rho}_c$ is $\frac{\epsilon}{2}$-differentially private. 
  Moreover by  \cref{cor:privateDensitydegreeconcentrated}, for fixed $\hat{\bm \rho}_c$, $\hat{\bm \rho}$ is $\frac{\epsilon}{2}$-differentially private.
  Therefore by composition theorem, \cref{algo:privateDensityEstimation} is $\e$-differentially private.

  By the \cref{lem:private_edge_density}, with probability at least $1-\frac{1}{n^{10}}$, we have $\Abs{\hat{\bm \rho}_c-\rho(\bm Q_0)}\leq \frac{100\log(n)}{n\epsilon}+\frac{\rho(\bm Q_0)}{10}$.
  By the \cref{lem:rhoQtorho}, with probability $1-O(\frac{R}{n})$, we have $\Abs{\rho(\bm Q_0)-\rho}\leq \frac{1}{10}\rho$. 
  Therefore with probability at least $1-O(\frac{R}{n})$, we have
  \begin{align*}
      \frac{1}{2}\rho\leq \hat{\bm \rho}_{c}+\frac{100\log(n)}{n\epsilon} \leq 2\rho+\frac{200\log(n)}{n\epsilon} \,. 
  \end{align*}
  As result, with probability at least $1-O(\frac{R}{n})$, the degree of graph $\bm G$ is bounded by $10R \hat{\bm \rho}_{u} n \log(n)$.

  Now applying \cref{cor:privateDensitydegreeconcentrated}, with probability at least $1-O(\frac{R}{n})$, we have
  \begin{equation*}
      \Abs{\hat{\bm \rho}-\rho(\bm G)}^2\leq O\Paren{\frac{\log^2(n) D^2}{\e^2 n^4}+\frac{\log^4(n)}{\e^4 n^4}}\leq \polylog(n)\cdot O\Paren{\frac{R^2\rho^2}{\e^2 n^2}+ \frac{1}{\e^4 n^4}}\,.
  \end{equation*}
\end{proof}

As result, we obtain the following corollary, which is easier to apply for obtaining our utility guarantees:
\begin{corollary}\label{corollary:target-density-estimation-graphon}
  Under the setting of \cref{thm:mainSBM}, let $\rho=\frac{d}{n}$, we can output a $\e$-differentially private estimator $\hat{\bm \rho}$ such that with probability at least $1-O\Paren{\frac{R}{n}}$, $\Abs{\hat{\bm \rho}-\rho}\leq \frac{\rho}{10}$ and
  \begin{equation*}
      \Abs{\hat{\bm \rho}-\rho}^2\lesssim \frac{\rho^2}{n\e}+\frac{\rho\polylog(n)}{n^2}+\Abs{\rho-\rho(Q_0)}^2\,.
  \end{equation*}
  Furthermore, the algorithm runs in  $\poly(n)$ time.
\end{corollary}
\begin{proof}
  By taking $t=\sqrt{\frac{10\rho(\bm Q_0)\log n}{n^2}}$ in \cref{lem:convergence-edge-density}, we have $\Pr\Brac{\Paren{\rho(\bm Q_0)-\rho(\bm G)}^2\geq t}\leq \frac{1}{n^{10}}$. 
  Combined with \cref{lem:densityestimationgraphon}, with probability at least $1-O\Paren{\frac{\Lambda}{n}}$,
  \begin{equation*}
      \Abs{\hat{\bm \rho}-\rho}^2\lesssim  \polylog(n)\cdot O\Paren{\frac{R^2\rho^2}{\e^2 n^2}+ \frac{1}{\e^4 n^4}}+\frac{10\rho\log(n)}{n^2}+\Abs{\rho-\rho(\bm Q_0)}^2\,.
  \end{equation*}
  Under the assumption that $R\leq \sqrt{\frac{n\epsilon}{\polylog(n)}}$ and $\e^4 n^2\geq \polylog(n)$, we have 
\begin{equation*}
  \Paren{\hat{\bm \rho}-\rho}^2\leq \frac{\rho^2}{n\e}+\frac{\rho\polylog(n)}{n^2}+\Abs{\rho-\rho(\bm Q_0)}^2\,.
\end{equation*}
\end{proof}

Finally we need a result from~\cite{borgs2015private} for bounding the difference between $\rho(\bm Q_0)$ and the target edge density $\rho$.
\begin{lemma}[Lemma 12 in~\cite{borgs2015private}]\label{lem:rhoQtorho}
  Let $W:[0,1]^2\to [0,\Lambda]$ be a normalized graphon. 
  Let $\rho\in (0,\frac{1}{\Lambda})$, and $\bm Q_0$ be the edge connection probability matrix generated from graphon(i.e $\bm Q_0=\rho W(\bm x_i,\bm x_j)$ where $\bm x_i,\bm x_j$ are the labels of vertex $i$ and vertex $j$). 
  Let $\rho(\bm Q_0)=\frac{\norm{\bm Q_0}_1}{n^2}$.

  Then for any $\delta>0$, we have 
  \begin{equation*}
    \Pr\Brac{\Abs{\rho(\bm Q_0)-\rho}\geq \delta\rho}\leq O(\frac{\Lambda}{n\delta^2})\,.
  \end{equation*}
  As corollary, we have
  \begin{equation*}
    \E [\Paren{\rho(\bm Q_0)-\rho}^2]\leq O\Paren{\frac{\Lambda\rho^2}{n}}\,.
  \end{equation*} 
\end{lemma}

\subsubsection{Proof of error rate for private graphon estimation}
We first prove an error bound which holds with high probability given $\bm Q_0$
\begin{lemma}\label{lem:agnosticSBM}
    Under the setting of \cref{thm:formalmaintheorem}, conditioning on the edge connection probability matrix $\bm Q_0$, with probability at least $1-O\Paren{\frac{\Lambda}{n}+\frac{1}{(R\rho n)^{100}}}$, we have
\begin{equation*}
    \deltatwosquared(\hat{\bm B},B_0)\lesssim \frac{R k}{\rho n}+\frac{R^2 k^2\log(n)}{n\epsilon}+\Paren{\hat{\epsilon}_k^{(O)}(\bm Q_0)}^2+R^2\Abs{\rho-\rho(\bm Q_0)}^2
\end{equation*}
where $\hat{\epsilon}_k^{(O)}(\bm \Qnull)$ represents approximation error defined as
\begin{equation*}
    \hat{\epsilon}_k^{(O)}(\bm \Qnull)\coloneqq \frac{1}{\rho n}\min_{Z,B} \Normf{\rho\cdot ZBZ^\top-\bm \Qnull}\,.
\end{equation*}
with minimum taken over balanced community membership matrix $Z\in \{0,1\}^{n\times k}$ and $B\in [0,R]^{k\times k}$.
\end{lemma}
\begin{proof}
By \cref{lem:densityestimationgraphon}, with probability at least $1-O(\frac{R}{n})$ we have $\frac{1}{2}\rho\leq \hat{\bm \rho}\leq 2\rho$.

Let $\bm \Qnull\in \R^{n\times n}$ be the edge connection probability matrix, i.e vertices $i,j$ are connected with probability $\bm \Qnull(i,j)$.
Let $S\subseteq [n]$ be the set of vertices with degree larger than $20R d$.
We let $\bar{\bm A}$ be the adjacency matrix of the graph obtained by removing edges incident to vertices in $S$, and $\bar{\bm \Qnull}\in [0,1]^{n\times n}$ be the edge connection probability matrix restricted to $S$.

Since $\bm W=\bm A-\bm Q_0$ is a symmetric random matrix with
\begin{equation*}
    \bm W(i,j) =
      \begin{cases}
        1 - \bm Q_0({i,j}) & \text{w.p. } \bm Q_0({i,j})\\
        - \bm Q_0({i,j}) & \text{w.p. } 1- \bm Q_0({i,j})
      \end{cases}
\end{equation*}
Furthermore, we have $\bm Q_0(i,j)\leq R\cdot \rho$ by definition.
By \cref{theorem:pruned-spectral-norm}, with probability at least \(1-\frac{1}{n^2}\), we have \(\norm{\bar{\bm A}- \bar{\bm \Qnull}}\leq O\Paren{\sqrt{R\rho n}}\).
By \cref{lem:degree-counting}, with probability at least \(1-\exp(-R\rho n)\), we have \(\normo{\bar{\bm A}-\bm A}\leq R \rho n^2 \cdot\exp(-R \rho n)\).

We let $\bm \Yin=\frac{\bm A}{\hat{\bm \rho}}$, $\bm Y_2=\frac{\bar{\bm A}}{\hat{\bm \rho}}$, $\bm Y_1=\frac{\bar{\bm Q_0}}{\hat{\bm \rho}}$ and $\bm Y_0=\bm Z_0B_0\bm Z_0^\top$.
By definition, \(Y_2,\Yin\) satisfies the constraints \(\cA_2(Y_2;\Yin)\).
As result, we can apply \cref{thm:abstract-main-theorem}, and get
\begin{equation*}
    \deltatwosquared(\hat{\bm B},\bm B_0)\leq \frac{1}{n^2}\cdot O\Paren{k \cdot\norm{\bm Y_2 - \bm Y_1}^2 + \normf{\bm Y_1-\bm Y_0}^2 + R\normo{\bm Y_2-\bm \Yin}+\frac{R^2 nk^2\log(n)}{\epsilon}}\,.
\end{equation*}

Now we have $\norm{\bm Y_2-\bm Y_1}^2\leq \frac{R \rho n}{\hat{\bm \rho}^2}\leq O\Paren{\frac{R n}{\rho}}$, and
\begin{equation*}
    \normo{\bm Y_2-\bm \Yin}\leq \frac{1}{\hat{\bm \rho}}\normo{\bar{\bm A}-\bm A}\leq \frac{R \rho n^2}{\hat{\bm \rho}} \cdot\exp\Paren{-R \rho n}\leq 2R n^2 \exp(-R \rho n)\,.
\end{equation*}
Furthermore, by \cref{corollary:target-density-estimation-graphon}, 
\begin{equation*}
  \Paren{\hat{\bm \rho}-\rho}^2\leq \frac{\rho^2}{n\e}+\frac{\rho\polylog(n)}{n^2}+\Abs{\rho-\rho(\bm Q_0)}^2\,.
\end{equation*}

As result, with probability at least $1-\frac{1}{(Rd)^{100}}-O\Paren{\frac{R}{n}}$,
\begin{align*}
    \normf{\bm Y_1-\bm Y_0}^2 &\leq \Normf{\frac{\bar{\bm Q}_0}{\hat{\bm \rho}}-\frac{\bm Z_0 \bm B_0\bm Z_0}{\rho}}^2\\
    &\leq \frac{2}{\hat{\bm \rho}^2}\Normf{\bm Q_0-\bar{\bm Q}_0}^2+ 2\Normf{\bm Q_0}^2 \cdot \Paren{\frac{\hat{\bm \rho}-\rho}{\rho\hat{\bm \rho}}}^2+ \Normf{\frac{\bm Q_0-\bm Z_0 \bm B_0\bm Z_0}{\rho}}^2\\
    &\leq \exp(-R\rho n)\cdot \frac{2R^2 \rho^2 n^2}{\hat{\bm \rho}^2}+\frac{R^2 n^2}{\rho^2}\cdot\Paren{\frac{\rho^2}{n\e}+\frac{\rho\polylog(n)}{n^2}+\Abs{\rho-\rho(\bm Q_0)}^2}\\
    &+n^2 \Paren{\hat{\epsilon}_k^{(O)}(\bm Q_0)}^2\\
    &\leq 8\exp(-R\rho n)\cdot R^2 n^2+R^2n^2\Paren{\frac{1}{n\e}+\frac{\polylog(n)}{\rho n^2}+\Abs{\rho-\rho(\bm Q_0)}^2}\\
    &+n^2 \Paren{\hat{\epsilon}_k^{(O)}(\bm Q_0)}^2\,.
\end{align*}

As result, we have
\begin{align*}
    \deltatwosquared(\hat{\bm B},\bm B_0)&\leq \frac{1}{n^2}\cdot O\Paren{k \cdot\norm{\bm Y_2 -\bm Y_1}^2 + \normf{\bm Y_1-\bm Y_0}^2 + R\normo{\bm Y_2-\bm \Yin}+\frac{R^2 nk^2\log(n)}{\epsilon}}\\
    &\lesssim \frac{R k}{\rho n}+R^2 \exp(-R\rho n)+\frac{R^2 k^2\log(n)}{n\epsilon}+\Paren{\hat{\epsilon}_k^{(O)}(\bm Q_0)}^2+R^2 \Abs{\rho-\rho(\bm Q_0)}^2\\
    & \lesssim\frac{R k}{\rho n}+\frac{R^2 k^2\log(n)}{n\epsilon}+\Paren{\hat{\epsilon}_k^{(O)}(\bm Q_0)}^2+R^2\Abs{\rho-\rho(\bm Q_0)}^2\,.
\end{align*}
which concludes the proof.
\end{proof}

Now we prove \cref{thm:formalmaintheorem} by taking expectation over $\bm A$ and $\bm Q_0$.
\begin{proof}[Proof of \cref{thm:formalmaintheorem}]
    The running time and privacy guarantees of the algorithm directly follow as analyzed in the stochastic block models.

    Now we prove the utility guarantees for graphon estimation.

By \cref{lem:agnosticSBM}, for any edge connection probability matrix $\bm Q_0$, with probability at least $1-O(\frac{\Lambda}{n})-\frac{1}{(R\rho n)^{100}}$ over $\bm A$, we have 
\begin{align*}
    \deltatwosquared(\hat{\bm B},\bm B_0)\lesssim \frac{R k}{\rho n}+\frac{R^2 k^2\log(n)}{n\epsilon}+\Paren{\hat{\epsilon}_k^{(O)}(\bm Q_0)}^2+R^2\Abs{\rho-\rho(\bm Q_0)}^2
\end{align*}
Therefore we have
\begin{align*}
    \E\Brac{\deltatwosquared(\hat{\bm B},B_0)|\bm Q_0}&\lesssim \frac{R k}{\rho n}+\frac{R^2 k^2\log(n)}{n\epsilon}+\Paren{\hat{\epsilon}_k^{(O)}(\bm Q_0)}^2+\frac{R^2}{(R\rho n)^{100}}+R^2\Abs{\rho-\rho(\bm Q_0)}^2\\
    &\lesssim \frac{R k}{\rho n}+\frac{R^2 k^2\log(n)}{n\epsilon}+\Paren{\hat{\epsilon}_k^{(O)}(\bm Q_0)}^2+R^2\Abs{\rho-\rho(\bm Q_0)}^2\,.
\end{align*}

Now by \cref{lem:block_approximation_error}, we have
\begin{equation*}
    \Paren{\hat{\e}_k^{(O)}(\bm \Qnull)}^2\leq O\Paren{\Paren{\e_k^{(O)}(W)}^2+(\bm\e_n(W,\bm Q_0))^2}\,.
\end{equation*}
where $\epsilon_k^{(O)}(W)=\min_{B} \norm{W[B]-W}_2$ with minimum taken over $k\times k$ symmetric matrices $B$, and $\epsilon_n(W,\bm Q_0)\coloneqq \min_\pi\Norm{W\Brac{\frac{\pi \bm \Qnull \pi^\top}{\rho}}-W}_2$ with minimum taken over $n\times n$ permutation matrices.
Furthermore, \(\E\bm \e_n^2(W,\bm \Qnull)\leq R^2\sqrt{k/n}\).

By \cref{lem:rhoQtorho}, we have $\E\Abs{\rho-\rho(\bm Q_0)}^2\leq O\Paren{\frac{R \rho^2}{n}}$
As result, we have
\begin{align*}
    \E\deltatwosquared(\hat{\bm B},\bm B_0)&\lesssim \frac{R k}{\rho n}+\frac{R^2 k^2\log(n)}{n\epsilon}+\E\Paren{\hat{\epsilon}_k^{(O)}(\bm Q_0)}^2+R^2\E\Abs{\rho-\rho(\bm Q_0)}^2\\
    &\lesssim \frac{R k}{\rho n}+\frac{R^2 k^2\log(n)}{n\epsilon}+R^2\sqrt{\frac{k}{n}}+\Paren{\epsilon_k^{(O)}(W)}^2+\frac{R^3 \rho^2}{n}\\
    &\lesssim \frac{R k}{\rho n}+\frac{R^2 k^2\log(n)}{n\epsilon}+R^2\sqrt{\frac{k}{n}}+\Paren{\epsilon_k^{(O)}(W)}^2\,.
\end{align*}

Since we also have
\begin{equation*}
    \E\deltatwosquared(W[\bm B_0],W)\leq O\Paren{\Paren{\hat{\e}_k^{(O)}(\Qnull)}^2+\E \e_n^2(W,\bm Q_0)}\leq O\Paren{\Paren{\e_k^{(O)}(W)}^2+R^2\sqrt{k/n}}\,,
\end{equation*}
it can be concluded that
\begin{equation*}
  \E\deltatwosquared(W[\hat{\bm B}],W)\lesssim \frac{R k}{\rho n}+\frac{R^2 k^2\log(n)}{n\epsilon}+R^2\sqrt{\frac{k}{n}}+\Paren{\epsilon_k^{(O)}(W)}^2\,.
\end{equation*}
\end{proof}

%% file: content/sosidentifiability.tex
\section{Extended sum-of-squares spectral \Holder inequality}
\label{sec:sosidentifiability}

In this section, we prove the extended sum-of-squares spectral \Holder inequality in \cref{cor:sos-spectral-hoelder}, which is the key result that is needed in the proof of our sos identifiability lemma \cref{lem:abstractidentifiability}. The proof of \cref{cor:sos-spectral-hoelder} follows directly from the two general sos inequalities in \cref{lem:general_lr_sosidentifiability} and \cref{lem:certify_cut_norm_bound}.

\begin{lemma}\label{lem:general_lr_sosidentifiability}
  Consider the following polynomial constraint system with respect to $V\in \R^{n\times k}$, $\Vnull \in \R^{n\times k}$ and $M\in \R^{n\times n}$
  \begin{equation*}
    \cA=\Set{V^\top V=\Id_k; \Vnulltrans \Vnull =\Id_k; (\Id_n-VV^\top)M(\Id_n-\Vnull \Vnulltrans )=0}\,.
  \end{equation*}
  For any symmetric matrix $W\in \R^{n\times n}$, there is a sum-of-squares proof that 
  \begin{equation*}
    \cA \sststile{8}{M, V, V_0}  \iprod{M,W}\leq \tfrac 3 2 \Normf{M}^2+\Normf{VV^\top W}^2+\Normf{\Vnull \Vnulltrans W}^2\,.
  \end{equation*} 
\end{lemma}

\begin{proof}
  First, expanding $(\Id_n - VV^\top)M(\Id_n - \Vnull \Vnulltrans )=0$ and rearranging terms, we have
  \begin{equation*}
    \cA \sststile{8}{M,V,V_0} M=VV^\top M+M\Vnull  \Vnulltrans -VV^\top M \Vnull  \Vnulltrans \,.
  \end{equation*}
  Plugging this into $\iprod{M,W}$, we get
  \begin{equation*}
    \cA \sststile{8}{M,V,V_0} \iprod{M,W}= \iprod{VV^\top M,W}+\iprod{M\Vnull  \Vnulltrans ,W}-\iprod{VV^\top M \Vnull  \Vnulltrans , W}\,.
  \end{equation*}
  Now, we bound the three terms on the right hand side separately. First, by AM-GM inequality we have
  \begin{equation*}
    \sststile{4}{M,V,V_0} \iprod{VV^\top M,W}=\iprod{M,VV^\top W}\leq \frac{1}{2}\Normf{M}^2+ \frac{1}{2}\Normf{VV^\top W}^2\,.
  \end{equation*}
  For the same reason, we have 
  \begin{equation*}
    \sststile{4}{M,V,V_0} \iprod{M \Vnull \Vnulltrans,W} = \iprod{M, \Vnull \Vnulltrans W} \leq \frac{1}{2} \Normf{M}^2+ \frac{1}{2}\Normf{\Vnull \Vnulltrans  W}^2\,.
  \end{equation*}
  For the third term, we have
  \begin{equation*}
    \sststile{8}{M,V,V_0} - \iprod{VV^\top M \Vnull \Vnulltrans ,W} = - \iprod{ M ,VV^\top W\Vnull \Vnulltrans } \leq \frac{1}{2} \Normf{M}^2+ \frac{1}{2}\Normf{VV^\top W \Vnull \Vnulltrans}^2\,.
  \end{equation*}
  Notice that, by the cyclicity of the matrix trace and AM-GM inequality, we have the following degree-$8$ sum-of-squares inequality for the term $\Normf{VV^\top W \Vnull \Vnulltrans}^2$
  \begin{align*}
    \cA \sststile{8}{M,V,V_0}
    \Normf{VV^\top W \Vnull  \Vnulltrans }^2
    = & \Tr\Paren{ W VV^\top W \Vnull  \Vnulltrans} \\
    = & \iprod{VV^\top W, W \Vnull  \Vnulltrans} \\
    \leq & \frac{1}{2}\Normf{VV^\top W}^2 + \frac{1}{2}\Normf{\Vnull  \Vnulltrans W}^2 \,.
  \end{align*}
  Combining everyting together, we have 
  \begin{align*}
    \cA \sststile{8}{M,V,V_0}  \iprod{M,W} &=  \iprod{VV^\top M,W}+\iprod{M\Vnull  \Vnulltrans ,W}-\iprod{VV^\top M \Vnull  \Vnulltrans , W}\\
    &\leq \frac{3}{2} \Normf{M}^2+\Normf{VV^\top W}^2+\Normf{\Vnull \Vnulltrans W}^2\,.
  \end{align*}
\end{proof}

\begin{lemma}\label{lem:certify_cut_norm_bound}
  Consider the following polynomial constraint system with respect to $V\in \R^{n\times k}$
  \begin{equation*}
    \cA\coloneqq \Set{V^\top V=\Id_k}\,.
  \end{equation*}
  For any symmetric matrix $W \in \R^{n\times n}$, there is a sum-of-squares proof that
  \begin{equation*}
    \cA \sststile{4}{V} \Normf{VV^\top W}^2 \leq k\cdot \Norm{W}^2\,.
  \end{equation*}
\end{lemma}

\begin{proof}
  Since $V^\top V=\Id_k$, we have
  \begin{align*}
    \cA \sststile{4}{V} \Normf{VV^\top W}^2
    & = \Tr\Paren{WVV^\top VV^\top W}\\
    & = \Tr\Paren{WVV^\top W}\\
    & = \Tr\Paren{VV^\top W^2} \,.
  \end{align*}
  Notice that $\Norm{W^2} \leq \Norm{W}^2$. Therefore, we can write $W^2 = \Norm{W}^2 \Id_n - B^{\top} B$ for some matrix $B \in \R^{n\times n}$. It follows that
  \begin{align*}
    \cA \sststile{4}{V} \Normf{VV^\top W}^2
    & = \Tr\Paren{VV^\top \Paren{\Norm{W}^2 \Id_n - B^{\top} B}} \\
    & = \Tr\Paren{VV^\top} \cdot \Norm{W}^2 - \Tr\Paren{VV^\top B^{\top} B} \\
    & = \Tr\Paren{V^\top V} \cdot \Norm{W}^2 - \Tr\Paren{(B V)^{\top} B V} \\
    & = \Tr\Paren{\Id_k} \cdot \Norm{W}^2 - \Normf{BV}^2 \\
    & \leq k\cdot \Norm{W}^2\,.
  \end{align*}
\end{proof}

By combining the previous two lemmas, we obtain the following direct corollary.

\begin{corollary}
  \label{cor:sos-spectral-hoelder}
  Let \(V,\Vnull\) be \(n\)-by-\(k\) matrices of indeterminates and let \(M\) be an \(n\)-by-\(n\) matrix of indeterminates.
  Consider the following polynomial system in variables \(V,\Vnull,M\),
  \begin{equation*}
    \cA=\Set{V^\top V=\Id_k; \Vnulltrans \Vnull =\Id_k; (\Id_n-VV^\top)M(\Id_n-\Vnull \Vnulltrans )=0}\,.
  \end{equation*}
  Then, for every symmetric matrix $W\in \R^{n\times n}$,
  \begin{equation*}
    \cA \sststile{8}{M, V, V_0}
    \iprod{M,W}\leq \tfrac 3 2 \Normf{M}^2 + 2k \norm{W}^2\,.
  \end{equation*}
\end{corollary}

%% file: content/doublystochastic.tex
\section{Graphon distance as quadratic optimization over the Birkhoff polytope}\label{sec:graphonBirkhoff}

\paragraph{Doubly stochastic distance metric}
Recall that we define $B_k$ be the set of all $k$-by-$k$ doubly stochastic matrices (also known as the Birkhoff polytope). For $k$-by-$k$ matrices with nonnegative entries $B,\Bnull$, we define the error metric
 \begin{equation*}
	\dsdistance(B,\Bnull)
	=\min_{S\in B_k} \frac{1}{k^2}\sum_{a,a',b,b'\in [k]}\Paren{B(a,b)-\Bnull(a',b')}^2\cdot S(a,a')\cdot S(b,b') \,.
\end{equation*}

It is easy to relate this distance metric to $\delta_2(\cdot,\cdot)$. Similar inequalities appears in~\cite{xu2020gwdistance,Peyre15Gromov}.

\begin{lemma}[Relation between distance matrices]\label{lem:dsdistance}
	Let  $B,\Bnull$ be $k$-by-$k$ matrices with non-negative entries, then $\delta_2^2(B,\Bnull ) = \dsdistance(B, \Bnull)$.
\end{lemma}

\begin{proof}
	The direction $\delta_2^2(B,\Bnull ) \geq \dsdistance(B, \Bnull)$ is trivial because the two measure preserving functions $\phi_1, \phi_2: [0,1] \rightarrow [k]$ that optimizes $\delta_2^2(B,\Bnull)$ defines a doubly stochastic matrix
	\begin{equation*}
		S(a, b) = \phi_1^{-1}(a) \cdot \phi_2^{-1}(b) \,,
	\end{equation*}
	such that
	\begin{equation*}
		\frac{1}{k^2}\sum_{a,a',b,b'\in [k]}\Paren{B(a,b)-\Bnull(a',b')}^2\cdot S(a,a')\cdot S(b,b') = \delta_2^2(B,\Bnull) \,.
	\end{equation*}

	Now, we prove $\delta_2(B,\Bnull )^2 \leq \dsdistance(B, \Bnull)$.
	We consider graphons the $W=W[B]$, $W^*=W[\Bnull]$ and given a doubly stochastic matrix $S\in [0,1]^{k\times k}$,  we construct a mapping $\phi:[0,1]\to [0,1]$ such that
	\begin{equation*}
		\Norm{W^\phi-W^*}_2^2\leq \frac{1}{k^2}\sum_{a,a',b,b'\in [k]}\Paren{B(a,b)-\Bnull(a',b')}^2\cdot S(a,a')\cdot S(b,b')\,.
	\end{equation*}
	The proof then follows since by the definition, we have $\delta_2(B,\Bnull )^2\leq \Norm{W^\phi-W^*}_2^2$.
	
	Now we describe our construction of bijective mapping $\phi$. 
	We consider the standard partition of $[0,1]$ into $k$ equal length intervals. The intuition is that, for each interval $t$, we map certain mass of points to each intervals $h\in [k]$, according to the doubly stochastic matrix $S$. Concretely consider the $a$-th interval $I_a$, we further partition $I_a$ into $k$ sub-intervals $I_{a,1}\,,\ldots\,,I_{a,k}$ each with weight respectively proportional to $S(a,1),S(a,2),\ldots,S(a,k)$. Then for $x$ inside the $b$-th sub-interval, we require $\phi(x)\in \Brac{\frac{b-1}{k},\frac{b}{k}}$. Such bijective mapping $\phi$ exists, since $S$ is a doubly stochastic matrix with each row and each column summing up to $1$.
	
	Under the standard partition, $W(x,y)=B_1(a,b)$ when $\frac{a-1}{k}\leq x<\frac{a}{k}$ and $\frac{b-1}{k}\leq y<\frac{b}{k}$. 
	Then we have 
	\begin{align*}
		\Norm{W^\phi-W^*}_2^2 &=\sum_{a,a',b,b'\in [k]} \int_{x\in I_{a,a'}}\int_{y\in I_{b,b'}} 
		\Paren{B_1(a,b)-B_2(a',b')}^2 dx dy\\
		&= \frac{1}{k^2}\sum_{a,a',b,b'\in [k]}\Paren{B_1(a,b)-B_2(a',b')}^2 \cdot S(a,a')\cdot  S(b,b')
	\end{align*}
	which finishes the proof.
\end{proof}

%% file: content/rounding.tex
\section{Polynomial optimization over convex polytopes via sum-of-squares}\label{section:sos-optimization-polytope}

We reuse the notation introduced in \cref{sec:graphonBirkhoff} and expand it. For a degree-$d$ $n$-multivariate polynomial $p(x)=p(x_1,\ldots,x_n)$, we index its coefficients with unordered multi-indices $\alpha$ in $[n]^d$ so that $p(x)=\sum_{\alpha\in [n]^d}p_\alpha \underset{i\in \alpha}{\prod}x_i$. We let $\Normo{p}$ be the sum of the absolute values of the coefficients of $p$.  For a convex polytope $P\,,$ we write $V(P)$ for the set of its extreme points. We denote by $\cB_n\subseteq \R^n$ the unit ball in $\R^n$.

Our main tool is the following statement about optimization of quadratic polynomials with non-negative coefficients over convex polytopes
\begin{theorem}[Reweighed pseudo-distribution]\label{theorem:pseudo-distribution-polytopes}
	Let $C>0$ be a constant, let $\zeta>n^{-C}$ and let $\ell\geq q+(n/\zeta)\cdot (\log n)^{C'}$, for some $C'$ depending only on $C$. Let $\cA(x)$ be a set of polynomial constraints of degree at most $q$ in $n$-dimensional vector of indeterminates $x$ and let $p(x)$ be a quadratic polynomial.
	Suppose $\cA(x)$ is sos-consistent up to degree-$\ell$. Then, there exists a pseudo-distribution of level $\ell-(n/\zeta)\cdot \log(n)^{C'}$ satisfying $\cA(x)$ such that
	\begin{align}
		\max_{i,j \in [n]}\Abs{\tilde{\E}\Brac{ x_i x_j}- \tilde{\E}\Brac{x_i}\tilde{\E}\Brac{x_j}}\leq \zeta\,.
	\end{align}
\end{theorem}

We defer the proof of the Theorem to the end of the section.
Crucially, \cref{theorem:pseudo-distribution-polytopes} implies the following result, of which \cref{lem:abstractrounding} is an immediate consequence.

\begin{corollary}[Sos multiplicative approximation over polytopes]\torestate{\label{corollary:sos-multiplicative-approx-polytope}
		Let $C>0$ be a constant, let $1/2> \zeta>n^{-C}\,, \tau>0$ and let $\ell\geq 2+(n^3/(\zeta^2\cdot \tau^2))\cdot (\log n)^{C'}$, for some $C'$ depending only on $C$.
		Let $p(x)$ be a quadratic polynomial with non-negative coefficients and let $P$ be a convex polytope satisfying
		\begin{enumerate}[(i)]
			\item $P\subseteq \cB_n\,,$
			\item  $\forall i \in [n]\,, \underset{z\in V(P)\,, \textnormal{ s.t. }z_i\neq 0}{\min} z_i\geq \tau>0\,.$
		\end{enumerate}
		Suppose  there exists a level-$\ell$ pseudo-distribution over $P$ satisfying the constraint
		\begin{align}\label{eq:pseudo-distribution-quadratic-polynomial-constraint}
			\Set{p(x)\leq t}\,.
		\end{align}
		Then there exists $x^*\in P$ satisfying
		\begin{align*}
			p(x^*)\leq \Paren{1+5n\sqrt{\zeta}}t\,.
		\end{align*}
	}
\end{corollary}
\begin{proof}
	We assume without loss of generality that the constant coefficient in $p$ is zero. 
	By \cref{theorem:pseudo-distribution-polytopes} and choice of $\ell$, we may assume there exists a level-$\geq 2$ pseudo-distribution $\mu$ over $P$ satisfying 		
	\cref{eq:pseudo-distribution-quadratic-polynomial-constraint} and such that, for all $i,j$
	\begin{align}\label{eq:bound-covariance}
		\Abs{\tilde{\E}_\mu \Brac{x_ix_j}-\tilde{\E}_\mu\Brac{x_i}\tilde{\E}_\mu \Brac{x_j}}\leq \tau^2\zeta^2\,.
	\end{align}

	By Carathéodory's theorem, we have $\tilde{\E}_\mu \Brac{x}=\sum_{z\in S}\mu(z) z $ for some $S\subseteq V(P)$ with size at most $n+1$. 
	Let $\tilde{x}$ be obtained from zeroing out coefficients $\mu(z)$ which are smaller than $2\tau\sqrt{\zeta}$. 
	Then we have $\tilde{x}_i\leq \tilde{\E}_\mu \Brac{x_i}$, and $\tilde{x}_i=0$ if $\tilde{\E}_\mu \Brac{x_i}\leq 2\tau\sqrt{\zeta}$.
	Moreover for non-zero entries of $\tilde{x}$, we have $\tilde{x}_i\geq 2\tau\sqrt{\zeta}$.
	By choice of $\zeta$, for some $r\in \Brac{1-2n\tau \sqrt \zeta,1 }$ we have $x^*:=\frac{\tilde{x}}{r}\in P\,.$ 
	
	Define $\cQ_1$ to be the set of indices in $[n]$ such that $\tilde{x}_i\geq 2\tau\sqrt{\zeta}$.
	Similarly, let $\cQ_2$ be the set of pairs $(i,j)$ such that $\tilde{\E}_\mu \Brac{x_ix_j}\leq \tau^2\zeta$. 
	Observe that by \cref{eq:bound-covariance} we have $\cQ_1\times \cQ_1 \subseteq \cQ_2$ as $\tau^2\zeta+\tau^2\zeta^2<4\tau^2\zeta\,.$
	
	Hence, putting things together
	\begin{align*}
		\sum_{i,j\in [n]} p_\ij x^*_ix^*_j&= 	\sum_{\substack{i\in \cQ_1\\j \in \cQ_1}} p_\ij x^*_ix^*_j\\
		&\leq \frac{1}{r^2}	\sum_{\substack{i\in \cQ_1\\j \in \cQ_1}} p_\ij  \tilde{\E}_\mu\Brac{x_i}\tilde{\E}_\mu\Brac{x_j}\\
		&\leq \frac{1}{r^2}	\sum_{\substack{i\in \cQ_1\\j \in \cQ_1}} p_\ij  \Paren{\tilde{\E}_\mu \Brac{x_ix_j}+\tau^2\zeta^2}\\
		&\leq  \frac{1}{r^2}	\sum_{\substack{(\ij)\in \cQ_2}} p_\ij  \Paren{\tilde{\E}_\mu \Brac{x_ix_j}+\tau^2\zeta^2}\\
		&\leq  \frac{1+\zeta}{r^2}\sum_{\substack{(\ij)\in \cQ_2}}p_\ij \tilde{\E}_\mu \Brac{x_ix_j}\\
		&\leq \frac{1+\zeta}{r^2}\sum_{i,j \in[n]}p_\ij\tilde{\E}_\mu  \Brac{x_ix_j}\\
		&\leq (1+5n\sqrt \zeta)\sum_{i,j \in[n]}p_\ij\tilde{\E}_\mu  \Brac{x_ix_j}\,.
	\end{align*}
\end{proof}

\paragraph{Additive sos approximation over convex polytopes}
We remark that \cref{theorem:pseudo-distribution-polytopes} also shows how sum-of-squares can be used to certify an additive approximation for polynomial optimization over polytopes. We include the result for completeness although we emphasize that we do not use it anywhere.

\begin{corollary}[Sos additive approximation over polytopes]\label{corollary:sos-additive-approx-polytope}
	Let $C>0$ be a constant, let $1/2> \zeta>n^{-C}\,, \tau>0$ and let $\ell\geq 2+(n/\zeta)\cdot (\log n)^{C'}$, for some $C'$ depending only on $C$.
	Let $p(x)$ be a quadratic polynomial and let $P$ be a convex polytope in $\cB_n\,.$ 
	Suppose  there exists a level-$\ell$ pseudo-distribution over $P$ satisfying the constraint
	\begin{align}\label{eq:pseudo-distribution-quadratic-polynomial-constraint-2}
		\Set{p(x)\leq t}\,.
	\end{align}
	Then there exists $x^*\in P$ satisfying
	\begin{align*}
		p(x^*)\leq t+\zeta\Normo{p}\,.
	\end{align*}
\end{corollary}
\begin{proof}
  Using \cref{theorem:pseudo-distribution-polytopes}, let $\mu$ be a pseudo-distribution over $P$ satisfying
  \begin{align*}
    \Set{p(x)\leq t}\,,
  \end{align*}
  and such that
  \begin{align*}
		\sum_{\ij\in[n]} 	\Abs{\tilde{\E}_\mu\Brac{p_\ij x_i x_j}-p_\ij \tilde{\E}_\mu\Brac{x_i}\tilde{\E}_\mu\Brac{x_j}}\leq \zeta \cdot \Normo{p}\,.
  \end{align*}
 Notice such pseudo-distribution must exist by assumption on $\ell$.
  Now,  picking $x^*=\tilde{\E}\brac{x}$, we get
  \begin{align*}
    \tilde{\E}\Brac{p(x)}\leq p(\tilde{\E}\brac{x})  \leq \tilde{\E}\Brac{p(x)} +\zeta \cdot \Normo{p}
  \end{align*}
  concluding the argument as desired.
\end{proof}

\paragraph{Proof of pseudo-distribution reweighing}
We now prove the main theorem of the section, our argument leverages a statement appearing  in ~\cite{barak2017quantum}.

\begin{proof}[Proof of \cref{theorem:pseudo-distribution-polytopes}]
	Let $\mu$ be a pseudo-distribution satisfying $\cA(x)$.
	We claim that the result follows if (dropping the subscript $\mu$)
	\begin{align}\label{eq:polytope-reweighted}
		\tilde{\E}\Brac{\Snormt{x-\tilde{\E} \brac{x}}}\leq \zeta/n\,.
	\end{align} 
	Indeed, observe that for any $i,j\in[n]$
	\begin{align*}
		\Abs{\tilde{\E}\Brac{ x_i x_j}- \tilde{\E}\Brac{x_i}\tilde{\E}\Brac{x_j}}&=\Abs{\tilde{\E}\Brac{x_i x_j}-\tilde{\E}\Brac{x_i}\tilde{\E}\Brac{x_j}}\\
		&= \Abs{\tilde{\E}\Brac{\Paren{x_i-\tilde{\E}\brac{x_i}}\Paren{x_j-\tilde{\E}\brac{x_j}}}}\\
		&\leq \tilde{\E}\Brac{\Paren{x_i-\tilde{\E}\Brac{x_i}}^2}^{1/2}\tilde{\E}\Brac{\Paren{x_j-\tilde{\E}\Brac{x_j}}^2}^{1/2}\\
		&\leq \Paren{\sum_{a\,,b\in [n]}\tilde{\E}\Brac{\Paren{x_a-\tilde{\E}\Brac{x_a}}^2}^{1/2}\tilde{\E}\Brac{\Paren{x_b-\tilde{\E}\Brac{x_b}}^2}^{1/2}}\\
		&=\Paren{\sum_{a}\tilde{\E}\Brac{\Paren{x_a-\tilde{\E}\Brac{x_a}}^2}^{1/2}}^2\\
		&\leq  n \cdot  \sum_{a}\tilde{\E}\Brac{\Paren{x_a-\tilde{\E}\Brac{x_a}}^2}\\
		&\leq  n \cdot  	\tilde{\E}\Brac{\Snormt{x-\tilde{\E} \brac{x}}}\\
		&\leq \zeta\,.
	\end{align*}
	where in the third step we used Cauchy-Schwarz for pseudo-distributions \cref{fact:pd-cauchy-schwarz}.
	
	It remains to show a pseudo-distribution satisfying $\cA(x)$ and \cref{eq:polytope-reweighted} exists.
	Let $\mu'$ be a level-$\ell$ pseudo-distribution satisfying $\cA(x)$ but not \cref{eq:polytope-reweighted}.
	Note then  it must be $\E_{\mu'}\Brac{\Snormt{x}}\geq n^{-C}$ as otherwise the inequality is satisfied by choice of $\zeta$.
	By Lemma 7.1 in~\cite{barak2017quantum} (see restatement in \cref{lem:bks-sos-reweighing}) and its direct corollary \cref{cor:bks-sos-reweighing}, we can always reweight $\mu'$ to obtain a level-$\Paren{\ell-(n/\zeta)\cdot (\log n)^{C'}}$ pseudo-distribution over $P$ satisfying $\cA(x)$ and such that \cref{eq:polytope-reweighted} holds.
\end{proof}

%% file: content/soslipschitz.tex
\section{Lipschitz extensions within sum-of-squares}
\label{sec:lip_sos}

In this section, we demonstrate how to incorporate Lipschitz extensions into our sum-of-squares framework for private graphon estimation.
To this end, we first set up the framework for Lipschitz extensions in the context of differential privacy.

\paragraph{Lipschitz extensions as a privacy tool}

Lipschitz extensions are basic mathematical objects that can be defined on general metric space (see the lecture note by Naor \cite{lipschitz2015note} for more background).
Here we focus on the space of graphs (equivalently, adjacency matrices) equipped with node distance\footnote{Recall the node distance between two $n$-vertex graphs $G,G'$ is the minimum number of vertices of $G$ that need to be rewired to obtain $G'$.} $\dnode$.
Let $\bbG_{n,d}$ denote the set of adjacency matrices of $n$-vertex graphs with maximum degree at most $d$, and let $\bbG_{n}$ be the abbreviation of $\bbG_{n,n}$.
A graph function $f:\bbG_n\to\R$ has a Lipschitz constant $c$ if $\abs{f(G)-f(G')} \le c\cdot\dnode(G,G')$ for all $G,G'\in\bbG_n$.
In differential privacy, the smallest Lipschitz constant of $f$ is also known as the sensitivity of $f$.

Functions with lower sensitivity can be approximated more accurately by private algorithms.
For example, the Laplace mechanism releases a private approximation of $f(G)$ by adding Laplace noise proportinal to the sensitivity of $f$ to $f(G)$.
However, there are many cases where the function of interest can only have low sensitivity when its input space is restricted to a subset $\bbG'_n$ of $\bbG_n$.
For example, the sensitivity of the number of edges in a graph is $d$ when restricted to $\bbG_{n,d}$.
Given a function $f:\bbG_{n}'\to\R$ with low sensitivity $\Delta'$, if we have a Lipschitz extension $\hat{f}$ of $f$ to $\bbG_n$, then we can add noise proportinal to this low sensitivity $\Delta'$ to $\hat{f}(G)$ and guarantee privacy on all graphs in $\bbG_n$. Moreover, we can still approximate $f(G)$ accurately for $G\in\bbG_n'$, as $\hat{f}(G)=f(G)$ for $G\in\bbG_n'$.

\paragraph{Towards efficient Lipschitz extensions}

Lipschitz extensions always exist for real-valued functions~\cite{mcshane1934extension}.
However, in general it is not known if these extensions are efficiently computable even when $f$ is efficiently computable.
There is a line of work on constructing efficient Lipschitz extensions for designing node-DP algorithms~\cite{blocki2013differentially,nodeDP,raskhodnikova2015efficient,borgs2015private,MR3631012,borgs2018revealing,kalemaj2023node}.
Let us first abstract a common thread in this line of work.
Many interesting statistics of graphs can be formulated as the optimal value of an optimizaiton problem, i.e.,
\begin{equation}
  \label{eq:lip_sos-f}
  f(A) = \max_{x\in\cX} s(x;A) \,,
\end{equation}
where $A$ is the adjacency matrix of input graph, $x=(x_i)_{i\in[m]}$ are auxiliary variables, $\cX\subseteq\R^m$ is the feasible set.
For objective functions $s(x;A)$ that are linear and nondecreasing in $A$ after fixing $x\in\cX$, previous work~\cite{nodeDP,borgs2015private,borgs2018revealing} constructed a Lipschitz extension $\hat{f}$ of $f$ from $\bbG_{n,d}$ to $\bbG_{n}$ as follows,
\begin{equation}
  \label{eq:lip_sos-f_extension}
  \hat{f}(A) = \max_{C} f(C) \text{ such that } 
  \left\{
    \begin{array}{l}
      C \in [0,1]^{n\times n} \text{ is symmetric} \,,   \\
      0 \leq C_{ij} \leq A_{ij} \text{ for all } i,j\in[n] \,, \\
      \sum_{j \neq i} C_{ij} \leq d \text{ for all } i\in[n] \,. 
    \end{array}
  \right.
\end{equation}
It is not difficult for verify that
\begin{itemize}
\item $\max_{A\sim A' \in \bbG_n} \abs{\hat{f}(A) - \hat{f}(A')} = \max_{A\sim A' \in \bbG_{n,d}} \abs{f(A) - f(A')}$ ;
\item $\hat{f}(A) = f(A)$ for all $A\in\bbG_{n,d}$ .
\end{itemize}
To see $\hat{f}:\bbG_n\to\R$ has the same sensitivity as $f:\bbG_{n,d}\to\R$, we observe that for every fixed $x\in\cX$, the optimizaiton problem in \cref{eq:lip_sos-f_extension} is a linear programing.
Since the constraint set is a 0/1 polytope, the maximum is achieved by some 0/1 solution $C\in\bbG_{n,d}$.
Suppose $A,A'$ differ in the first row and column, we can find a feasible $C'\in\bbG_{n,d}$ for $A'$ by zeroing out the first row and column of $C$.

\subsection{Simulating Lipschitz extensions within sum-of-squares}

Now we show how to simulate the above Lipschitz extension \cref{eq:lip_sos-f_extension} within sos.
To fit into the sum-of-squares framework, we need to work with polynomials.
We assume $s(x,C)$ is a polynomial in variables $x$ and $C$.
When we substitue for $C$ a concrete assignment $A$, we use $s(x;A)$ to denote the resulting polynomial in variables $x$ with coefficients depending on $A$.
We also assume the feasible set $\cX\subseteq\R^m$ can be encoded by a polynomial system, i.e.,
\[
  \cX = \Set{x\in\R^m : x \text{ satisfies } \cP_1(x)}
\]
where 
\begin{equation}
  \cP_1(x) \seteq \set{p_1(x)\ge0, \dots, p_k(x)\ge0}  
\end{equation}
for some polynomials $p_1, \dots, p_k$.

We define the sum-of-squares relaxtion of function $f$ in \cref{eq:lip_sos-f} as
\begin{equation}
  \label{eq:lip_sos-f_sos}
  \fsos(A) \seteq \max t \text{ s.t. $\exists$ level-$\ell$ pseudo-distribution } 
  \mu \sdtstile{}{x} \Set{s(x;A)\ge t} \cup \cP_1(x) \,.
\end{equation}
For a graph $A\in\bbG_n$ and a degree bound $d\in[n]$, define a polynomial system
\begin{equation}
  \cP_2(C;A,d) \seteq \Set{0 \le C \le A, C\Ind \le d\cdot\Ind, C=\transpose{C}} \,,
\end{equation}
which encodes weighted subgraphs of $A$ with weighted maximum degree at most $d$.
Then we can simulate the Lipschitz extension in \cref{eq:lip_sos-f_extension} within sos as follows,\footnote{We say a level-$\ell'$ pseudo-distribution satisfies a polynomial system in variables $x,C$ with the additional condition that $\deg(x)\le\ell$, if we only allow polynomials with degree at most $\ell$ in $x$ and total degree at most $\ell'$ when defining the pseudo-distribution.}
\begin{equation}
  \label{eq:lip_sos-f_sos_extension}
	\begin{split}
  \hatfsos(A) \seteq \max t \text{ s.t. $\exists$ }  
    & \text{level-$\ell'$ pseudo-distribution } \\
    & \mu' \sdtstile{\deg(x)\leq\ell}{x,C} \Set{s(x,C)\ge t} \cup \cP_1(x) \cup \cP_2(C;A,d) \,.
	\end{split}
\end{equation}

We remark that $\hatfsos$ is a Lipschitz extension of $\fsos$ from $\bbG_{n,d}$ to $\bbG_{n}$ in the following sense:
\textit{(i)} $\hatfsos$ extends $\fsos$, i.e., $\hatfsos(A) = \fsos(A)$ for all $A\in\bbG_{n,d}$;
\textit{(ii)} the sensitivity proof of $\fsos$ can be carried out in the same way for $\hatfsos$.\footnote{We might not have the most stringent case that $\max_{A\sim A' \in \bbG_n} \abs{\hatfsos(A) - \hatfsos(A')} = \max_{A\sim A' \in \bbG_{n,d}} \abs{\fsos(A) - \fsos(A')}$, but property \textit{(ii)} suffices to design differentially private algorithms.}
In the remaining part of this section, we bound sensitivity of $\hatfsos$ and prove its extension property, in \cref{lem:lip_sos-hatfsos_sensitivity} and \cref{lem:lip_sos-hatfsos_extension} respectively.

\paragraph{Sensitivity}

To bound sensitivity of functions defined in terms of sos-consistency like $\fsos$ and $\hatfsos$, we introduce the following general lemma that relates sos-consistency of two polynomial systems.

\begin{lemma}
  \label{lem:lip_sos-sos_consistency}
  Consider two polynomial systems $\cP(x)$ and $\cQ(x)$ in variables $x=(x_i)$.
  Suppose there exists a linear function $L(x)$ such that for every polynomial inequality $q(x)\ge0$ in $\cQ(x)$, we have $\cP(x) \sststile{\deg(q)}{} q(L(x))$. 
  Then, if $\cP(x)$ is sos-consistent up to level $\ell$, $\cQ(x)$ is also sos-consistent up to level $\ell$.
\end{lemma}
\begin{proof}
  Suppose $\cP(x)$ is sos-consistent up to level $\ell$, which is witnessed by pseudo-distribution $\mu$.
  In the following, we construct a pseudo-distribution $\mu'$ that witnesses the level-$\ell$ sos-consistency of $\cQ(x)$. 
  For each polynomial $q(x)$ of degree at most $\ell$, let $\tE_{\mu'} q(x) = \tE_{\mu} q(L(x))$.
  Note $\tE_{\mu} q(L(x))$ is well-defined as $L$ is linear and thus $q(L(x))$ is a polynomial of degree at most $\ell$.
  Then it is a routine work to verify $\mu' \sdtstile{}{} \cQ(x)$. %
\end{proof}

We remark that the polynomial systems in \cref{eq:combined_poly_sys} satisfy the abstract property assumed in \cref{lem:lip_sos-sos_consistency}, as shown in \cref{lem:abstractsensitivity}.

As a corollary of \cref{lem:lip_sos-sos_consistency}, we can bound sensitivity of functions defined in terms of sos-consistency.

\begin{corollary}
  \label{coro:lip_sos-meta_sensitivity}
  Let $\cY$ be a space of datasets, $\ell\in\N$, and $\Delta\ge0$.
  Consider a function $h:\cY\to\R$ defined as
  \[
    h(Y) \seteq \max t \text{ s.t. $\cP(x;Y,t)$ is sos-consistent up to level $\ell$,} 
  \]
  where $\cP(x;Y,t)$ is a polynomial system in variables $x=(x_i)$.
  Suppose for every $t$ and every pair of neigboring datasets $Y,Y'$,
  there exists a linear function $L(x)$ such that any polynomial inequality $p(x)\ge0$ in $\cP(x;Y',t-\Delta)$ has an sos proof $\cP(x;Y,t) \sststile{\deg(p)}{} p(L(x))$.
  Then, 
  \[
    \max_{Y\sim Y'\in\cY} \Abs{h(Y)-h(Y')} \le \Delta \,.
  \]
\end{corollary}
\begin{proof}
  Fix an arbitrary pair of neigboring datasets $Y,Y'$.
  Due to symmetry of $Y$ and $Y'$, it suffices to prove $h(Y') \ge h(Y)-\Delta$.  
  Suppose $h(Y)=t$. Then $\cP(x;Y,t)$ is sos-consistent up to level $\ell$.
  By \cref{lem:lip_sos-sos_consistency}, $\cP(x;Y',t-\Delta)$ is also sos-consistent up to level $\ell$.
  Thus $h(Y') \ge t-\Delta = h(Y)-\Delta$.
\end{proof}

Then it is straightforward to bound the sensitivity of $\hatfsos$ using \cref{coro:lip_sos-meta_sensitivity}.

\begin{lemma}[Sensitivity of $\hatfsos$]
  \label{lem:lip_sos-hatfsos_sensitivity}
  Consider function $\hatfsos$ defined in \cref{eq:lip_sos-f_sos_extension}.
  Let $\Delta\ge0$. Suppose for every $t$ and every pair of node-adjacent graphs $A,A'\in\bbG_n$, there exist linear functions $L_1(x)$ and $L_2(C)$ such that for every polynomial inequality $p(x,C)\ge0$ in $\Set{s(x,C)\ge t-\Delta} \cup \cP_1(x) \cup \cP_2(C;A',d)$, we have
  \begin{equation*}
    \Set{s(x,C)\ge t} \cup \cP_1(x) \cup \cP_2(C;A,d) 
    \sststile{\deg(p)}{}
    p(L_1(x),L_2(C)) \ge 0 \,.
  \end{equation*}
  Then 
  \[
    \max_{A\sim A' \in \bbG_n} \abs{\hatfsos(A) - \hatfsos(A')} \le \Delta \,.
  \]
\end{lemma}
\begin{proof}
  By \cref{coro:lip_sos-meta_sensitivity}.
\end{proof}

\paragraph{Extension}

Now we prove the extension property of $\hatfsos$.
Recall $\hat{f}$ in \cref{eq:lip_sos-f_extension} is shown to be an extension of $f$ in \cref{eq:lip_sos-f} by the assumption that $s(x;A)$ is nondecreasing in $A$ after fixing $x$.
Assuming an sos version of this nondecreasing property, we can also show that $\hatfsos$ is an extension of $\fsos$.

\begin{lemma}[$\hatfsos$ extends $\fsos$]
  \label{lem:lip_sos-hatfsos_extension}
  Consider $\fsos$ and $\hatfsos$ defined in  \cref{eq:lip_sos-f_sos} and \cref{eq:lip_sos-f_sos_extension} respectively.
  Suppose for every $t$ that
  \begin{equation}
    \label{eq:lip_sos-nondecreasing}
    \Set{s(x,C) \ge t} \cup \cP_1(x) \cup \cP_2(C;A,d) 
    \sststile{}{x,C}
    \Set{s(x;A) \ge t} \cup \cP_1(x) \,.
  \end{equation}
  Then there exists an $\ell^*\ge\ell$ such that as long as $\ell'\ge\ell^*$, we have $\hatfsos(A) = \fsos(A)$ for all $A\in\bbG_{n,d}$.
\end{lemma}

The reason why we might need to take the level $\ell'$ of pseudo-distributions in $\hatfsos$ larger than the level $\ell$ in $\fsos$ is that additional variables $(C_{ij})$ are introduced to $\hatfsos$.
The level $\ell^*$ depends on the degree of the sos proof in \cref{eq:lip_sos-nondecreasing}. The larger the degree of that sos proof, the larger $\ell^*$ needs to be.
However, in cases when $s(x,C)$ and $s(x;A)$ are polynomials of the same degree and the sos proof in \cref{eq:lip_sos-nondecreasing} is equal to the degree of $s(x;A)$, we can take $\ell^*=\ell$.
The polynomial system in \cref{eq:combined_poly_sys} is such an example, as show in \cref{lem:abstractextension}.

\cref{lem:lip_sos-hatfsos_extension} follows directly from the following two lemmas.

\begin{lemma}
  \label{lem:lip_sos-hatfsos_extension_le}
  Consider $\fsos$ and $\hatfsos$ defined in  \cref{eq:lip_sos-f_sos} and \cref{eq:lip_sos-f_sos_extension} respectively.
  Suppose for every $t$ that
  \[
    \Set{s(x,C) \ge t} \cup \cP_1(x) \cup \cP_2(C;A,d) 
    \sststile{}{}
    \Set{s(x;A) \ge t} \cup \cP_1(x) \,,
  \]
  Then there exists an $\ell^*\geq\ell$ such that as long as $\ell'\ge\ell^*$, we have $\fsos(A) \ge \hatfsos(A)$ for any $A\in\bbG_{n,d}$.
\end{lemma}
\begin{proof}
  Fix an arbitrary $A\in\bbG_{n,d}$.
  Suppose $\hatfsos(A)=t$ witnessed by a degree-$\ell'$ pseudo-distribution $\mu' \sdtstile{\deg(x)\le\ell}{x,C} 
  \set{s(x,C) \ge t} \cup \cP_1(x) \cup \cP_2(C;A,d)$.
  We construct a degree-$\ell$ pseudo-distribution $\mu$ by setting $\tE_\mu p(x) = \tE_{\mu'} p(x)$ for every polynomial $p(x)$ of degree at most $\ell$.
  Using our assumption, we can make $\mu \sdtstile{}{x} \set{s(x;A) \ge t} \cup \cP_1(x)$ by choosing a large enough $\ell'$.
  Thus, $\fsos(A) \ge t = \hatfsos(A)$.  
\end{proof}

\begin{lemma}
  \label{lem:lip_sos-hatfsos_extension_ge}
  Consider $\fsos$ and $\hatfsos$ defined in \cref{eq:lip_sos-f_sos} and \cref{eq:lip_sos-f_sos_extension} respectively.
  For any $\ell'\ge\ell$, we have $\hatfsos(A) \ge \fsos(A)$ for all $A\in\bbG_{n,d}$.
\end{lemma}
\begin{proof}
  Fix an arbitrary $A\in\bbG_{n,d}$ and an arbitrary $\ell'\ge\ell$.
  Suppose $\fsos(A) = t$ witnessed by a degree-$\ell$ pseudo-distribution $\mu \sdtstile{}{x} \set{s(x;A) \ge t} \cup \cP_1(x) \,. $
  We construct a \mbox{degree-$\ell'$} pseudo-distribution $\mu'$ in variables $x,C$ as follows. 
  For every polynomial $p(x,C)$ of degree in $x$ at most $\ell$ and total degree at most $\ell'$, let $\tE_{\mu'} p(x,C) = \tE_{\mu} p(x;A)$.
  Note $\tE_{\mu} p(x;A)$ is well-defined as $p(x;A)$ is a polynomial in variables $x$ of degree at most $\ell$.
  Then it is a routine work to verify $\mu' \sdtstile{\deg(x)\le\ell}{x,C} \Set{s(x,C) \ge t} \cup \cP_1(x) \cup \cP_2(C;A,d)$.
  Thus $\hatfsos(A) \ge t = \fsos(A)$.
\end{proof}

\begin{proof}[Proof of \cref{lem:lip_sos-hatfsos_extension}]
  By \cref{lem:lip_sos-hatfsos_extension_le} and \cref{lem:lip_sos-hatfsos_extension_ge}.
\end{proof}

%% file: content/Acknowledgement.tex
\section{Acknowledgement}
The authors are grateful to Afonso S.Bandeira for helpful discussions. 
We thank Rares Darius Buhai for feedbacks on the earlier version of the paper.
We thank the anonymous reviewers for their valuable comments and suggestions.
Hongjie Chen, Jingqiu Ding, Yiding Hua and David Steurer are supported by funding from the European Research Council (ERC) under the European Union’s Horizon 2020 research and innovation programme (grant agreement No 815464).
Chih-Hung Liu is supported by Ministry of Education, Taiwan under Yushan Fellow Pro gram with the grant number NTU-112V1023-2 and by National Science and Technology Council, Taiwan with the grant number 111-2222-E-002-017-MY2.
Tommaso D'Orsi is supported by the project MUR FARE2020 PAReCoDi.

%% file: content/subsample.tex
\section{Algorithm based on subsample and aggregation}

In this section, we present a simple differentially private algorithm for estimating the graphon from a graph $G$ with $n$ vertices.
The idea is based on the classical framework, namely subsample and aggregate. 
We first randomly partition the vertex set $[n]$ into $m$ disjoint subsets $V_1,\ldots,V_m$ of equal size, and then for each $t\in [m]$, we run the non-private algorithms on the induced subgraph on $V_i$. 
As result, for each induced subgraph, we obtain an estimate of the graphon.
Finally, we aggregate the outputs of the $m$ non-private algorithms to obtain the final estimate.

For simplification, we assume that the graphon $W$ is $k$-block, and we only consider the case of dense graphons.

\subsection{Non-private algorithm}
We first describe a simple non-private algorithm for graphon estimation.
\begin{algorithmbox}[Non-private algorithm for graphon estimation]
    \label{algo:non-private}
    \mbox{}\\
    \textbf{Input:} Adjacency matrix $\bm A$, privacy parameter $\epsilon>0$.

    \noindent
    \textbf{Output:} $\e$-differentially private graphon estimator $\hat{B}\in \R^{k\times k}$.

    \begin{enumerate}[1.]
        \item Find the best rank-$k$ approximation $\hat{\bm Q}$ of the adjacency matrix $\bm A$;
        \item  Run $O(1)$ approximation algorithm for balanced k-means\cite{ding2020faster} on the row vectors of $\hat{\bm Q}$. 
        Obtain 
        \begin{itemize}
            \item $\tilde{\bm Q}\in \R^{n\times n}$ where only $k$ rows are different;
            \item $\hat{\bm Z}\in \{0,1\}^{n\times k}$ which records the cluster of $i$-th row vector for $i\in [n]$;
        \end{itemize}  
        \item Output $\hat{\bm B}=\frac{k^2}{n^2}\hat{\bm Z}^\top\tilde{\bm Q}\hat{\bm Z}$, i.e we average entries in $\tilde{\bm Q}$ which belong to the same cluster.
    \end{enumerate}
\end{algorithmbox}

\begin{theorem}
    \label{thm:non-private}
    For an arbitrary integer $k=O(1)$, consider a $k$-block graphon $W$, and graph $\bm G$ generated from $W$.
    Let the edge connection probability matrix for $\bm G$ be $\bm Q_0\in \R^{n\times n}$.
    Then the non-private algorithm in \cref{algo:non-private} outputs an estimate $\hat{\bm B}\in \R^{k\times k}$ such that with high probability, we have
    \begin{equation*}
        \frac{1}{n}\normf{\bm Z\hat{\bm B}\bm Z^\top-\bm Q_0}\leq O\Paren{\sqrt{\frac{k}{n}}+\Paren{\frac{k}{n}}^{1/4}}
    \end{equation*}
    for some balanced community membership matrix $\bm Z\in \{0,1\}^{n\times k}$.
    Furthermore, the running time of the algorithm is bounded by $n^{\poly(k)}$.
\end{theorem}

We split the proof in \cref{lem:non-private-pca} and \cref{lem:non-private-kmeans}.
\begin{lemma}
    \label{lem:non-private-pca}
    With high probability over $\bm A$, we have $\normf{\bm Q_0-\hat{\bm Q}}\leq O(\sqrt{kn})$.
\end{lemma}
\begin{proof}
    For the first step, we have
    \begin{equation*}
        \normf{\bm A- \hat{\bm Q}}\leq \normf{\bm A-\bm Q_0}\,. 
    \end{equation*}
    As result, we have
    \begin{equation*}
        \normf{\bm Q_0-\hat{\bm Q}}^2\leq 2\iprod{\bm A-\bm Q_0,\bm Q_0-\hat{\bm Q}}\,. 
    \end{equation*}
    Since $\norm{\bm A-\bm Q_0}\leq O(\sqrt{n})$ with probability at least $1-\frac{1}{n^{100}}$, and $\bm Q_0-\hat{\bm Q}$ has rank $2k$, we have
    \begin{equation*}
        \iprod{\bm A-\bm Q_0,\bm Q_0-\hat{\bm Q}}\leq O(\sqrt{k\cdot n})\cdot \normf{\bm Q_0-\hat{\bm Q}}\,.
    \end{equation*}
    Therefore we conclude that $\normf{\bm Q_0-\hat{\bm Q}}\leq O(\sqrt{k\cdot n})$ with probability at least $1-\frac{1}{n^{100}}$.
\end{proof}

\begin{lemma}
    \label{lem:non-private-kmeans}
    Consider the setting of \cref{thm:non-private}. 
    Suppose for some balanced community membership matrix $\bm Z_0\in \{0,1\}^{n\times k}$ and $B_0\in \R^{k\times k}$, we have $\normf{\bm Z_0B_0\bm Z_0^\top-\hat{\bm Q}}\leq \delta$. 
    Then there exists balanced community membership matrix $\bm Z\in \{0,1\}^{n\times k}$ such that $\Normf{\bm Z\hat{\bm B}\bm Z^\top-\bm Z_0B_0\bm Z_0^\top}\leq O\Paren{\delta}$.
\end{lemma}
\begin{proof}
Let $\hat{\bm Z}\in \{0,1\}^{n\times k}$ be the clustering obtained from the balanced k-means algorithm(theorem 4 in \cite{ding2020faster}).
Then $\frac{k}{n}\hat{\bm Z}\hat{\bm Z}^\top \tilde{\bm Q}=\tilde{\bm Q}$ as we are averaging the rows in the same clusters. 
On the other hand, by the approximation ratio guarantee of the balanced k-means algorithm\cite{ding2020faster}, we have $\normf{\tilde{\bm Q}-\bm Q_0}\leq O(\delta)$.
This implies  $\normf{\tilde{\bm Q}-\tilde{\bm Q}^\top}\leq O(\delta)$ since $\bm Q_0$ is symmetric. 
It follows that 
 \begin{equation*}
    \Normf{\frac{k}{n}\Paren{\tilde{\bm Q}-\tilde{\bm Q}^\top} \hat{\bm Z}\hat{\bm Z}^\top}\leq O(\delta)\,,
 \end{equation*}   
    since $\frac{k}{n}\hat{\bm Z}\hat{\bm Z}^\top$, as an operation for averaging columns, only reduces the Frobenius norm of a matrix. 

    However, $\frac{k}{n}\tilde{\bm Q}^\top \hat{\bm Z}\hat{\bm Z}^\top=\tilde{\bm Q}$ since the columns correspond to the same cluster are the same in $\tilde{\bm Q}^\top$. As a result, we have $\frac{k}{n}\tilde{\bm Q}^\top \hat{\bm Z}\hat{\bm Z}^\top=\tilde{\bm Q}$. 
    It follows that 
    \begin{equation*}
         \Normf{\frac{k}{n}\tilde{\bm Q}\hat{\bm Z}\hat{\bm Z}^\top-\tilde{\bm Q}^\top}\leq O(\delta)\,,
    \end{equation*}
    By triangle inequality, we again have 
    \begin{equation*}
        \Normf{\frac{k}{n}\tilde{\bm Q}\hat{\bm Z}\hat{\bm Z}^\top-\bm Q_0}\leq O(\delta)\,,
    \end{equation*}
    which finishes the proof.
\end{proof}

\begin{proof}[Proof of \cref{thm:non-private}]
    By \cref{lem:non-private-pca}, we have $\normf{\bm Q_0-\hat{\bm Q}}\leq O(\sqrt{kn})$ with high probability. 
    By \cref{lem:samplingupperbound}, with high probability, we have $\frac{1}{n}\normf{\bm Z_0 B_0 \bm Z_0^\top-\bm Q_0}\leq O\Paren{\Paren{\frac{k}{n}}^{1/4}}$.
    By the triangle inequality, we have 
    \begin{equation*}
        \frac{1}{n}\normf{\bm Z_0B_0\bm Z_0^\top-\hat{\bm Q}}\leq \frac{1}{n}\normf{\bm Z_0B_0\bm Z_0^\top-\bm Q_0}+\frac{1}{n}\normf{\hat{\bm Q}-\bm Q_0}\leq O\Paren{\sqrt{\frac{k}{n}}+\Paren{\frac{k}{n}}^{1/4}}\,.
    \end{equation*}
    By \cref{lem:non-private-kmeans}, with high probability, we have $\frac{1}{n}\normf{\hat{\bm Z}\hat{\bm B}\hat{\bm Z}^\top-\hat{\bm Q}}\leq O\Paren{\sqrt{\frac{k}{n}}+\Paren{\frac{k}{n}}^{1/4}}$.
    Therefore we conclude that $\frac{1}{n}\normf{\hat{\bm Z}\hat{\bm B}\hat{\bm Z}^\top-\bm Q_0}\leq O\Paren{\sqrt{\frac{k}{n}}+\Paren{\frac{k}{n}}^{1/4}}$.    

    Finally, for the running time, the best rank-$k$ approximation takes $\poly(n)$ time to compute while the balanced k-means algorithm with $O(1)$ approximation ratio takes $n^{\poly(k)}$ time to compute(by taking $\Delta=\frac{1}{\poly(n)}$ in the theorem 4 of \cite{ding2020faster}). 
    Therefore the running time of the algorithm is bounded by $n^{\poly(k)}$.
\end{proof}

\subsection{Private algorithm based on subsample and aggregation}
We describe the algorithm in \cref{algo:subsample-and-aggregate}.
\begin{algorithmbox}[Private algorithm based on subsample and aggregation]
    \label{algo:subsample-and-aggregate}
    \mbox{}\\
    \textbf{Input:} Adjacency matrix $A$, privacy parameter $\epsilon>0$.

    \noindent
    \textbf{Output:} $\e$-differentially private graphon estimator $\hat{W}$.

    \begin{enumerate}[1.]
        \item Randomly partition the vertex set $[n]$ into $m=\frac{k^3\log(n)}{\e}$ disjoint subsets $V_1,\ldots,V_T$ of equal size.
        \item For each $t\in [m]$, run the non-private \cref{algo:non-private} on the induced subgraph $G[V_t]$ to obtain the estimate $\hat{\bm B}_t$.
        \item Aggregation: apply the exponential mechanism $\bm B\propto\exp\Paren{-\frac{\epsilon \Score(B)}{2}}$ with score function being 
        \(\Score(B)=\Abs{\Set{t\in [m]:\delta_2(\hat{\bm B}_t,B)\leq \Paren{\sqrt{\frac{k}{n}}+\Paren{\frac{k}{n}}^{1/4}}\cdot \poly(k/\e)\cdot\polylog(n)}}\)
    \end{enumerate}
\end{algorithmbox}

\begin{theorem}
    \label{thm:subsample-and-aggregate}
    For an arbitrary integer $k=O(1)$, suppose the underlying graphon $W$ is $k$-block graphon, generated from the matrix $B_0\in \R^{k\times k}$.
    Then given the random graph $\bm G$, the \cref{algo:subsample-and-aggregate} is $\e$-differentially private, and runs in $n^{\poly(k)}$ time. 
    Further it outputs a matrix $\hat{\bm B}$ such that $\delta_2(\hat{\bm B}, B_0)\leq \Paren{\sqrt{\frac{k}{n}}+\Paren{\frac{k}{n}}^{1/4}}\cdot \poly(k/\e)\cdot\polylog(n)$.
\end{theorem}

We first prove that our algorithm is differentially private.

\begin{lemma}\label{lem:privacy-sa}
    The \cref{algo:subsample-and-aggregate} is $\epsilon$-differentially private. 
\end{lemma}

\begin{proof}
    By \cite{subsample-and-aggregate}, we only need to show that the exponential mechanism is differentially private with respect to inputs $\hat{B}_1,\hat{B}_2,\ldots,\hat{B}_m$.
    Since for neighboring dataset, the sensitivity of the score function is bounded by $1$.
     By the promise of exponential mechanism, the algorithm is differentially private.
\end{proof}

Next we prove that the utility guarantees of our algorithm. 
\begin{lemma}\label{lem:utility-sa}
    With high probability, \cref{algo:subsample-and-aggregate} outputs a matrix $\hat{\bm B}$ such that $\delta_2(\hat{\bm B}, B_0)\leq \Paren{\sqrt{\frac{k}{n}}+\Paren{\frac{k}{n}}^{1/4}}\cdot \poly(k/\e)\cdot\polylog(n)$.
\end{lemma}
\begin{proof}
    Let $n'=\frac{n}{m}$.
    By the guarantees of the non-private algorithm in \cref{thm:non-private}, with high probability, we have $\delta_2(\hat{B_t},B_0)\leq \poly(m)\cdot \Paren{\sqrt{\frac{k}{n'}}+\Paren{\frac{k}{n'}}^{1/4}}$ for $t\in [m]$.
    With high probability, the value of score function is lower bounded by $m-O\Paren{\frac{k^2 \log(n)}{\e}}$.
    Therefore when $m\gg \frac{k^3\log(n)}{\e}$, with high probability, $1-o(1)$ fraction among $\hat{B_1},\hat{B_2},\ldots,\hat{B_t}$ are within distance $\Paren{\sqrt{\frac{k}{n}}+\Paren{\frac{k}{n}}^{1/4}}\cdot \poly(k/\e)\cdot\polylog(n)$ to $B_0$.
    As result, we have 
    \begin{equation*}
        \delta_2(\hat{\bm B},B_0)\leq \poly(k/\e)\cdot\polylog(n)\cdot \Paren{\sqrt{\frac{k}{n'}}+\Paren{\frac{k}{n'}}^{1/4}}\,.
    \end{equation*}
\end{proof}

Now we finish the proof of \cref{thm:subsample-and-aggregate}.
\begin{proof}[Proof of \cref{thm:subsample-and-aggregate}]
    By \cref{lem:privacy-sa}, the algorithm is $\e$-differentially private.
    By \cref{lem:utility-sa}, with high probability, the algorithm outputs a matrix $\hat{\bm B}$ such that $\delta_2(\hat{\bm B}, B_0)\leq \Paren{\sqrt{\frac{k}{n}}+\Paren{\frac{k}{n}}^{1/4}}\cdot \frac{\poly(k)\cdot\polylog(n)}{\e^2}$.
    Finally, since the algorithm for each fold of the data runs in time $n^{\poly(k)}$ and the sampling algorithm can be implemented in time $n^{\poly(k)}$, the total running time of the algorithm is bounded by $n^{\poly(k)}$.
\end{proof}

%% file: content/densityestimation.tex
\section{Private estimation of edge density}\label{sec:privatedensityestimation}
In this section, we provide a private and robust algorithm for estimating the number of edges in the random graph.
A naive and natural private estimator is $\hat{\bm \rho}= \rho(\bm G)+\Lap\Paren{\frac{10}{n\epsilon}}$.
As shown in~\cite{SmoothErdosRenyi,borgs2018revealing}, it is possible to beat the guarantees of this private estimator for Erdos Renyi random graphs.
In this paper, we need to study the edge density estimation for more general random graphs.

\subsection{Coarse average edge density estimation for random graphs}
Now we prove the utility guarantees of the naive algorithm, which adds Laplacian noise to the empirical edge density.
\begin{lemma}[Private estimation of edge density]\label{lem:private_edge_density}
  For a $n$-vertex random graph $\bm G$ with each pair of vertices $i,j$ connected with probability $Q(i,j)\in [0,1]$, let the average density $\rho(\bm Q)\coloneqq \sum_{i,j\in [n]}\frac{Q(i,j)}{n^2}$ and let $\bm \rho(\bm G)\coloneqq \frac{2|E(\bm G)|}{n(n-1)}$.
  For $\epsilon=O(1)$, let $\hat{\rho}=\bm \rho(\bm G)+\Lap\Paren{\frac{10}{n\epsilon}}$.
  Then with probability at least $1-\frac{1}{n^{10}}$, we have $\Abs{\hat{\bm \rho}-\rho(\bm Q)}\leq \frac{200\log(n)}{n\epsilon}+\frac{\rho(\bm Q)}{10}$\,.
\end{lemma}
\begin{proof}
  By the density of Laplacian distribution, for $\eta\sim \Lap\Paren{\frac{10}{n\epsilon}}$, with probability $1-\exp(-t)$, we have $|\eta|\leq \frac{20t}{n\epsilon}$.
 Taking $t=10\log(n)$, with probability at least $1-\frac{1}{n^{10}}$, we have $|\eta|\leq \frac{200\log(n)}{n\epsilon}$.
  
    On the other hand, by taking $\alpha=R$ in \cref{lem:convergence-edge-density}, with probability at least $1-\frac{1}{n^{100}}$, we have $|\rho(\bm Q)-\rho(\bm G)|\leq \frac{\rho(\bm Q)}{10}$. 
    By union bound, we have the claim.
\end{proof}

\subsection{Target edge density estimation with higher accuracy}
For achieving higher accuracy in estimating the edge density, we use the following result from~\cite{SmoothErdosRenyi}\footnote{Particularly they consider degree $D$ concentrated graphs, which include the graphs with degree bounded by $D$ as a special case.}.

\begin{theorem}[Corollary of theorem 3.3 in~\cite{SmoothErdosRenyi}]
    For any $n$-vertex graph with degree bounded by $D$, there is a $\e$-differentially (node) private algorithm, which given the graph $\bm G$ and degree bound $D$, outputs $\hat{\bm \rho}$ such that 
    \begin{equation*}
        \E \Brac{\rho(G)-\hat{\bm \rho}}^2\leq O\Paren{\frac{D^2}{\e^2 n^4}+\frac{1}{\e^4 n^4}}\,.
    \end{equation*}
    Furthermore the algorithm runs in $O(n^2)$ time.
\end{theorem}

By repeatedly running the algorithm and taking the median of the estimators, we obtain the following corollary
\begin{corollary}\label{cor:privateDensitydegreeconcentrated}
    For any $n$-vertex graph with degree bounded by $D$, there is a $\e$-differentially (node) private algorithm, which given the graph $\bm G$ and degree bound $D$, outputs $\hat{\bm \rho}$ such that with probability at least $1-\frac{1}{n^{100}}$
    \begin{equation*}
        (\rho(\bm G)-\hat{\bm \rho})^2\leq O\Paren{\frac{\log^2(n) D^2}{\e^2 n^4}+\frac{\log^4(n)}{\e^4 n^4}}\,.
    \end{equation*}
    Furthermore the algorithm runs in $\poly(n)$ time.
\end{corollary}

We describe our algorithm in \cref{algo:privateDensityEstimation}.
\begin{algorithmbox}[Private algorithm for estimation of target edge density]
    \label{algo:privateDensityEstimation}
    \mbox{}\\
    \textbf{Input:} Adjacency matrix $A$, privacy parameter $\epsilon>0$, and a bound of underlying graphon $R\geq 20 \Lambda =20 \normi{W}$.

    \noindent
    \textbf{Output:} $\e$-differentially private target edge density estimator $\hat{\bm \rho}$.
    \begin{enumerate}[1.]
    \item  Run algorithm 1 from~\cite{SmoothErdosRenyi} with privacy parameter $\frac{\e}{10\log(n)}$ for $\log(n)$ independent rounds.
    \item Return the median of the $\log(n)$ estimators obtained in step 2 as $\hat{\bm \rho}$.
    \end{enumerate}
\end{algorithmbox}

\begin{proof}[Proof of \cref{cor:privateDensitydegreeconcentrated}]
    By the composition theorem and post-processing theorem, the algorithm is differentially private. 
    For the utility guarantees, we   that for each of the $\log(n)$ estimators $\hat{\bm \rho}_i$, by Markov inequality,
    \begin{equation*}
        \Pr\Brac{(\rho(\bm G)-\hat{\bm \rho}_i)^2\lesssim \frac{\log^2(n) D^2}{\e^2 n^4}+\frac{\log^4(n)}{\e^4 n^4}}\geq 1-\frac{1}{100}\,.
    \end{equation*}
    By independently running $\log(n)$ rounds, we have 
    \begin{equation*}
            \Pr\Brac{(\rho(\bm G)-\hat{\bm \rho}_i)^2\lesssim \frac{\log^2(n) D^2}{\e^2 n^4}+\frac{\log^4(n)}{\e^4 n^4}}\geq 1-\exp(-100\log(n))=1-\frac{1}{n^{100}}\,.
    \end{equation*}
\end{proof}

%% file: content/robustgraphonestimation.tex
\section{Robust estimation of stochastic block model}\label{sec:robustness}

In this section, we present a $\poly(n,k)$-time algorithm which robustly estimates block-connectivity matrices for stochastic block model.

\subsection{Node-robust average edge density estimation}
First, we give a polynomial time algorithm for robustly estimating the average edge density in the graph.
Our algorithm can be described as 
\begin{algorithmbox}[Node-robust edge density estimation]\label{algo:node-robust-density-estimation}
        \mbox{}\\
        \textbf{Input:} Corrupted adjacency matrix $\bm A$, corruption fraction $\eta$.
        \begin{enumerate}[1.]
        \item Consider the level-$O(1)$ sum-of-squares relaxation of the following integer program: 
        \begin{align}\label{eq:robust-density-estimation}
            &\min \frac{2}{n(n-1)}\iprod{\bm A,(1-x)(1-x)^\top}\\
            \text{s.t\qquad} &x_i^2=x_i\nonumber\\
            &\sum_{i=1}^n x_i\leq \eta n \nonumber
        \end{align}
        \item Return the objective value $\hat{\rho}$.
        \end{enumerate}
\end{algorithmbox}

We give the robust density estimation guarantees of the algorithm in the stochastic block models.
\begin{theorem}\label{thm:robust-density-estimation}
    Consider the following model: $\bm G_0\sim \SBM(n,d,B_0)$ and $\bm G$ is obtained from $\bm G_0$ by arbitrarily corrupting $\eta\cdot n$ vertices. 
    Suppose $\normi{B_0}\leq R$.
    Then with high probability over $\bm G_0$, the output of \cref{algo:node-robust-density-estimation} satisfies
    \begin{equation*}
        \abs{\hat{\bm \rho}-\bm \rho}\leq O\Paren{\frac{\eta R d\log(n)}{n}+\frac{\sqrt{\rho}\log(n)}{n}}\,.
    \end{equation*}
\end{theorem}

We split the proof of the theorem into two lemmas.

\begin{lemma}\label{lem:robust-density-ub}
   The objective value of the level-$O(1)$ sum-of-squares relaxation of \cref{eq:robust-density-estimation} is upper bounded by $\rho(\bm G_0)$.
\end{lemma}
\begin{proof}
  Let $x^*\in \Set{0,1}^n$ be the indicator vector of the corrupted vertices. 
  Then $x^*$ is a feasible solution to the program \cref{eq:robust-density-estimation}, and the corresponding objective value is upper bounded by 
  \begin{equation*}
    \frac{2}{n(n-1)}|E(\bm G_0)|=\rho(\bm G_0)\,.
  \end{equation*}
\end{proof}

\begin{lemma}\label{lem:robust-density-lb}
    With probability at least $1-\frac{1}{n^{100}}$, the objective value of the level-$O(1)$ sum-of-squares relaxation of \cref{eq:robust-density-estimation} is lower bounded by $\rho(\bm G_0)-O\Paren{\frac{\eta R d\log(n)}{n}}$.
\end{lemma}
\begin{proof}
    Let $x^*\in \Set{0,1}^n$ be the indicator vector for the set of the corrupted vertices. 
    Then we have
    \begin{align*}
       \iprod{\bm A,(1-x)(1-x)^\top}=\iprod{\bm A\odot (1-x^*) (1-x^{*})^\top,(1-x)(1-x)^\top}\\
       +2\iprod{\bm A\odot x^*1^\top, (1-x)(1-x)^\top}-\iprod{\bm A\odot x^*x^{*\top},(1-x)(1-x)^\top}\,.
    \end{align*}
    First, 
    \begin{align*}
        \iprod{\bm A\odot (1-x^*) (1-x^{*})^\top,(1-x)(1-x)^\top}&=\iprod{\bm A_0\odot (1-x^*) (1-x^{*})^\top,(1-x)(1-x)^\top}\\
        &\geq \Norm{\bm A_0\odot (1-x^*) (1-x^{*})^\top}_1+\eta\cdot n \cdot \norm{\bm A_0}\,.
    \end{align*}
    Now since $\Norm{\bm A_0\odot (1-x^*) (1-x^{*})^\top}_1$ is given by the number of edges incident to uncorrupted vertices, with probability at least $1-\frac{1}{n^{100}}$, we have
    \begin{equation*}
        \Normf{\bm A_0\odot (1-x^*) (1-x^{*})^\top}^2\leq \Abs{E(\bm G)}-O\Paren{\eta nR d\log(n)}\,.
    \end{equation*}
    Moreover, we have $\Norm{\bm A_0}\leq O\Paren{R d\log(n)}$.
    Therefore, 
    \begin{equation*}
        \iprod{\bm A\odot (1-x^*) (1-x^{*})^\top,(1-x)(1-x)^\top}\geq \Abs{E(\bm G)}-O\Paren{\eta nR d\log(n)}\,.
    \end{equation*}

    Second, we have
    \begin{equation*}
        2\iprod{\bm A\odot x^*1^\top, (1-x)(1-x)^\top}\geq 0\,.
    \end{equation*}
    Finally, we have 
    \begin{equation*}
        \iprod{\bm A\odot x^*x^{*\top},(1-x)(1-x)^\top}\leq \Normf{\bm A\odot x^*x^{*\top}}=\Normf{\bm A_0\odot x^*x^{*\top}}\leq O\Paren{\eta^2 nR d\log(n)}\,.
    \end{equation*}

    Therefore, we have
    \begin{equation*}
        \iprod{\bm A,(1-x)(1-x)^\top}\geq \Abs{E(\bm G)}-O\Paren{\eta nR d\log(n)}\,.
    \end{equation*}
    Finally, our proof is captured by degree-$O(1)$ sum-of-squares proof.
    Therefore the claim follows.
\end{proof}

\begin{proof}[Proof of \cref{thm:robust-density-estimation}]
    The proof follows from the above two lemmas. On the one hand, by \cref{lem:convergence-edge-density}, with probability at least $1-\exp\Paren{-\frac{n^2 t^2}{10\rho(\bm Q)}}$,
    \begin{equation*}
      \Abs{\rho(\bm G_0)-\rho(\bm Q)}\leq t  \,.
    \end{equation*}
    Therefore, with probability at least $1-\frac{1}{\poly(n)}$, we have $\Abs{\rho(\bm G_0)-\rho(\bm Q)}\lesssim \frac{\sqrt{\rho(\bm Q)}\log(n)}{n}$.
    As result, by \cref{lem:robust-density-ub}, with probability at least $1-\frac{1}{\poly(n)}$, we have $\hat{\bm \rho}\leq \bm \rho+O\Paren{\frac{\sqrt{\bm \rho}\log(n)}{n}}$.
  
    On the other hand, we have the lower bound for robust density estimation. 
    By \cref{lem:robust-density-lb}, we have 
    \begin{equation*}
        \hat{\bm \rho}\geq \rho(\bm G_0)-O\Paren{\frac{\eta R d\log(n)}{n}}\geq \rho-O\Paren{\frac{\eta R d\log(n)}{n}+\frac{\sqrt{\bm \rho}\log(n)}{n}} \,.      
    \end{equation*}

    In all we have
    \begin{equation*}
        \Abs{\hat{\bm\rho}-\bm \rho}\leq O\Paren{\frac{\eta R d\log(n)}{n}+\frac{\sqrt{\bm \rho}\log(n)}{n}}\,.
    \end{equation*}
\end{proof}

\subsection{Robust estimation algorithm}

\begin{theorem}
    \label{thm:robust-estimator}
     Consider balanced stochastic block model (\cref{definition:sbm-balanced}).
      With high probability over $\bm G_0\sim \SBM(n,d,B_0)$, given average degree $d$ and any graph $\bm G$ obtained from $\bm G_0$ by arbitrarily corrupting $\eta\cdot n$ vertices, there is a $\poly(n,k)$-time algorithm which outputs a matrix $\hat{\bm B} \in [0,1]^{k\times k}$, and a community membership matrix $\hat{\bm Z} \in \{0,1\}^{n\times k}$ such that
    \begin{equation*}
        \|\hat{\bm Z}\hat{\bm B}\hat{\bm Z}^\top-\bm \Znull \Bnull \bm\Znull^{\top}\|_F^2\leq O_{R}\Paren{\frac{n^2}{d}\cdot k+\eta\cdot n^2} \,,
    \end{equation*}
    where $O_{R}$ hides $R$ which is the upper bound of entries in $\Bnull$.
\end{theorem}

As a corollary, we obtain a robust estimator for $k$-block graphon, because the matrix $\hat{Z}\hat{B}\hat{Z}^\top$ here is a k-block matrix and we can output a $k$-block-wise constant function as our estimator for the underlying graphon.

The key ingredient of our robust algorithm is the following integer optimization program, where the observed corrupted graph is denoted as $A$ and the hidden uncorrupted graph is denoted as $A_0$.

\begin{equation}
\label{eq:robust-program}
    \begin{aligned}
    \min_{Z, B, C} \quad & \|ZBZ^\top\|_F^2 - 2 \iprod{\frac{n}{d} C, ZBZ^{\top}} \\
    \textrm{s.t.} \quad & Z \in \{0,1\}^{n\times k} \\
      & \sum_{a=1}^k Z_{i,a} = 1 \quad \forall i \in [n] \\
      & \sum_{i=1}^n Z_{i,a} = \frac{n}{k} \quad \forall a \in [k] \\
      & B \in \{0,R\}^{k \times k} \\
      & C \in[0,1]^{n \times n} \text { is symmetric}\\
      & C_{i, j} \leq A_{i, j} \quad \forall i, j \in [n] \\
      & \sum_{j \in [n]} C_{i, j} \leq 20 R d \quad \forall i \in[n] \\
    \end{aligned}
\end{equation}

Since the program in \cref{eq:robust-program} is an integer program which cannot be efficiently solved, we apply sum-of-squares relaxation in our algorithm to obtain a polynomial-time computable approximation. We can show that the approximate solution of the sum-of-squares program is close enough to the true value $\bm \Znull \Bnull \bm\Znull^{\top}$.

Given the approximate solution, we can apply balanced-$k$-means algorithm to obtain an integral solution $\hat{Z}$ and the corresponding $\hat{B}$. It can be shown that $\hat{Z}\hat{B}\hat{Z}^\top$ is also close to the true value $\bm \Znull \Bnull \bm\Znull^{\top}$.

\begin{algorithmbox}[Node-robust polynomial time algorithm]
\label{algo:node-robust}
	\mbox{}\\
	\textbf{Input:} Corrupted adjacency matrix $A$.
	\begin{enumerate}[1.]
    \item Obtain level-$O(1)$ pseudo-distribution $\mu$ for program \cref{eq:robust-program}.
    \item Run $O(1)$-approximation algorithm for k-means\cite{KANUNGO200489} on the row vectors of $\tilde{\E}_{\mu} [ZBZ^\top]$ to obtain:
    \begin{itemize}
        \item $\hat{Z} \in \{0,1\}^{n\times k}$ which is a community membership matrix;
        \item $\tilde{\theta}\in \R^{n\times n}$ where row $i$ is the center of cluster for $i$-th vertex.
        \item diagonal matrix $D\in \R^{k\times k}$, where the $i$-th entry is given by $\frac{1}{n_i}$, with $n_i$ being the number of vertices in the $i$-th cluster.
    \end{itemize}
    \item Compute the k-by-k block connectivity matrix $\hat{B} = D\hat{Z}^\top \tilde{\theta} \hat{Z}D$ by averaging entries in $\tilde{\theta}$ according to community membership matrix $\hat{Z}$.
    \item Output $\hat{Z}$ and $\hat{B}$. 
	\end{enumerate}
\end{algorithmbox}

\begin{proof}[Proof sketch]
    In Step 1 of the algorithm, we can show that, with high probability,
    \begin{equation*}
        \|\tilde{\E} [ZBZ^\top]-\bm \Znull \Bnull \bm\Znull^{\top}\|_F^2 \leq O_{R} \Paren{\frac{n^2}{d}\cdot k+\eta\cdot n^2}\,.
    \end{equation*}

    In step 2 of the algorithm, using the constant approximation ratio guarantee of balanced-$k$-means, it follows that
    \begin{equation*}
        \|\tilde{\theta}-\bm \Znull \Bnull \bm\Znull^{\top}\|_F^2 \leq O_{R} \Paren{\frac{n^2}{d}\cdot k+\eta\cdot n^2}\,.
    \end{equation*}
    
    Finally, using the fact that $\hat{Z}$ is a balanced community membership matrix, we can show that
    \begin{equation*}
        \Normf{\hat{Z}\hat{B}\hat{Z}^\top- \bm \Znull \Bnull \bm\Znull^{\top}}^2 \leq O_{R} \Paren{\frac{n^2}{d}\cdot k+\eta\cdot n^2} \,.
    \end{equation*}
\end{proof}

\subsection{Proofs}

\begin{lemma}
    In \cref{algo:node-robust}, with high probability, the estimator $\tilde{\E} [ZBZ^\top]$ obtained in step 1 satisfies
    \begin{equation*}
        \|\tilde{\E} [ZBZ^\top]-\bm \Znull \Bnull \bm\Znull^{\top}\|_F^2 \leq O_{R} \Paren{\frac{n^2}{d} \cdot k+\eta\cdot n^2}\,.
    \end{equation*}
\end{lemma}

\begin{proof}
    Throughout the proof, we will use the following notation:
    \begin{itemize}
        \item Let $S_1$ denote the set of uncorrupted nodes.
        \item Given matrix $M \in \R^{n \times n}$ and set $S \in [n]$, we use $M_S$ to denote matrix $M \odot (\Ind_S \Ind_S^{\top})$, that is we zero out all rows and columns of $M$ that is not in $S$.
    \end{itemize}
    
    By \cref{theorem:pruned-spectral-norm} and \cref{lem:degree-counting}, it follows that there exists a set of nodes $S_2$ such that
    \begin{equation*}
        \Norm{\Paren{\frac{n}{d} A_0 - \bm \Znull \Bnull \bm\Znull^{\top}}_{S_2}}
        \leq O \Paren{\frac{n}{\sqrt{d}}} \,,
    \end{equation*}
    and, with probability $1-\exp(-d)$, the size of $S_2$ is at least $1 - O\Paren{\exp(-d) \cdot n}$.

    Now, consider the set of nodes $S_1 \cap S_2$. It follows that, with probability $1-\exp(-d)$, the size of $S_1 \cap S_2$ is at least $1 - O\Paren{\exp(-d) \cdot n} - \eta \cdot n$. By monotonicity of spectral norm, it follows that
    \begin{equation*}
        \Norm{\Paren{\frac{n}{d} A_0 - \bm \Znull \Bnull \bm\Znull^{\top}}_{S_1 \cap S_2}}
        \leq O \Paren{\frac{n}{\sqrt{d}}} \,.
    \end{equation*}
    Moreover, since all nodes in $S_1 \cap S_2$ are uncorrupted, we have $(A_0)_{S_1 \cap S_2} = A_{S_1 \cap S_2}$. Therefore, the triplet $(\bm Z_0, B_0, (A_0)_{S_1 \cap S_2})$ is a feasible solution to program \cref{eq:robust-program}, and it follows that
    \begin{equation}
    \label{eq:robust-eq1}
        \tilde{\E} \Brac{\|ZBZ^\top\|_F^2 - 2 \iprod{\frac{n}{d} C, ZBZ^{\top}}}
        \leq \|\bm \Znull \Bnull \bm\Znull^{\top}\|_F^2 - 2 \iprod{\frac{n}{d} (A_0)_{S_1 \cap S_2}, \bm \Znull \Bnull \bm\Znull^{\top}} \,.
    \end{equation}
    Notice that,
    \begin{align*}
        \iprod{\frac{n}{d} C, ZBZ^{\top}}
        & = \iprod{\frac{n}{d} C_{S_1 \cap S_2}, ZBZ^{\top}} + \iprod{\frac{n}{d} (C-C_{S_1 \cap S_2}), ZBZ^{\top}} \\
        & \leq \iprod{\frac{n}{d} A_{S_1 \cap S_2}, ZBZ^{\top}} + O_{R} \Paren{\Paren{\eta + \exp(-d)} n^2} \\
        & = \iprod{\frac{n}{d} (A_0)_{S_1 \cap S_2}, ZBZ^{\top}} + O_{R} \Paren{\Paren{\eta + \exp(-d)} n^2} \,,
    \end{align*}
    where the inequality is from constraint $C_{i,j} \leq A_{i, j}$ and $\sum_{j \in [n]} C_{i,j} \leq 20 R d$, and that the size of $S_1 \cap S_2$ is at least $1 - O\Paren{\exp(-d) \cdot n} - \eta \cdot n$. Plugging this into \cref{eq:robust-eq1}, we get
    \begin{equation*}
    \begin{aligned}
        \tilde{\E} \Brac{\|ZBZ^\top\|_F^2 - 2 \iprod{\frac{n}{d} (A_0)_{S_1 \cap S_2}, ZBZ^{\top}}} & \\
        \leq \|\bm \Znull \Bnull \bm\Znull^{\top}\|_F^2 & - 2 \iprod{\frac{n}{d} (A_0)_{S_1 \cap S_2}, \bm \Znull \Bnull \bm\Znull^{\top}} + O_{R} \Paren{\Paren{\eta + \exp(-d)} n^2} \,.
    \end{aligned}
    \end{equation*}
    Adding $\|\frac{n}{d} (A_0)_{S_1 \cap S_2}\|_F^2$ on both sides, we get
    \begin{equation}
        \label{eq:robust-eq2}
            \tilde{\E} \|ZBZ^\top - \frac{n}{d} (A_0)_{S_1 \cap S_2}\|_F^2
            \leq \|\bm \Znull \Bnull \bm\Znull^{\top} - \frac{n}{d} (A_0)_{S_1 \cap S_2}\|_F^2 + O_{R} \Paren{\Paren{\eta + \exp(-d)} n^2} \,.
    \end{equation}
    The left hand side can be decomposed as
    \begin{align*}
        \tilde{\E} \|ZBZ^\top - \frac{n}{d} (A_0)_{S_1 \cap S_2}\|_F^2
        = & \tilde{\E} \|ZBZ^\top - \bm \Znull \Bnull \bm\Znull^{\top} + \bm \Znull \Bnull \bm\Znull^{\top} - \frac{n}{d} (A_0)_{S_1 \cap S_2}\|_F^2 \\
        = & \tilde{\E} \|ZBZ^\top - \bm \Znull \Bnull \bm\Znull^{\top}\|_F^2 + \|\bm \Znull \Bnull \bm\Znull^{\top} - \frac{n}{d} (A_0)_{S_1 \cap S_2}\|_F^2 \\
        & + 2 \tilde{\E} \iprod{ZBZ^\top - \bm \Znull \Bnull \bm\Znull^{\top}, \bm \Znull \Bnull \bm\Znull^{\top} - \frac{n}{d} (A_0)_{S_1 \cap S_2}} \,.
    \end{align*}
    Plug this into \cref{eq:robust-eq2}, we get
    \begin{equation}
        \label{eq:robust-eq3}
        \begin{aligned}
        \tilde{\E} \|ZBZ^\top - & \bm \Znull \Bnull \bm\Znull^{\top}\|_F^2 \\
        \leq & 2 \tilde{\E} \iprod{\bm \Znull \Bnull \bm\Znull^{\top} - ZBZ^\top, \bm \Znull \Bnull \bm\Znull^{\top} - \frac{n}{d} (A_0)_{S_1 \cap S_2}} + O_{R} \Paren{\Paren{\eta + \exp(-d)} n^2}
        \\
        = & 2 \tilde{\E} \iprod{\bm \Znull \Bnull \bm\Znull^{\top} - ZBZ^\top, \Paren{\bm \Znull \Bnull \bm\Znull^{\top} - \frac{n}{d} A_0}_{S_1 \cap S_2}} \\
        & + 2 \tilde{\E} \iprod{\bm \Znull \Bnull \bm\Znull^{\top}- ZBZ^\top, \bm \Znull \Bnull \bm\Znull^{\top} - \Paren{\bm \Znull \Bnull \bm\Znull^{\top}}_{S_1 \cap S_2}} \\
        & + O_{R} \Paren{\Paren{\eta + \exp(-d)} n^2} \,.        
        \end{aligned}
    \end{equation}
    For the term $\tilde{\E} \iprod{\bm \Znull \Bnull \bm\Znull^{\top} - ZBZ^\top, \Paren{\bm \Znull \Bnull \bm\Znull^{\top} - \frac{n}{d} A_0}_{S_1 \cap S_2}}$, we can apply \cref{cor:sos-spectral-hoelder} and get
    \begin{equation}
        \label{eq:robust-eq4}
        \begin{aligned}
        \tilde{\E} \iprod{\bm \Znull \Bnull \bm\Znull^{\top} - & ZBZ^\top, \Paren{\bm \Znull \Bnull \bm\Znull^{\top} - \frac{n}{d} A_0}_{S_1 \cap S_2}} \\
        \leq & 0.1 \tilde{\E} \Normf{\bm \Znull \Bnull \bm\Znull^{\top} - ZBZ^\top}^2 + O \Paren{k \cdot \Norm{\Paren{\bm \Znull \Bnull \bm\Znull^{\top} - \frac{n}{d} A_0}_{S_1 \cap S_2}}^2} \\
        \leq & 0.1 \tilde{\E} \Normf{\bm \Znull \Bnull \bm\Znull^{\top} - ZBZ^\top}^2 + O \Paren{\frac{n^2}{d} \cdot k} \,.
        \end{aligned}
    \end{equation}
    For the term $\tilde{\E} \iprod{\bm \Znull \Bnull \bm\Znull^{\top}- ZBZ^\top, \bm \Znull \Bnull \bm\Znull^{\top} - \Paren{\bm \Znull \Bnull \bm\Znull^{\top}}_{S_1 \cap S_2}}$, notice that entries of $\tilde{\E} \Brac{\bm \Znull \Bnull \bm\Znull^{\top}- ZBZ^\top}$ are bounded in absolute value by $R$, therefore, it follows that
    \begin{equation*}
        \tilde{\E} \iprod{\bm \Znull \Bnull \bm\Znull^{\top}- ZBZ^\top, \bm \Znull \Bnull \bm\Znull^{\top} - \Paren{\bm \Znull \Bnull \bm\Znull^{\top}}_{S_1 \cap S_2}}
        \leq R \Normsum{\bm \Znull \Bnull \bm\Znull^{\top} - \Paren{\bm \Znull \Bnull \bm\Znull^{\top}}_{S_1 \cap S_2}} \,.
    \end{equation*}
    Since $\bm \Znull \Bnull \bm\Znull^{\top} - \Paren{\bm \Znull \Bnull \bm\Znull^{\top}}_{S_1 \cap S_2}$ has at most $2 \Paren{\eta + \exp(-d)} n^2$ non-zero entries and each entry is bounded by $R$, therefore,
    \begin{equation*}
        \Normsum{\bm \Znull \Bnull \bm\Znull^{\top} - \Paren{\bm \Znull \Bnull \bm\Znull^{\top}}_{S_1 \cap S_2}} \leq 2 R^2 \Paren{\eta + \exp(-d)} n^2 \,.
    \end{equation*}
    Hence,
    \begin{equation}
        \label{eq:robust-eq5}
        \tilde{\E} \iprod{\bm \Znull \Bnull \bm\Znull^{\top}- ZBZ^\top, \bm \Znull \Bnull \bm\Znull^{\top} - \Paren{\bm \Znull \Bnull \bm\Znull^{\top}}_{S_1 \cap S_2}}
        \leq O_{R} \Paren{\Paren{\eta + \exp(-d)} n^2} \,.
    \end{equation}
    Plug \cref{eq:robust-eq4} and \cref{eq:robust-eq5} into \cref{eq:robust-eq3}, we get
    \begin{equation*}
        \tilde{\E} \|ZBZ^\top - \bm \Znull \Bnull \bm\Znull^{\top}\|_F^2
        \leq 0.2 \tilde{\E} \Normf{ZBZ^\top - \bm \Znull \Bnull \bm\Znull^{\top}}^2 + O \Paren{\frac{n^2}{d} \cdot k} + O_{R} \Paren{\Paren{\eta + \exp(-d)} n^2} \,.
    \end{equation*}
    Rearranging terms, it follows that
    \begin{equation*}
        \tilde{\E} \|ZBZ^\top - \bm \Znull \Bnull \bm\Znull^{\top}\|_F^2
        \leq O_{R} \Paren{\frac{n^2}{d} \cdot k + \eta n^2} \,.
    \end{equation*}
    By Jensen's inequality, we get
    \begin{align*}
        \|\tilde{\E} [ZBZ^\top] - \bm \Znull \Bnull \bm\Znull^{\top}\|_F^2
        \leq & \tilde{\E} \|ZBZ^\top - \bm \Znull \Bnull \bm\Znull^{\top}\|_F^2 \\
        \leq & O_{R} \Paren{\frac{n^2}{d} \cdot k + \eta n^2} \,.
    \end{align*}
\end{proof}

Now, we are ready to prove \cref{thm:robust-estimator}.

\begin{proof}[Proof of \cref{thm:robust-estimator}]
    From the $O(1)$-approximation guarantee of k-means algorithm, it follows that
    \begin{equation*}
        \Normf{\tilde{\theta}-\bm \Znull \Bnull \bm\Znull^{\top}}^2
        \leq O_{R} \Paren{\frac{n^2}{d}\cdot k+\eta\cdot n^2} \,.
    \end{equation*}
    Since $\bm \Znull \Bnull \bm\Znull^{\top}$ is symmetric, it follows that
    \begin{equation*}
        \Normf{\tilde{\theta}^{\top}-\bm \Znull \Bnull \bm\Znull^{\top}}^2
        \leq O_{R} \Paren{\frac{n^2}{d}\cdot k+\eta\cdot n^2} \,.
    \end{equation*}
    Hence, by triangle inequality,
    \begin{equation*}
        \Normf{\tilde{\theta}-\tilde{\theta}^{\top}}^2
        \leq O_{R} \Paren{\frac{n^2}{d}\cdot k+\eta\cdot n^2} \,.
    \end{equation*}
    Since $\hat{Z}D\hat{Z}^\top$ is an averaging operator, which only reduces the Frobenius norm of a matrix, it follows that
    \begin{equation*}
        \Normf{\Paren{\tilde{\theta}-\tilde{\theta}^\top} \hat{Z}D\hat{Z}^\top}^2
        \leq O_{R} \Paren{\frac{n^2}{d}\cdot k+\eta\cdot n^2}\,.
    \end{equation*}
    Moreover, since the columns of $\tilde{\theta}^\top$ are the centers of the communities in $\hat{Z}$, it follows that $\tilde{\theta}^\top \hat{Z}D\hat{Z}^\top=\tilde{\theta}^{\top}$. Hence,
    \begin{equation*}
         \Normf{\tilde{\theta}\hat{Z}D\hat{Z}^\top-\tilde{\theta}^{\top}}^2
         \leq O_{R} \Paren{\frac{n^2}{d}\cdot k+\eta\cdot n^2}\,.
    \end{equation*}
    Applying triangle inequality again, we get 
    \begin{equation*}
        \Normf{\tilde{\theta}\hat{Z}D\hat{Z}^\top-\bm \Znull \Bnull \bm\Znull^{\top}}^2
        \leq O_{R} \Paren{\frac{n^2}{d}\cdot k+\eta\cdot n^2}\,.
    \end{equation*}
    Notice that, since $\tilde{\theta}^\top \hat{Z} D\hat{Z}^\top=\tilde{\theta}^{\top}$, we have $\hat{Z}D\hat{Z}^\top \tilde{\theta}=\tilde{\theta}$ by taking transpose on both sides. Therefore,
    \begin{equation*}
        \hat{Z}\hat{B}\hat{Z}^\top
        = \hat{Z} D\hat{Z}^\top \tilde{\theta} \hat{Z}D\hat{Z}^\top
        = \tilde{\theta} \hat{Z}D\hat{Z}^\top \,.
    \end{equation*}
    Thus,
    \begin{equation*}
\|\hat{Z}\hat{B}\hat{Z}^\top-\bm \Znull \Bnull \bm\Znull^{\top}\|_F^2= \Normf{ \tilde{\theta}\hat{Z}D\hat{Z}^\top-\bm \Znull \Bnull \bm\Znull^{\top}}^2\leq O_{R} \Paren{\frac{n^2}{d}\cdot k+\eta\cdot n^2} \,.
    \end{equation*}
\end{proof}

%% file: content/backgroundsos.tex
\section{Background of sum-of-squares hierarchy}
\label{section:backgroundsos}

\subsection{Sum-of-squares toolkit}\label{preliminaries:sos-toolkit}

Here, we introduce some sum-of-squares results that are used in our proofs. We start with a Cauchy-Schwarz inequality for pseudo-distributions.
\begin{fact}[Cauchy-Schwarz for pseudo-distributions ~\cite{barak2012hypercontractivity}]\label{fact:pd-cauchy-schwarz}
	Let $f,g$ be vector polynomials of degree at most $d$ in indeterminate $x\in \R^n$. Then, for any level $2d$ pseudo-distribution $D$,
	\begin{align*}
		\tilde{\E}_D\Brac{\iprod{f,g}}\leq \sqrt{\tilde{\E}_D\brac{\norm{f}^2} }\cdot \sqrt{\tilde{\E}_D\brac{\norm{g}^2} }\,.
	\end{align*}
\end{fact}
We will also repeatedly use the following two sos version of Cauchy-Schwarz inequality:
\begin{fact}[Sum-of-squares Cauchy-Schwarz I]
	\label{fact:sos-cauchy-schwarz}
	Let $x,y \in \R^d$ be indeterminites. Then,
	\[ 
	\sststile{4}{x,y} \Iprod{x, y}^2 \leq \Paren{\sum_i x_i^2} \Paren{\sum_i y_i^2} \,.
	\]
\end{fact}
and,
\begin{fact}[Sum-of-squares Cauchy-Schwarz II]
	\label{fact:simple-sos-version-of-cauchy-schwarz}
	Let $x,y \in \R^d$ be indeterminites. Then, for any $C > 0$,
	\[ 
	\sststile{4}{x,y} \Iprod{x, y} \leq \frac{C}{2} \Norm{x}^2 + \frac{1}{2C} \Norm{y}^2 \,,
	\]
	and,
	\[ 
		\sststile{4}{x,y} \Iprod{x, y} \geq - \frac{C}{2} \Norm{x}^2 - \frac{1}{2C} \Norm{y}^2 \,,
	\]
\end{fact}

We will use the following fact that shows how spectral certificates are captured within the SoS proof system.
\begin{fact}[Spectral Certificates] \label{fact:spectral-certificates}
	For any $m \times m$ matrix $A$, 
	\[
	\sststile{2}{u} \Set{ \iprod{u,Au} \leq \Norm{A} \Norm{u}_2^2}\mper
	\]
\end{fact}

The next fact establishes a certificate on the infinity-to-one  norm of a matrix.

We will use the notions of pseudo-covariance and conditional pseudo-distributions.

\begin{definition}[Pseudo-covariance]
	Let $\alpha, \beta$ be multi-indices over $[n]$. Let $d/2\geq \card{\alpha} + \card{\beta}$. Let $\tilde{\mu}$ be a level-$d$ pseudo-distribution in indeterminates  $x_1,\ldots, x_n\,.$ .
	Then we write
	\begin{align*}
		\text{Cov}_{{\tilde{\mu}}}{x^\alpha, x^\beta} = \tilde{\E}_{\tilde{\mu}} \brac{x^\alpha x^\beta}-\tilde{\E}_{\tilde{\mu}} \brac{x^\alpha}\tilde{\E}_{\tilde{\mu}} \brac{x^\beta}\,.
	\end{align*}
	Similarly, we define $\Var_{\tilde{\mu}}(x^{\alpha}) = \text{Cov}_{\tilde{\mu}}{x^{\alpha}, x^{\alpha}}\,.$
\end{definition}

\begin{definition}[Conditional pseudo-distribution]\label{definition:conditional-pseudo-distribution}
	Let $D$ be a degree-$d$ pseudo-distribution in indeterminates $x_1,\ldots, x_n\,.$ Let $t\geq 0$. Suppose $D$ satisfies $\Set{x_i^2=1\,, \forall i \in [n]}\,.$ Then for any $\alpha\in [n]^t$ such that $\tilde{\E} \brac{\frac{1+x^{\alpha}}{2}}> 0$ we may define the conditional pseudo-distribution of level $d-t$ as:
	\begin{align*}
		\tilde{\E}_D \brac{p(x) \,|\, x^\alpha=1} = \frac{\tilde{\E}_D \brac{p(x) \frac{1+x^\alpha}{2}}}{\tilde{\E}_D \brac{\frac{1+x^\alpha}{2}}}\,.
	\end{align*}
	Similarly, if $\tilde{\E} \brac{\frac{1+x^\alpha}{2}} < 1$, we may define the conditional pseudo-distribution of level $d-t$ as: 
	\begin{align*}
		\tilde{\E}_D \brac{p(x) \,|\, x^\alpha=-1} = \frac{\tilde{\E}_D \brac{p(x) \frac{1-x^\alpha}{2}}}{\tilde{\E}_D \brac{\frac{1-x^\alpha}{2}}}\,.
	\end{align*}
\end{definition}

It is straightforward to see that, after conditioning, the result is a valid pseudo-distribution of level $d-t$.
Notice also that, when $D$ is an actual distribution, then we simply recover the corresponding conditional distribution.

Last, we introduce the following crucial observation about pseudo-distributions.

\begin{lemma}[E.g. see ~\cite{schramm2022notes}]\label{csp:lemma:moment-matching-distribution}
	Let $D$ be a level $d$ pseudo-distribution over indeterminates $x_1,\ldots, x_n$ satisfying $\Set{x^2=1\,, \forall i \in [n]}$. 
	Then, for any $S\subseteq [n]$ with $\card{S}\leq d$, there exists a distribution $D'$ over $\Set{\pm 1}^n$ such that, for all multi-indices $\alpha$ over $S$, 
	\begin{align*}
		\tilde{\E}_{D}\brac{x^{\alpha}} = \tilde{\E}_{D'}\brac{x^{\alpha}}
	\end{align*}
\end{lemma}

In other words, \cref{csp:lemma:moment-matching-distribution} states that, for any level-$d$ pseudo-distribution $D$ over the hypercube and any subset $S$ of $d$ indeterminates, there exists an actual distribution $D'$ over the hypercube matching its first $d$ moments on $S$.
Notably, combining this results with \cref{definition:conditional-pseudo-distribution}, one gets that these low-degree moments of $D$ and $D'$ match even after conditioning.

\subsection{Sum-of-squares reweighing}

Here, we recap the definition and the results of sum-of-squares reweighing.

\begin{definition}[Sum-of-squares reweighing]
	Let $\mu$ be a pseudo-distribution of level $r$. Suppose $p$ is a sum-of-squares polynomial of degree $r' \leq r$ and $\tilde{\E}_{\mu} p > 0$, then $\mu'$ is a degree-$(r-r')$ reweighing of $\mu$ by $p$ if
	\begin{equation*}
		\mu'(x) = \frac{\mu(x)p(x)}{\tilde{\E}_{\mu} p} \,.
	\end{equation*}
\end{definition}

The key result of pseudo-distribution reweighing that we will need is the following subspace fixing reweighing lemma of~\cite{barak2017quantum}.

\begin{lemma}[Subspace fixing reweighing, Lemma 7.1 of~\cite{barak2017quantum}]
	\label{lem:bks-sos-reweighing}
	Let $\mu$ be a distribution over the unit ball $\set{x : \Norm{x} \leq 1}$ of $\R^d$ such that $\E_\mu \Norm{x}^2 \geq d^{-C}$ for some constant $C$. Then, $\mu$ has a degree $k= \frac{d}{\delta} \cdot (\log d)^{C'}$ reweighing $\mu'$ for some constant $C'$ such that
	\begin{equation*}
		\Norm{\E_{\mu'(x)}x}^2 \geq (1-\delta) \E_{\mu(x)} \Norm{x}^2 \,.
	\end{equation*}
	Further, the reweighing polynomial $p = \frac{\mu'}{\mu}$ can be found in time $2^{O(k)}$, has all coefficients upper bounded by $2^{O(k)}$ in the monomial basis, and satisfies $p(x) \leq k^{O(k)} \Norm{x}^k$. Moreover, the result holds even if $\mu$ is a pseudo-distribution of level at least $k+2$.
\end{lemma}

In our proof, we will use the following direct corollary of \cref{lem:bks-sos-reweighing}.

\begin{corollary}
	\label{cor:bks-sos-reweighing}
	Let $\mu$ be a distribution over the unit ball $\set{x : \Norm{x} \leq 1}$ of $\R^d$ such that $\E_\mu \Norm{x}^2 \geq d^{-C}$ for some constant $C$. Then, $\mu$ has a degree $k=\frac{d}{\delta} \cdot (\log d)^{C'}$ reweighing $\mu'$ for some constant $C'$ such that
	\begin{equation*}
		\Norm{\E_{\mu'(x)}x}^2 \geq (1-\delta) \E_{\mu' (x)} \Norm{x}^2 \,.
	\end{equation*}
\end{corollary}

\begin{proof}
	During the proof, we denote $\mu$ by $\mu_0$. Apply \cref{lem:bks-sos-reweighing} on $\mu_0$, we know that it has a degree $\frac{d}{\delta'} \cdot (\log d)^{c}$ reweighing $\mu_1$ such that
	\begin{equation*}
		\Norm{\E_{\mu_1 (x)}x}^2 \geq (1-\delta') \E_{\mu_0 (x)} \Norm{x}^2 \,.
	\end{equation*}
	If $\E_{\mu_0 (x)} \Norm{x}^2 < (1-\delta') \E_{\mu_1 (x)} \Norm{x}^2$, we apply \cref{lem:bks-sos-reweighing} on $\mu_1$ to obtain $\mu_2$ that satisfies
	\begin{equation*}
		\Norm{\E_{\mu_2 (x)}x}^2 \geq (1-\delta') \E_{\mu_1(x)} \Norm{x}^2 \,.
	\end{equation*}
	We can keep doing the reweighing until for some $t \geq 0$, we have $\E_{\mu_{t} (x)} \Norm{x}^2 \geq (1-\delta') \E_{\mu_{t+1} (x)} \Norm{x}^2$. Then, it follows that
	\begin{equation*}
		\Norm{\E_{\mu_{t+1} (x)}x}^2 \geq (1-\delta') \E_{\mu_{t}(x)} \Norm{x}^2 \geq (1-\delta')^2 \E_{\mu_{t+1}(x)} \Norm{x}^2 \geq (1-2\delta') \E_{\mu_{t+1}(x)} \Norm{x}^2 \,.
	\end{equation*}
	Now, we show that $t < \frac{C \log d}{\log \frac{1}{1-\delta'}}$. Since $\E_{\mu_{i-1} (x)} \Norm{x}^2 < (1-\delta') \E_{\mu_{i} (x)} \Norm{x}^2$ for all $i \leq t$, it follows that
	\begin{equation*}
		\E_{\mu_{0} (x)} \Norm{x}^2 < (1-\delta')^t \E_{\mu_{t} (x)} \Norm{x}^2\,.
	\end{equation*}
	Since $d^{-C} \leq \E_{\mu_{0} (x)} \Norm{x}^2$ and $\E_{\mu_{t} (x)} \Norm{x}^2 \leq 1$, we have
	\begin{equation*}
		d^{-C} < (1-\delta')^t\,.
	\end{equation*}
	Therefore, we get $t < \frac{C \log d}{\log \frac{1}{1-\delta'}}$ by taking $\log$ on both sides and rearranging terms.

	Now, we can take $\mu' = \mu_{t+1}$ and $\delta=2\delta'$, it follows that $\mu'$ is a reweighing of $\mu$ such that
	\begin{equation*}
		\Norm{\E_{\mu'(x)}x}^2 \geq (1-\delta) \E_{\mu' (x)} \Norm{x}^2 \,.
	\end{equation*}
	The degree of reweighing is $k = t \cdot \frac{d}{\delta'} \cdot (\log d)^{c} \leq \frac{C \log d }{\log \frac{1}{1-\delta'}} \cdot \frac{d}{\delta'} \cdot (\log d)^{c} = \frac{2 C d}{\delta \log \frac{1}{1-\delta/2}} \cdot (\log d)^{c+1} \leq \frac{d}{\delta} \cdot (\log d)^{C'}$ for some constant $C'$.
\end{proof}

\subsection{Sum-of-squares proofs for properties of community membership matrix}

We use the constraints
\begin{equation*}
	\cA_1(Z)\seteq
	\Set{Z\odot Z=Z, Z \Ind=\Ind, \transpose{Z} \Ind = \tfrac n k \Ind}\,.
\end{equation*}
for the set of $\{0,1\}^{n\times k}$ community membership matrices. 
Here we give low degree sum-of-squares proofs for some key properties of the community membership matrices which we will use in our sum-of-squares identifiablity proof.

\begin{lemma}
	\label{lem:sos-community-orthogonality}
	We have $\cA_1(Z)\sststile{4}{Z} Z(i,j)Z(i,j')=0$ for all distinct $j,j'\in [n]$. 
	Furthermore, 
	\begin{equation*}
		\cA_1(Z) \sststile{4}{Z} \frac{k}{n}Z^\top Z=\Id\,.
	\end{equation*}
\end{lemma}
\begin{proof}
	For distinct $j,j'\in [n]$, we have
	\begin{equation*}
		Z(i,j)Z(i,j')=Z(i,j)\Paren{1-\sum_{j''\neq j'} Z(i,j'')}=Z(i,j)-Z(i,j)^2-\sum_{j''\neq j'\,, j''\neq j}Z(i,j'')Z(i,j)
	\end{equation*}
 which allows us to conclude $Z(i,j)\sum_{j'\neq j}Z(i,j')=0$.
 Now since $Z(i,j)Z(i,j')\geq 0$, we have the sum-of-squares proof that $Z(i,j)Z(i,j')=0$. 

  Moreover, we have $\sum_{j\in [k]} Z_{i,j}^2=\sum_{j\in [k]} Z_{i,j}=\frac{n}{k}$.
  Therefore, we have \(\frac{k}{n}Z^\top Z=\Id\).
\end{proof}

\begin{lemma}
	\label{lem:block-matrix-sos-bound}
	Let $B \in [0, R]^{k \times k}$, we have
	\begin{equation*}
		\cA_1(Z)\sststile{4}{Z}
		0 \le ZB\transpose{Z} \le R \cdot \Ind_n \transpose{\Ind_n} \,.
	\end{equation*}
\end{lemma}

\begin{proof}
	For every $i, j \in [n]$, it follows that
	\begin{align*}
		\cA_1(Z) \sststile{4}{Z}
		\Paren{ZBZ^{\top}}_{i, j}
		& = \sum_{a, b \in [k]} Z_{i, a} B_{a, b} Z_{j, b} \\
		& = \sum_{a, b \in [k]} Z_{i, a}^2 B_{a, b} Z_{j, b}^2 \\
		& \geq 0 \,.
	  \end{align*}
	Since $\normm{B} \leq R$ and $\cA_1(Z) \sststile{2}{Z} Z_{i, a} = Z_{i, a}^2 \geq 0$ for all $i \in [n]$ and $a \in [k]$, it follows that, for every $i, j \in [n]$,
	\begin{align*}
	  \cA_1(Z) \sststile{4}{Z}
	  \Paren{ZBZ^{\top}}_{i, j}
	  & = \sum_{a, b \in [k]} Z_{i, a} B_{a, b} Z_{j, b} \\
	  & \leq R \cdot \sum_{a, b \in [k]} Z_{i, a} Z_{j, b} \\
	  & = R \cdot \Paren{\sum_{a \in [k]} Z_{i, a}} \cdot \Paren{\sum_{b \in [k]} Z_{j, b}} \\
	  & = R \,.
	\end{align*}
\end{proof}

\begin{lemma}[Sos \Holder inequality for \(\ell_\infty\) and \(\ell_1\)]
	\label{lem:sos-version-of-ell1-elli-Hoelder}
	Let $B \in [0, R]^{k \times k}$ and $M \in \R^{n \times n}$, we have
	\begin{equation*}
		\cA_1(Z)\sststile{4}{Z}
		- R \cdot \Normsum{M} \le \iprod{ZB\transpose{Z}, M} \le R \cdot \Normsum{M} \,.
	\end{equation*}
\end{lemma}

\begin{proof}
	By \cref{lem:block-matrix-sos-bound}, we have
	\begin{equation*}
		\cA_1(Z)\sststile{4}{Z}
		0 \le ZB\transpose{Z} \le R \cdot \Ind_n \transpose{\Ind_n} \,.
	\end{equation*}
	Hence,
	\begin{align*}
		\cA_1(Z)\sststile{4}{Z}
		\iprod{ZB\transpose{Z}, M}
		& = \sum_{i, j \in [n]} \Paren{ZB\transpose{Z}}_{i, j} M_{i, j} \\
		& \le \sum_{i, j \in [n]} \Paren{ZB\transpose{Z}}_{i, j} |M_{i, j}| \\
		& \le \sum_{i, j \in [n]} R \cdot |M_{i, j}| \\
		& = R \cdot \Normsum{M} \,.
	\end{align*}
	The other direction follows by taking $\Paren{ZB\transpose{Z}}_{i, j} M_{i, j} \geq - \Paren{ZB\transpose{Z}}_{i, j} |M_{i, j}|$ for all $i, j \in [n]$.
\end{proof}
    
\begin{lemma}
	\label{lem:sos-proof-properties-of-map-between-community-membership-and-doubly-stochastic}
	Let $B, \Bnull$ be $k$-by-$k$ matrices with non-negative entries.
	Let \(p(S;B,B_0)\) denote the following quadratic polynomial in \(S\) with coefficients depending on \(B,B_0\),
	\[
		p(S; B,B_0) \seteq \frac{1}{k^2}\sum_{a,a',b,b'\in [k]}\Paren{B(a,b)-\Bnull(a',b')}^2\cdot S(a,a')\cdot S(b,b')\,.
	\]
	Let $S(Z) \seteq \frac{k}{n} \transpose Z \Znull$ and \(\cA_{\mathrm{ds}}(S)\seteq \Set{ S \ge 0,~ S\Ind_k =\Ind_k,~\transpose S \Ind_k=\Ind_k }\), it follows that
	\[
		\cA_1(Z),\cA_1(\Znull) \sststile{O(1)}{Z,\Znull} \cA_{\mathrm{ds}}(S(Z)),p\bigparen{S(Z); B,B_0}=\tfrac 1{n^2}\normf{ZB\transpose{Z}-\Znull\Bnull\Znulltrans}^2\,.
	  \]
\end{lemma}

\begin{proof}
	Consider the entries of $S(Z)$, we have
	\begin{equation*}
		S(Z)_{i, j} = \sum_{t \in [n]} \frac{k}{n} Z_{t, i} (Z_0)_{t, j} \,.
	\end{equation*}
	Therefore, it follows that
	\begin{equation}
		\label{eq:appendix_sos_toolkit_1}
		\cA_1(Z),\cA_1(\Znull) \sststile{O(1)}{Z,\Znull}
		S(Z)_{i, j}
		= \sum_{t \in [n]} \frac{k}{n} Z_{t, i} (Z_0)_{t, j}
		= \sum_{t \in [n]} \frac{k}{n} Z_{t, i}^2 (Z_0)_{t, j}^2
		\geq 0 \,.
	\end{equation}
	Now, consider the row sum of $S(Z)$. For the $i$-th row, we have
	\begin{align*}
		\cA_1(Z),\cA_1(\Znull) \sststile{O(1)}{Z,\Znull}
		\sum_{j \in [k]} S(Z)_{i, j}
		= & \sum_{j \in [k]} \sum_{t \in [n]} \frac{k}{n} Z_{t, i} (Z_0)_{t, j} \\
		= & \sum_{t \in [n]} \frac{k}{n} Z_{t, i} \sum_{j \in [k]} (Z_0)_{t, j} \\
		= & \sum_{t \in [n]} \frac{k}{n} Z_{t, i} \\
		= & 1 \,.
	\end{align*}
	Hence, we get
	\begin{equation}
		\label{eq:appendix_sos_toolkit_2}
		\cA_1(Z),\cA_1(\Znull) \sststile{O(1)}{Z,\Znull}
		S\Ind_k =\Ind_k \,.
	\end{equation}
	Using a symmetric argument for the columns, we can also get
	\begin{equation}
		\label{eq:appendix_sos_toolkit_3}
		\cA_1(Z),\cA_1(\Znull) \sststile{O(1)}{Z,\Znull}
		\transpose S \Ind_k=\Ind_k \,.
	\end{equation}
	Combining \cref{eq:appendix_sos_toolkit_1}, \cref{eq:appendix_sos_toolkit_2} and \cref{eq:appendix_sos_toolkit_3}, we finish the proof for the first part of the lemma
	\begin{equation*}
		\cA_1(Z),\cA_1(\Znull) \sststile{O(1)}{Z,\Znull} \cA_{\mathrm{ds}}(S(Z)) \,.
	\end{equation*}
	Now, consider the polynomial $\tfrac 1{n^2}\normf{ZB\transpose{Z}-\Znull\Bnull\Znulltrans}^2$. Opening the squares and compare the degree of of each monomial, we can get the following identity
	\begin{equation*}
		\frac{1}{n^2} \normf{ZB\transpose Z - Z_0 B_0 \transpose {Z_0}}^2
		=\frac{1}{n^2}\sum_{a,a',b,b'\in [k]}\Paren{B_{a, b}-\Bnull_{a', b'}}^2 \iprod{Z_{(\cdot, a)}, (Z_0)_{(\cdot, a')}} \iprod{Z_{(\cdot, b)}, (Z_0)_{(\cdot, b')}}\,,
	\end{equation*}
	where $Z_{(\cdot, a)}$ denotes the $a$-th column of $Z$ (the same definition applies to $(Z_0)_{(\cdot, a')}$ etc.). Since $S(Z) \seteq \frac{k}{n} \transpose Z \Znull$, it follows that
	\begin{align*}
		\frac{1}{n^2} \normf{ZB\transpose Z - Z_0 B_0 \transpose {Z_0}}^2
		& = \frac{1}{k^2}\sum_{a,a',b,b'\in [k]}\Paren{B_{a, b}-\Bnull_{a', b'}}^2 S(Z)_{a, a'} S(Z)_{b, b'} \\
		& = p\bigparen{S(Z); B,B_0} \,.
	\end{align*}
	Thus, we have
	\[
		\cA_1(Z),\cA_1(\Znull) \sststile{O(1)}{Z,\Znull} \cA_{\mathrm{ds}}(S(Z)),p\bigparen{S(Z); B,B_0}=\tfrac 1{n^2}\normf{ZB\transpose{Z}-\Znull\Bnull\Znulltrans}^2\,.
	\]
\end{proof}

%% file: content/DistanceMetric.tex
\section{Distance metrics for graphon estimation}
\label{sec:distance-metrics}

In this section, we discuss the different distance metrics that we use to compare the estimated graphons with the true graphon. 
For $k$-stochastic block model, a natural distance metric is given by 
\begin{equation*}
    \delta_{p}(B, B_0) = \frac{1}{k^2} \min_{\pi:[k]\to [k]}\sum_{i=1}^{k} \sum_{j=1}^{k} \Paren{B(\pi(i),\pi(j)) - B_{0}(i,j)}^2,
\end{equation*}
where the minimum is taken over all permutations $\pi$ of $[k]$.
In this part, we discuss its relation to our distance metric $\dsdistance(B,B_0)$. 
Particularly, we prove that they are equivalent up to $\poly(k)$ factors.

\begin{theorem}
\label{thm:distance-metric}
Let $B$ and $B_0$ be two symmetric non-negative $k\times k$ matrices. 
Then $\delta_{p}(B, B_0)\leq \frac{1}{k^4} \dsdistance(B,B_0)$. 
\end{theorem}
\begin{proof}
Let $S\in \R^{k\times k}$ be the doubly stochastic matrix that minimizes the distance $\dsdistance(B,B_0)$, i.e we have 
\begin{equation*}
    \frac{1}{k^2} \sum_{a,a',b,b'} \Paren{B(a,b)-B_0(a',b')}^2 S(a,a')S(b,b') = \dsdistance(B,B_0). 
\end{equation*}
Then by the Birkhoff–von Neumann theorem, we have the decomposition
\begin{equation*}
    S = \sum_{t=1}^{k^2} \theta_t P_t,
\end{equation*}
where $P_t$ are $k\times k$ permutation matrices. 
Let $t_{\max}=\text{argmin}_{t\in [k]} \theta_t$, then $t_{\max}\geq \frac{1}{k^2}$.  
As result, we have
\begin{equation*}
    \sum_{a,a',b,b'} \Paren{B(a,b)-B_0(a',b')}^2 S(a,a')S(b,b')
    \geq \frac{1}{k^4} \sum_{a,a',b,b'} \Paren{B(a,b)-B_0(a',b')}^2 P_t(a,a')P_t(b,b')\,.
\end{equation*}
By the definition, we have \(\delta_{p}(B, B_0)\leq \frac{1}{k^2}\sum_{a,a',b,b'} \Paren{B(a,b)-B_0(a',b')}^2 P_t(a,a')P_t(b,b')\).
 Therefore $\delta_{p}(B, B_0)\leq \frac{1}{k^4} \dsdistance(B,B_0)$, which concludes the proof. 
\end{proof}

%% file: content/lowerbound.tex
\section{Lower bound for privacy cost}
\label{sec:lower_bound}

In this section, we prove \cref{thm:lb_main-theorem} of which \cref{thm:lower_bound} is a direct corollary by specifying the graph distribution to be the $(B,d,n)$-block model.

\begin{theorem}
  \label{thm:lb_main-theorem}
  Suppose there is an $\eps$-differentially private algorithm such that for any symmetric matrix $B\in[0,4]^{k\times k}$ with entries averaging to 1, on input an $n$-vertex random graph $\bm G$ sampled from some distribution defined by $B$, outputs $\hat{B}(\bm G)\in\R^{k\times k}$ satisfying
  \[
    \Pr\Paren{\delta_2\Paren{\hat{B}(\bm G), B}\leq \frac{1}{20}} \geq \frac{2}{3} \,.
  \]
  Then
  \[
    n \geq \Omega\Paren{\frac{k^2}{\eps}} \,.
  \]
\end{theorem}

To prove \cref{thm:lb_main-theorem}, we need some definitions and results in~\cite{mcmillan2018non}.
Let $\bbS_k$ denote the set of $k\times k$ permutation matrices.
We define the following more "combinatorial" metric $\lbdeltatwo$ that lower bounds the more "analytical" $\delta_2$ metric.

\begin{definition}
  Given two $k\times k$ matrices $A,B$, define
  \begin{equation}
    \lbdeltatwo(A,B) \seteq \min_{P_1,P_2\in\bbS_k} \frac{1}{k} \Normf{P_1AP_2-B} \,.
  \end{equation}
\end{definition}

\begin{lemma}[Lemma 1 in~\cite{mcmillan2018non}]
  \label{lem:lb-lem1_MS18}
  For every two $k\times k$ matrices $A,B$,
  \[
    \lbdeltatwo(A,B) \le \delta_2(A,B) \,.
  \]
\end{lemma}

McMillan and Smith~\cite[Lemma 4]{mcmillan2018non} apply the probabilistic method to show that, 
there exists a set $S$ of $2^{\Omega(k^2)}$ symmetric $k\times k$ binary matrices such that every pair of $B,B'\in S$ satisfies $\lbdeltatwo(B,B') \ge \Omega(1)$.
These matrices are not normalized (i.e. have entries averaging to 1), but the following fact shows that normalization can only help.

\begin{fact}
  \label{fact:lb-normalization_helps}
  Given $x,y \in \bits^n$, consider $\bar{x} \seteq n\cdot x/\normo{x}$ and $\bar{y} \seteq n\cdot y/\normo{y}$.
  Then 
  \[
    \Normt{\bar{x}-\bar{y}} \ge \Normt{x-y} \,.
  \]
\end{fact}

Then it is not difficult to adapt the proof of McMillan and Smith~\cite[Lemma 4]{mcmillan2018non} and show the following result.

\begin{lemma}
  \label{lem:lb_packing_lb}
  There exists a set $T$ of $2^{\Omega(k^2)}$ symmetric $k\times k$ matrices such that 
  \textit{(i)} every $B\in T$ satisfies $B\in[0,4]^{k\times k}$ and $\normsum{B}=k^2$; and
  \textit{(ii)} every pair of $B,B'\in T$ satisfies $\lbdeltatwo(B,B') \ge \frac{1}{6}$.
\end{lemma}
\begin{proof}
  Let $N\in\N$ and $\bm B_1, \dots, \bm B_N$ be iid uniformly random symmetric $k\times k$ binary matrices.
  By Chernoff bound, 
  \[
    \Pr\Paren{\Normsum{\bm B_1} \le \frac{k^2}{4}} \le \exp\Paren{-\frac{k^2}{20}} \,.
  \]
  From the proof of McMillan and Smith~\cite[Lemma 4]{mcmillan2018non},
  \[
    \Pr\Paren{\lbdeltatwo(\bm B_1, \bm B_2) \le \frac{1}{6}}
    \le \exp\Paren{-2\cdot \frac{\Paren{\frac{k^2}{6}-\frac{1}{2}\binom k 2}^2 }{\binom k 2} } (k!)^2 = 2^{-\Omega(k^2)} \,.
  \]
  Let $N=2^{ck^2}$ for some constant $c>0$.
  The probability that $\normsum{\bm B_i}\ge\frac{k^2}{4}$ for any $i\in[N]$ and, $\lbdeltatwo(\bm B_1, \bm B_2) \le \frac{1}{6}$ for any pair of $i,j\in[N]$ is at least $1-2^{2ck^2}2^{-\Omega(k^2)}-2^{ck^2}e^{-k^2/20}$.
  This probability can be made nonzero by choosing $c$ to be a sufficiently small absolute constant.
  Thus there exists a set $S$ of $2^{ck^2}$ symmetric $k\times k$ binary matrices such that 
  \textit{(i)} every $B\in S$ satisfies $\normsum{B}\ge k^2/4$; and
  \textit{(ii)} every pair of $B,B'\in S$ satisfies $\lbdeltatwo(B,B')\ge1/6$.
  Then we normalize matrices in $S$ by considering the following set 
  \[
    T \seteq \Set{ \frac{k^2\cdot B}{\Normsum{B}} : B\in S } \,.
  \]
  By property \textit{(i)} of $S$, every $B\in T$ satisfies $B\in[0,4]^{k\times k}$.
  By property \textit{(ii)} of $S$ and \cref{fact:lb-normalization_helps}, we have $\lbdeltatwo(B,B') \ge 1/6$ for every $B,B'\in T$.
\end{proof}

Now we are ready to prove \cref{thm:lb_main-theorem}.

\begin{proof}[Proof of \cref{thm:lb_main-theorem}]
  Consider an arbitrary algorithm satisfying the theorem's assumption.
  By \cref{lem:lb_packing_lb}, we can pick a set of $N=2^{\Omega(k^2)}$ symmetric $k\times k$ matrices $\set{B_1, \dots, B_N}$ in $[0,4]^{k\times k}$ such that 
  \textit{(i)} $\normsum{B_i}=k^2$ for any $i\in[N]$; and
  \textit{(ii)} $\lbdeltatwo\paren{B_i, B_j} \geq 1/6$ for every pair of $i,j\in[N]$.
  Since $\lbdeltatwo$ lower bounds $\delta_2$ by \cref{lem:lb-lem1_MS18}, we have $\delta_2\paren{B_i, B_j} \geq 1/6$ for every pair of $i,j\in[N]$.

  For each $i\in[N]$, define 
  \[
    Y_i \seteq \Set{B\in\R^{k\times k} \,:\, \delta_2\Paren{B,B_i} \leq 1/20} \,.
  \]
  Note $Y_1,\dots,Y_N$ are disjoint to each other.
  For $i\in[N]$, let $\mu_i$ denote an arbitrary graph distribution defined by $B_i$.
  As $Y_1,\dots,Y_N$ are pairwise disjoint,
  \[
    \sum_{i=1}^{N} \Pr_{\bm G \sim \mu_1} \Paren{\hat{B}(\bm G) \in Y_i} \leq 1 \,.
  \]
  By our utility assumption on the algorithm, for any $i\in[N]$ we have 
  \[
    \Pr_{\bm G \sim \mu_i} \Paren{\hat{B}(\bm G) \in Y_i} \geq \frac{2}{3} \,.
  \]
  Thus
  \[
    \sum_{i=2}^{N} \Pr_{\bm G \sim \mu_1} \Paren{\hat{B}(\bm G) \in Y_i} \leq \frac{1}{3} \,.
  \]
  By our privacy assumption on the algorithm, for any $n$-vertex graphs $G,H$ and any set $S\subseteq\R^{k\times k}$, we have
  \[
    \Pr \Paren{\hat{B}(G) \in S} \geq e^{-\eps n}\cdot \Pr \Paren{\hat{B}(H) \in S} \,,
  \]
  as the node distance between any two $n$-vertex graphs is at most $n$.
  Then for each $i\ge2$, 
  \[
    \Pr_{\bm G \sim \mu_1} \Paren{\hat{B}(\bm G) \in Y_i} 
    \geq e^{-\eps n}\cdot \Pr_{\bm G \sim \mu_i} \Paren{\hat{B}(\bm G) \in Y_i} 
    \geq e^{-\eps n}\cdot \frac{2}{3} \,.
  \]
  Putting things together,
  \[
    (N-1) e^{-\eps n}\cdot \frac{2}{3} \leq \frac{1}{3}
    \implies
    n \gtrsim \frac{k^2}{\eps} \,.
  \]
\end{proof}

%% file: content/theoremsfromBCS.tex
\section{Upper bound of sampling error and approximation error}
In this section, we list the theorems we use from~\cite{borgs2015private,borgs2018revealing}.

\subsection{Bounds on general graphon}
First we state a lemma on the upper bound of sampling error for general graphon.
\begin{lemma}[Sampling error upper bound, Corollary D.1 in~\cite{borgs2015private}]\label{lem:samplingupperbound}
     Let $W$ be arbitrary graphon with $\normi{W}\leq R$, and let $\bm Q_0$ be the edge connection probability matrix generated from the graphon. 
     Let 
     \[\epsilon_n(W)\coloneqq \min_\pi\Norm{W\Brac{\frac{\pi \Qnull \pi^\top}{\rho}}-W}_2^2\] 
     with minimum taken over $n\times n$ permutation matrices.
     Then we have 
     \begin{equation*}
          \E \epsilon_n^2(W,\bm Q_0)\leq 4\Paren{\epsilon_k^{(O)}(W)}^2+ R^2\cdot O\Paren{\sqrt{\frac{k}{n}}}\,,
     \end{equation*}
     where $\epsilon_k^{(O)}(W)=\min_{B} \norm{W[B]-W}_2$ with minimum taken over $k\times k$ symmetric matrices $B$.
\end{lemma}

Based on this result, we state a bound on the block approximation error of edge connection probability matrix $\Qnull$. 
\begin{lemma}[Bound on the block approximation error of edge connection probability]\label{lem:block_approximation_error}
     Consider the random graph generated from graphon $W$ with edge connection probability matrix given by $\bm \Qnull=\rho H_n(W)$.
     Let $\hat{\epsilon}_k^{(O)}(\Qnull)$ be the approximation error:
     \begin{equation*}
         \hat{\epsilon}_k^{(O)}(\Qnull)\coloneqq \frac{1}{\rho n}\min_{Z,B} \Normf{\rho\cdot ZBZ^\top-\Qnull}\,.
     \end{equation*}
     with minimum taken over balanced community membership matrix $Z\in \{0,1\}^{n\times k}$ and $B\in [0,R]^{k\times k}$.
     Then we have
     \begin{equation*}
          \Paren{\hat{\e}_k^{(O)}(\Qnull)}^2\leq O\Paren{\Paren{\e_k^{(O)}(W)}^2+R^2\sqrt{\frac{k}{n}}}\,.
      \end{equation*}
      where $\epsilon_k^{(O)}(W)=\min_{B} \norm{W[B]-W}_2$ with minimum taken over $k\times k$ symmetric matrices $B$
\end{lemma}
\begin{proof}
     The proof has already been provided in~\cite{borgs2015private}:
    Without loss of generality, we assume that $n/k$ is an integer.
     Let $B_{\min}\in \R^{k\times k}$ be the best $k$-block approximation of 
     $W$, i.e $\epsilon_k^{(O)}(W)=\norm{W[B_{\min}]-W}_2$.
     Then we have
     \begin{align*}
         \Paren{\hat{\e}_k^{(O)}(\Qnull)}^2 &\leq \frac{1}{n^2}\normf{ZB_{\min} Z^\top-\Qnull/\rho}^2\\
        & \leq \Paren{\delta_2(B_{\min},W)}^2+\min_\pi \Norm{W\Brac{\frac{\pi \Qnull \pi^\top}{\rho}}-W}_2^2 \tag{minimum taken over $n\times n$ permutations}\\
        & =\Paren{\epsilon_k^{(O)}(W)}^2+ \Paren{\epsilon_n(W)}^2\,.
     \end{align*}
\end{proof}

\subsection{Bounds on H\"older continous graphs}
When we assume that the graphon is $\alpha$-H\"older continous, we can further bound the approximation error $\epsilon_k^{(O)}(W)$.
\begin{lemma}[Lemma 16 in \cite{borgs2015private}]
    Suppose $W:[0,1]^2\to [0,1]$ is $\alpha$-H\"older continous for some $\alpha\in (0,1]$, which is to say, for some large universal constant $C$ 
     \begin{equation*}
          \Abs{W(x,y)-W(x',y')}\leq C \Paren{|x-x'|+|y-y'|}^\alpha\,.
     \end{equation*}
    Then we have
   \begin{equation*}
        \epsilon_k^{(O)}(W)\leq \normi{W-W_{\mathcal{P}_k}}\leq O\Paren{\Paren{\frac{2}{k}}^\alpha}\,.
   \end{equation*}
   where $W_{\mathcal{P}_k}$ is a balanced $k$-block graphon.
\end{lemma}

Moreover, the sampling error for H\"older continous graphs can be bounded by $O\Paren{n^{-\alpha/2}}$(\cite{borgs2015private},appendix E).
Therefore, applying \cref{thm:formalmaintheorem}, the error rate for $\alpha$-H\"older continous graphon is given by:
\begin{equation*}
\E\delta_2^2(\hat{W}(\bm{G}),W) \leq O_R\Paren{
 \frac{k}{\rho n}+\frac{k^2\log(n)}{n\epsilon}+n^{-\alpha}+\Paren{\Paren{\frac{2}{k}}^{2\alpha}}} \,.
\end{equation*}

%% file: content/concentration.tex
\section{Concentration inequalities from probability theory}

\subsection{Degree-pruned random graph}

For proving \cref{thm:formalmaintheorem} and \cref{thm:mainSBM} for very sparse graph($\rho\cdot n=\omega(1)$), we need to use the following classical result from random matrix theory. 
\begin{theorem}[Originally proved in~\cite{feige2005spectral}, restatement of theorem 6.7 in~\cite{liu2022minimax}]
    \label{theorem:pruned-spectral-norm}
    Suppose $\bm W\in \R^{n\times n}$ is a random symmetric matrix with zero on the diagonal whose entries above the diagonal are independent with the following distribution
    \begin{equation*}
      \bm W_{ij} =
        \begin{cases}
          1 - p_{ij} & \text{w.p. } p_{ij}\\
          - p_{ij} & \text{w.p. } 1- p_{ij}
        \end{cases}       
    \end{equation*}
    Let $\alpha$ be a quantity such that $p_{ij} \leq \frac{\alpha}{n}$ and $\bar{\bm W}$ be the matrix obtained from $\bm W$ by zeroing out all the rows and columns having more than $20 \alpha$ positive entries. Then with probability $1-\frac{1}{n^2}$, we have
    \begin{equation*}
        \Norm{\bar{\bm W}} \leq \chi \sqrt{\alpha}
    \end{equation*}
    for some constant $\chi$.
\end{theorem}

We also need the following bound for the number edges incident to high degree vertices in the random graph.
\begin{lemma}[Number of edges incident to high degree vertices]\label{lem:degree-counting}
  Suppose $M\in \R^{n\times n}$ is a random symmetric matrix with zero on the diagonal whose entries above the diagonal are independent with the following distribution
    \begin{equation*}
      W_{ij} =
        \begin{cases}
          1 - p_{ij} & \text{w.p. } p_{ij}\\
          - p_{ij} & \text{w.p. } 1- p_{ij}
        \end{cases}       
    \end{equation*}
    Let $\alpha$ be the quantity such that $p_{ij} \leq \frac{\alpha}{n}$.
    Let $S\subseteq [n]$ be the set of rows with more than $20 \alpha$ positive entries. 
    Then with probability \(1-\exp(-\alpha)\), the rows in $S$ contains at most \(O\Paren{\exp(-\alpha)\cdot \alpha \cdot n}\) positive entries.
\end{lemma}
\begin{proof}
  For any $t>20\alpha$, for each row $i\in [n]$, the probability that it contains more than $t$ positive entries is bounded by $\exp\Paren{-\frac{t}{3}}$. 
  We denote the set of rows with more than $t$ positive entries as $S_t$.
  Then by the linearity of expectation, we have $\E|S_t|=\exp\Paren{-\frac{t}{3}}\cdot n$.

 Now we note that 
 \[S=\bigcup_{t>20\alpha}S_t\] 
  For each fixed $t>20\alpha$, the expectation of total number of positive entries in rows belonging to $S_t$ is bounded by $\exp\Paren{-\frac{t}{3}}\cdot n\cdot t$.
  Therefore, summing over $t$, the expectation of the number of positive entries in the rows indexed by $S$ is upper bounded by
  \begin{equation*}
    \sum_{t=20\alpha}^{n} \exp\Paren{-\frac{t}{3}}\cdot n\cdot t
    \leq \exp(-2\alpha)\cdot n\cdot \alpha\,.
  \end{equation*}
  By Markov inequality, with probability at least $1-\exp(-\alpha)$, the number of positive entries in the rows indexed by $S$ is upper bounded by $\exp(-\alpha)\cdot n\cdot \alpha$.
\end{proof}

As a result, we have the following corollary: 
\begin{lemma}\label{lem:truncationscoredeviation}
  Let $\rho,\hat{\rho}\in \R$ such that $\rho\leq 10\hat{\rho}$.
  For $A\in {0,1}^{n\times n}$, $Z\in \R^{n\times k}$, and $B\in \R^{k\times k}$,
  let $f(B;Z;A)=-\frac{1}{\hat{\rho}^2} \normf{ZBZ^\top}^2+\frac{2}{\hat{\rho}}\iprod{ZBZ^\top,A}$. 
  Let $B_0\in [0,R]^{k\times k}$ and $Z_0\in \{0,1\}^{n\times k}$.
  Let $\bm A\in {0,1}^{n\times n}$ be a random matrix with i.i.d entries, and $A(i,j)=1$ with probability $Q_0(i,j)$ for some matrix $Q_0\in [0,R\cdot \rho]^{k\times k}$.
  Let $\bar{A}$ be the matrix obtained from $A$ by removing rows and columns  with more than $\rho\cdot n\cdot R$ positive entries.
  Then we have
  \begin{equation*}
    \Abs{f(B_0;Z_0;A)-f(B_0;Z_0;\bar{A})}\leq n^2\cdot\exp(-\frac{1}{2}\rho\cdot n\cdot R)
  \end{equation*}
\end{lemma}
\begin{proof}
  We have
  \begin{equation*}
    f(B;Z;A)=f(B;Z;\bar{A})+\frac{2}{\hat{\rho}}\iprod{ZBZ^\top,A-\bar{A}}\,.
  \end{equation*}
  By \cref{lem:degree-counting}, with high probability, $A-\bar{A}$ has less than \(O\Paren{\exp(-R\rho n)\cdot R\rho n^2}\) positive entries. 
  As result, with high probability, we have
  \begin{equation*}
    \frac{2}{\hat{\rho}}\iprod{ZBZ^\top,A-\bar{A}}\leq \frac{2R\rho n^2}{\hat{\rho}}\cdot \exp(-R\rho n)\leq O(n^2)\cdot \exp(-R\rho n)\,.
  \end{equation*}
\end{proof}

\subsection{Concentration bound for edge density}
Now we provide a lemma on the concentration bound of estimating edge density $\rho$.
We first prove a lemma on the convergence of empirical edge density in general random graphs(which is similar to lemma 12 in~\cite{borgs2015private}).
\begin{lemma}[Convergence of empirical edge density]\label{lem:convergence-edge-density}
  Given $Q\in [0,1]^{n\times n}$, a random graph $\bm G\in \cG_n$ is generated by independently connecting each pair of vertices $i,j\in [n]$ with probability $Q(i,j)$. 
   Let $\rho(\bm G)=\frac{2|E(\bm G)|}{n(n-1)}$ and $\rho(Q)=\frac{\norm{Q}_1}{n(n-1)}$. Then with probability at least $1-\exp\Paren{-\frac{n^2 t^2}{10\rho(Q)}}$,
  \begin{equation*}
    \Abs{\rho(\bm G)-\rho(Q)}\leq t  \,.
  \end{equation*}
\end{lemma}
\begin{proof}
Since $\E \Paren{A(i,j)-Q(i,j)}^2\leq Q(i,j)$. 
By Chernoff bound, we have
  \begin{align*}
    \Pr\Brac{\Abs{2|E(\bm G)|-\norm{Q}_1}\geq n^2 t}&= \Pr\Brac{\Abs{\sum_{i,j\in [n]}\Paren{A(i,j)-Q(i,j)}}
    \geq n^2 t}\\
    &\leq \exp\Paren{-\frac{n^4 t^2}{6\norm{Q}_1}}
  \end{align*}
  Since $\rho(\bm G)=\frac{2|E(\bm G)|}{n(n-1)}$ by definition, with probability at least $1-\exp\Paren{-\frac{n^2 t^2}{10\rho(Q)}}$,
  \begin{equation*}
      \Abs{\bm \rho(G)-\frac{\norm{Q}_1}{n^2}}\leq t \,.
  \end{equation*}
\end{proof}